%% file: manuscript.tex
\documentclass[physrev,twocolumn,english,aps,reprint, superscriptaddress,showpacs,longbibliography,showkeys, nofootinbib]{revtex4-2}

\usepackage{amsmath, amsthm, amssymb,bbm,mathrsfs,bm,braket,color,graphicx,comment,xcolor,nicematrix,dsfont,multirow, mathtools,tikz, adjustbox}
\usepackage[title]{appendix}
\usetikzlibrary{matrix,decorations.pathreplacing}
\usepackage[colorlinks,citecolor=blue,urlcolor=blue,linkcolor = blue]{hyperref}
\usepackage[mathscr]{euscript}

\usepackage{enumitem}

\usepackage{xargs}
\usepackage{tikz-cd}

\usepackage{xargs}
\usepackage{enumitem}
\usepackage{tikz-cd}
\usetikzlibrary{matrix,decorations.pathreplacing}
\usepackage[colorlinks,citecolor=blue,urlcolor=blue,linkcolor = blue]{hyperref}
\usepackage[mathscr]{euscript}
\usepackage{enumitem}

\usepackage{ifthen}

\DeclareMathOperator{\Cl}{\mathcal{C}\!\ell}
\DeclareMathOperator{\Pauli}{\mathcal{P}}

\DeclareMathOperator{\ord}{\mathrm{ord}}

\newcommand{\CNOT}[1][]{\mathrm{CNOT}_{#1}}
\DeclareMathOperator{\stab}{\mathrm{Stab}}
\DeclareMathOperator{\Aut}{\mathcal{A}\hspace*{-0.5\mu}\mathit{ut}}

\theoremstyle{plain}
\newtheorem{theorem}{Theorem}[section]
\newtheorem{proposition}[theorem]{Proposition}
\newtheorem{corollary}[theorem]{Corollary}
\theoremstyle{definition}
\newtheorem{definition}[theorem]{Definition}
\newtheorem{example}[theorem]{Example}
\theoremstyle{remark}
\newtheorem{remark}[theorem]{Remark}

\DeclareFontEncoding{LS1}{}{}
\DeclareFontSubstitution{LS1}{stix2}{m}{n}

\definecolor{zlog}{HTML}{C83C4F} 
\definecolor{xlog}{HTML}{6974b3} 
\definecolor{zphys}{HTML}{979A1B} 
\definecolor{xphys}{HTML}{1E9A87} 
\definecolor{clgrey}{HTML}{888888}

\newcommand{\violate}[1]{\overset{\scalebox{0.385}{$\times$}}{#1}}
\newcommand{\trivial}[1]{\overset{\scalebox{0.47}{$\:\bullet$}}{#1}}
\newcommand{\trivialpm}[1]{\overset{\scalebox{0.47}{$\bullet, \pm$}}{#1}}

\newcommand{\lname}[2]{\ell_{#1}^{\, \scriptscriptstyle #2}} 

\newcommand{\state}[1][\psi]{\vert #1\rangle}

\newcommand{\npq}{n} 
\newcommand{\ipq}{j} 
\newcommand{\ipqhat}{\hat{\jmath}} 
\newcommand{\ipqcheck}{\check{\jmath}}  
\newcommand{\itq}{t} 
\newcommand{\itqhat}{\hat{\itq}}  
\newcommand{\icq}{c} 
\newcommand{\icqhat}{\hat{\icq}}  
\newcommand{\ep}{\kappa} 
\newcommand{\itime}{\lambda} 
\newcommand{\itimeend}{\Lambda} 
\newcommand{\idmatrix}{\mathbb{I}}
\newcommand{\reprpauli}[1][\npq]{\mathbb{Z}_{4}\times\mathbb{F}_2^{#1}\times\mathbb{F}_2^{#1}}

\newcommand{\logicalflowp}[1][]{\logicalpaulivec{\ep_{#1}}{\bmlogicalxi[#1]}{\bmlogicalzeta[#1]}}
\newcommand{\physicalflowp}[1][]{\physicalpaulivec{\eta_{#1}}{\bmphysicalxi[#1]}{\bmphysicalzeta[#1]}}
\newcommand{\flowp}[1][]{\paulivec{\ep_{#1}}{\bm{\xi_{#1}}}{\bm{\zeta_{#1}}}}
\newcommand{\logicalpaulivec}[3]{%
    \left({#1} \vert {\color{xlog}{#2}} \vert {\color{zlog}{#3}} \right)%
}
\newcommand{\physicalpaulivec}[3]{%
    \left({#1} \vert {\color{xphys}{#2}} \vert {\color{zphys}{#3}}\right)%
}
\newcommand{\paulivec}[3]{%
    \left({#1} \vert {#2} \vert {#3}\right)%
}
\newcommand*{\logicalxi}[1][]{{\color{xlog}\xi_{#1}}}
\newcommand*{\logicalzeta}[1][]{{\color{zlog}\zeta_{#1}}}
\newcommand*{\bmlogicalxi}[1][]{{\color{xlog}\bm{\xi_{#1}}}}
\newcommand*{\bmlogicalzeta}[1][]{{\color{zlog}\bm{\zeta_{#1}}}}
\newcommand*{\physicalxi}[1][]{{\color{xphys}x_{#1}}}
\newcommand*{\physicalzeta}[1][]{{\color{zphys}z_{#1}}}
\newcommand*{\bmphysicalxi}[1][]{{\color{xphys}\bm{x_{#1}}}}
\newcommand*{\bmphysicalzeta}[1][]{{\color{zphys}\bm{z_{#1}}}}

\newcommand{\logicalpauliop}[3]{%
     {#1}{\color{xlog}\foreach \x in {#2}{\bar{X}_{\x}}}{\color{zlog}\foreach \z in {#3}{\bar{Z}_{\z}}}
}
\newcommand{\pauliop}[3]{
     {#1}{\foreach \x in {#2}{X_{\x}}}{\foreach \z in {#3}{Z_{\z}}}
}
\newcommand{\physicalpauliop}[3]{
     {#1}{\color{xphys}\foreach \x in {#2}{X_{\x}}}{\color{zphys}\foreach \z in {#3}{Z_{\z}}}
}
\newcommand{\blackphysicalpauliop}[3]{
     {#1}{\foreach \x in {#2}{X_{\x}}}{\foreach \z in {#3}{Z_{\z}}}
}
\newcommand{\logicalp}{\bar{P}}
\newcommand{\labelname}[1][]{\ell_{#1}} 

\newcommandx*{\rowcolrestr}[3][1,3]{{}_{\scriptstyle #1}#2_{\scriptstyle #3}}

\newcommand{\msone}{\mspace{0.5mu}}
\newcommand{\mstwo}{\mspace{0.2mu}}
\newcommand{\yourfavouriteleftparanthesis}{\langle}
\newcommand{\yourfavouriterightparanthesis}{\rangle}

\newcommandx{\flowlabel}[3]{%
#1\ifthenelse{\equal{#2}{} \AND \equal{#3}{} \OR \equal{#1}{}}{}{\msone}%
\ifthenelse{\equal{#2}{}}{}{\textcolor{xlog}{\yourfavouriteleftparanthesis\mstwo#2\mstwo \yourfavouriterightparanthesis}}%
\ifthenelse{\equal{#1}{} \AND \equal{#2}{} \OR \equal{#3}{}}{}{\msone}{\textcolor{zlog}{#3}}%
}

\def\thickhline{\noalign{\hrule height.8pt}}

\newcommandx*{\labelbox}[6]{\mbox{\footnotesize$\begin{array}{c}
    \flowlabel{#1}{#2}{#3}\\\thickhline
    \flowlabel{#4}{#5}{#6}
\end{array}$}}

\newcommandx*{\mathbox}[2]{\mbox{\footnotesize$\begin{array}{c}
    #1\\\thickhline
    #2
\end{array}$}}
\newcommandx*{\cnotlabelbox}[2]{\labelbox{}{#1}{}{}{}{#2}}

\newcommandx{\smallflowlabel}[3]{%
{\scriptscriptstyle #1}\ifthenelse{\equal{#2}{} \AND \equal{#3}{} \OR \equal{#1}{}}{}{\scriptscriptstyle\msone}%
\ifthenelse{\equal{#2}{}}{}{\scriptscriptstyle\textcolor{xlog}{\yourfavouriteleftparanthesis\mstwo#2\mstwo \yourfavouriterightparanthesis}}%
\ifthenelse{\equal{#1}{} \AND \equal{#2}{} \OR \equal{#3}{}}{}{\msone}{\textcolor{zlog}{\scriptscriptstyle #3}}%
}

\NiceMatrixOptions
{
    custom-line =
    {
        letter = I ,
        tikz = dotted ,
        total-width = \pgflinewidth
    }
}

\usepackage{verbatim}
\immediate\write18{texcount -inc -sum main_prl.tex > /tmp/wordcount.tex}

\begin{document}

\author{Berend~Klaver}
\affiliation{Institute for Theoretical Physics, University of Innsbruck, A-6020 Innsbruck, Austria}
\affiliation{Parity Quantum Computing GmbH, A-6020 Innsbruck, Austria}

\author{Katharina~Ludwig}
\affiliation{Parity Quantum Computing Germany GmbH, Schauenburgerstraße 6, 20095 Hamburg, Germany}

\author{Anette~Messinger}
\affiliation{Parity Quantum Computing GmbH, A-6020 Innsbruck, Austria}

\author{Stefan~M.A.~Rombouts}
\affiliation{Parity Quantum Computing Germany GmbH, Schauenburgerstraße 6, 20095 Hamburg, Germany}

\author{Michael~Fellner}
\affiliation{Institute for Theoretical Physics, University of Innsbruck, A-6020 Innsbruck, Austria}
\affiliation{Parity Quantum Computing GmbH, A-6020 Innsbruck, Austria}

\author{Kilian~Ender}
\affiliation{Institute for Theoretical Physics, University of Innsbruck, A-6020 Innsbruck, Austria}

\author{Wolfgang~Lechner}
\affiliation{Institute for Theoretical Physics, University of Innsbruck, A-6020 Innsbruck, Austria}
\affiliation{Parity Quantum Computing GmbH, A-6020 Innsbruck, Austria}
\affiliation{Parity Quantum Computing Germany GmbH, Schauenburgerstraße 6, 20095 Hamburg, Germany}

\begin{abstract}
We propose the Parity Flow formalism, a method for tracking the information flow in quantum circuits.
This method adds labels to quantum circuit diagrams such that the action of Clifford gates can be understood as a recoding of quantum information.
The action of non-Clifford gates in the encoded space can be directly deduced from those labels without backtracking.
An application of flow tracking is to design resource-efficient quantum circuits by changing any present encoding via a simple set of rules.
Finally, the Parity Flow formalism can be used in combination with stabilizer codes to further reduce quantum circuit depth and to reveal additional operations that can be implemented in parallel.

\end{abstract}

\date{\today}

\title{The Parity Flow Formalism:\\Tracking Quantum Information Throughout Computation}

\maketitle
\section{Introduction}
The field of quantum computing is currently maturing from academic research to industrial application.
This shift, driven by the prospect of solving classically intractable computational problems, still needs to bridge the gap between the available quantum resources \cite{Cirac1995, Lloyd1996, Raussendorf2001, Koch2007, Haeffner2008, DiCarlo2009, Watson2018, arute2019, Bluvstein_2023} and the resources required to faithfully execute useful quantum algorithms \cite{Deutsch1992, Shor1997, Bernstein1997, Terhal2015, Babbush2018}.
One approach for bridging this gap is to consider how an initially abstract algorithm is physically implemented on given quantum hardware, which typically involves a decomposition into a set of feasible quantum operations~\cite{Reck1994, DiVincenzo1995, Barenco1995_Elementary}.
The decomposition of a quantum algorithm is not unique, resulting in a wide variety of possible implementations, each with its own resource requirements.
To assist in finding efficient implementations for a given hardware and to gain insight into the workings of a quantum circuit, it is possible to classically track how quantum information passes through a circuit, either on the level of individual qubits~\cite{Kivlichan2018, cowtan2019qubitrouting, ogorman2019, nannicini2021optimalqubitassignmentrouting}, or by tracking components of that information \cite{paykin2023, Kissinger2019, schmitz2023graphoptimizationperspectivelowdepth, Meijer_van_de_Griend_2023}.

In this work, we propose the \textit{Parity Flow formalism}, a method that efficiently tracks the components of the information of each individual qubit by means of generalized parity labels.
In this formalism, the effect of Clifford gates is interpreted as a recoding of information by updating these labels with a simple set of rules.
Importantly, \textit{flow tracking} remains classically efficient when the circuit is extended, and is thus suitable to understand and design quantum circuits and algorithms.
We compare flow tracking with other methods that track Clifford circuits~\cite{AaronsonGottesman2004, gottesman1998heisenberg, de_Beaudrap_2013,Gidney2021stimfaststabilizer, Winderl2023}, categorizing them either as co- or contravariant Clifford tracking, and include a method to simultaneously do both efficiently within this formalism.
As an application, flow tracking allows one to utilize auxiliary qubits for the reduction of the circuit depth and the identification of easily implementable logical multi-qubit operations in stabilizer codes. 
The Parity Flow formalism can be used equivalently as either labels directly written on the quantum circuit or as the \textit{(combined) flow tableau}. In the main text we describe the formalism in its label form and show its usefulness in circuit design, while the appendices give the rigorous definitions and proofs mainly using the tableau form.

\section{Flow tracking of quantum circuits}
We consider a quantum circuit as a dynamic code, where each tracked Clifford gate changes the encoding of the quantum state.
The role of each qubit within the encoding is tracked via labels, such that each label states the logical effect on the encoded quantum state when operating on this qubit with a non-Clifford gate.
To illustrate our approach, we start with the simple case of tracking circuits which consist only of $\mathrm{CNOT}$ gates. 
In this case, the introduced labels have a direct correspondence to the parity maps described in Refs.~\cite{Kissinger2019, Nash2020}. 
\begin{figure*}[ht!]
    \includegraphics[width=0.95\textwidth]{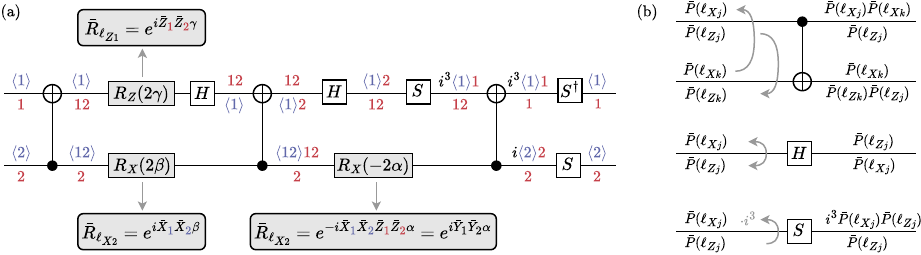}
   \caption{
   (a) Labeled quantum circuit with tracked single qubit Clifford gates (in white) and $\mathrm{CNOT}$ gates, and untracked interleaved Pauli rotations defined as $R_P(\alpha)=\exp(iP\alpha/2)$ (grey gates).
   The labels allow us to read off the logical effect of the physically applied rotations  (see the grey boxes). 
   The circuit demonstrates the logical application of the Heisenberg model~\cite{pengetal2022heisenberg}  $e^{i\alpha\bar{Y}_{1}\bar{Y}_{2}}e^{i\beta\bar{X}_{1}\bar{X}_{2}}e^{i\gamma\bar{Z}_{1}\bar{Z}_{2}}$ as $C_{3}e^{-i\alpha X_{2}}C_{2}e^{i\beta X_{2}}e^{i\gamma Z_{1}}C_{1}$ [cf.\ Eq.~\eqref{eq:alternating_unitary}]. 
   Note that the label transformation of the $S^\dagger$ gate can be obtained from that of three sequential $S$ gates.
   (b) Update rules for general labels $\labelname[X\small{j}]$ and $\labelname[Z\small{j}]$ under the action of a set of generators of the Clifford group, which are defined via their corresponding operators [see Eq.~\eqref{eq:update_2}].
   }
   \label{fig:full_circuit}
\end{figure*}
The effect of a $\mathrm{CNOT}$ gate with control qubit $1$ and target qubit $2$ on Pauli-$Z$ eigenstates $\ket{a_{1}}$ and $\ket{a_{2}}$, respectively, is described by
\begin{equation}\label{eq:CNOT_on_Z_state}
    \CNOT[1, 2] \ket{a_{1}} \ket{a_{2}} = \ket{a_1} \ket{a_{1}\oplus a_{2}},
\end{equation}
where $a_1, a_2 \in \{0,1\}$ and $\oplus$ denotes addition modulo~2. 
After application of the $\mathrm{CNOT}$ gate, no information is lost, but it is encoded with a linear map, i.e., the two physical qubits now hold the logical information of the state $a_1$ and the parity of $a_1$ and $a_2$.
Thus, for any quantum state, applying a $Z_2$-rotation after the $\CNOT[1, 2]$-gate has the same effect as a $Z_1 Z_2$-rotation would have had before this gate.
This effect can be tracked by assigning a $Z$-label to qubit~2, which contains the indices of the $a_i$ appearing in its state [see Eq.~\eqref{eq:CNOT_on_Z_state}]. 
For this example, these indices would be $\{1,2\}$. 
The transformation on Pauli-$X$ eigenstates (${\ket{0}_x=\ket+}$ and ${\ket{1}_x=\ket-}$) is calculated similarly as $\CNOT[1,2] \ket{a_1}_x \ket{a_2}_x = \ket{a_1 \oplus a_2}_x \ket{a_2}_x$ and analogous labels can be defined for the $X$ operators of each qubit (here $\{1,2\}$ as the $X$-label of qubit 1).

When tracking only $\mathrm{CNOT}$ gates, the $Z$-labels correspond directly to the parity labels introduced in Ref.~\cite{klaver2024}, and the $X$-labels are uniquely defined by the $Z$-labels (see Appendix Cor.~\ref{cor:cnot_only}).
In this case, tracking the evolution of Pauli eigenstates is equivalent to tracking the (unitary) transformation of the Pauli operators themselves.
As an example, the transformations ${\text{CNOT}^{\dagger}_{1,{2}}}Z_{2}{\text{CNOT}^{}_{1,{2}}}=\bar{Z}_{1}\bar{Z}_{2}$ and ${\text{CNOT}^{\dagger}_{1,{2}}}X_{1}{\text{CNOT}^{}_{1,{2}}}=\bar{X}_{1}\bar{X}_{2}$ are also described by the $X$- and $Z$-labels introduced above.
Note that throughout this manuscript, operators which are physically applied to an encoded state are denoted without a bar, while their logical effects are denoted with a bar.
By interpreting the labels on an operator level, it is easy to extend the label-tracking to arbitrary Clifford circuits.

In general, a label $\labelname[P]$ assigns a logical Pauli operator $\bar{P}$ to a physical Pauli operator $P$.
When appropriately updating the labels after a Clifford circuit $C$, they describe the logical effect $\bar{P}$ of applying a Pauli operator $P$ after this circuit, such that the label describes the conjugation ${P \mapsto \bar P(\labelname[P]^{\,\scriptscriptstyle C}) = C^\dagger P C}$. 
Note that we usually omit the superscript that indicates the preceding Clifford circuit on the label, except when strictly necessary.
In order to efficiently and iteratively track the complete Pauli group $\mathcal{P}$ on $n$ qubits, one needs to choose an independent set of $2n+1$ generators~\cite{gottesman1998heisenberg}, and track the labels for those.
In our case, we choose to track the labels of all single-qubit $X$- and $Z$-operators (and $i\mathbb{I}$, on which any Clifford transformation acts trivially), and any appearing Pauli operator will be described in terms of these generators ${g_m\in \{i\mathbb{I}, X_j, Z_j \,|\, { j =1,\dots, n}\}}$. 
Other generating sets also form a valid choice for expressing and tracking Pauli operators, for example all single-qubit $X$, $Y$, and $Z$-operators \cite{AaronsonGottesman2004}.
However, our choice of independent generators, also used for example in \cite{dehaenedemoor2003, GossetGrierKerznerSchaeffer2024}, allows us to uniquely write each Pauli operator as a vector such that the vector presentation is isomorphic to the Pauli group (see Appendix Prop.~\ref{prop:group_isomorphism}).
With this choice of generators, a label $\labelname[P]$ consists of a phase $i^\kappa$ with ${\kappa\in\{0,1,2,3\}}$ and two sets of indices for the appearing $X$- and $Z$-operators.
We write a label as $\labelname[P]=\flowlabel{i^\kappa}{\bar{X} \text{-indices}}{\bar{Z} \text{-indices}}=\flowlabel{i^\kappa}{j_1 j_2 ...}{k_1 k_2 ...}$, and the logical Pauli operator $\bar{P}$ corresponding to label $\ell_{P}$ is
\begin{equation}\label{eq:standard_form_pauli_main}
    \bar{P}({\ell_P}) = i^\kappa \, \prod_{\langle j\rangle\in \labelname[P]} \bar{X}_j \prod_{k \in \labelname[P]} \bar{Z}_k.
\end{equation}
Note that when iterating over a label we use angled brackets around the $\bar X_{j}$-indices, and a bare index for the $\bar Z_{k}$-indices, such that the colors only act as a visual aid.
Furthermore, we choose to write all Pauli operators in the standard order of $X$ operators placed left of the $Z$ operators, such that there is a unique label for each Pauli operator. 
As a result, given a Pauli operator $\bar P$ expressed in the tracked generators in standard order, it is straightforward to read off the corresponding unique label.

In a quantum circuit, we write the labels $\labelname[X_j]$ and $\labelname[Z_j]$ above and below the qubit line $j$, respectively. 
Equivalently, one can also write the operators $\bar{P}(\labelname[X_j])$ and $\bar{P}(\labelname[Z_j])$ instead of the label, which can be more practical depending on the use-case.
Since the conjugation captured by a label directly extends to complex exponentials of Pauli operators, the labels allow us to read off the logical effect of an applied physical rotation at a given moment in the circuit (see Fig.~\ref{fig:full_circuit}(a) for an example).
The tracking begins with an identity mapping between the physical and logical operators, such that
${\bar{P}(\labelname[g_m]^{\,{\text{\tiny id}}})=g_{m}}$, or equivalently ${\lname{X_j}{\text{id}}=\langle j\rangle}$ and ${\lname{Z_k}{\text{id}}= k}$.
As mentioned above, these labels change when a Clifford gate is applied and tracked. 
Given labels $\lname{g_{j}}{C_1}$ after an arbitrary Clifford circuit $C_{1}$, applying a second Clifford circuit $C_{2}$ transforms $\lname{g_{j}}{C_{1}}$ into $\lname{g_j}{'}:=\lname{g_j}{C_{2}C_{1}}$. 
The new label is calculated as follows:
We first decompose the corresponding Clifford conjugation of $g_{j}$ by $C_2$ into a product of the chosen generators as
\begin{equation}
\label{eq:update_1}
    C_2^{\dagger} g_{j} C_2=i^{\kappa'}\prod_{g_m\in G} g_m,
\end{equation}
resulting in a phase $i^{\kappa'}$ and a subset $G$ of the generators, which can be computed via Pauli decomposition~\cite{Hantzko_2024}.
We then calculate the Pauli operator described by $\lname{g_j}{'}$ as
\begin{equation}
\label{eq:update_2}
    \bar{P}(\lname{g_j}{'}) = i^{\kappa'}\prod_{g_m\in G} \bar{P}(\lname{g_m}{C_1}),
\end{equation}
from which one obtains $\lname{g_j}{'}$ after reordering $\bar{P}(\lname{g_j}{'})$ into standard order.
Importantly, the update rule for the label of each generator is calculated via Eq.~\eqref{eq:update_1} and depends only on the applied gate $C_{2}$, and not on the preceding labels $\lname{g_j}{C_{1}}$.
Therefore, it is easy to apply the update rules iteratively, as the transformation of labels defined by $C_{2}$ is applied to the labels resulting from the previously tracked circuit $C_1$ (see Appendix Prop.~\ref{prop:concatenation_flow}).
Applying these rules to the Hadamard and $S$ gate, we find that the non-trivial transformations are 
\begin{align*}
H^{\dagger}ZH&=\bar{X},\\
H^{\dagger}XH&=\bar{Z}\\
\text{and }S^{\dagger}XS&=-\bar{Y}=i^{3}\bar{X}\bar{Z},
\end{align*}
resulting in the update rules for labels visualized in Fig.~\ref{fig:full_circuit}(b).

By decomposing a unitary $U$ into Pauli rotations $\bar R_j(\theta) = e^{i\frac{\theta}{2}\bar P_j}$, it can be written as an alternating sequence of Clifford operations $C_j$ and single-qubit Pauli rotations $R_j$ as 
\begin{equation}
\label{eq:alternating_unitary}
    U = \bar{R}_d \cdots\bar R_2 \bar R_1 = C_{d+1}R_d\cdots C_{3}R_{2} C_{2}R_{1}C_{1}.
\end{equation}
As a result, the labels allow us to interpret and efficiently design quantum circuits when written in the form of Eq.~\eqref{eq:alternating_unitary}, even when perpetually appending Pauli rotations or Clifford gates. 
Additionally, products of physical operators have the logical effect of the product of the Pauli operators represented by the corresponding labels.
Thus, physical multi-qubit rotations can also be incorporated.
An important strength of such a circuit representation is that it circumvents the requirement to synthesize each multi-body rotation operator separately by fully conjugating a single-body operator with a Clifford circuit $\tilde{C}_j$ (i.e., encoding, rotating and decoding~\cite{phase_gadgets_cowtan}). 
Rather, in a given encoding, one can continue from there to directly implement the next operator via recoding, as the subsequent labels can be derived from the latest labels:
\begin{equation}
\label{eq:en_de_coding}
  \cdots C_{3} R_2 C_{2}R_1 C_{1} = \cdots\underbrace{(\tilde{C}^{\dagger}_{2} e^{i\beta P_2} \tilde{C}_{2})}_{\bar{R}_2(2\beta)}(\tilde{C}^{\dagger}_{1}R_1\tilde{C}_{1}).
\end{equation}
In this formulation, $C_{j}$ denotes the Clifford circuit required to change between two encodings $j-1$ and $j$, while $\tilde{C}_{j}=C_{j}\cdots C_{1}$ are the Clifford circuits used for the full conjugation. 
In many cases, avoiding the detour of completely encoding and decoding the quantum state, and instead embracing the code changes, removes an avoidable gate overhead \cite{schmitz2023graphoptimizationperspectivelowdepth, domínguez2024runtimereductionlinearquantum}.
By using the full Clifford group for code changes, any logical Pauli string can be constructed on one of the labels of a single physical-qubit, starting from any other previous encoding.
Besides the efficient synthesis of multi-qubit Pauli rotations, which have already been proven useful in various applications like QFT and QAOA \cite{klaver2024, dreier2025connectivityawaresynthesisquantumalgorithms}, we expect many other practical uses of this tracking formalism. 
The labels can, for example, be used to debug quantum circuits or show circuit identities: If two such circuits transform initial labels to the same output labels and result in the same sequence of logical Pauli rotations, they implement the same unitary. 

\section{Covariant and contravariant Clifford tracking}
\label{sec:covariant_and_contravariant_clifford_tracking}
To see how flow tracking fits in the context of Clifford tracking, we first consider the workings of the original Clifford tableau~\cite{AaronsonGottesman2004}, which is used for \textit{covariant} Clifford tracking of a quantum state.
Consider a quantum state described by the density matrix $\rho$ written in the Pauli decomposition as
\begin{equation}
    \label{eq:density_matrix}
    \rho =\sum_{\bar{P}} a_{\bar{P}}\bar{P},
\end{equation}
where the coefficients $a_{\bar{P}}$ are real numbers indicating the state components and the sum is over the Pauli operator basis~\cite{koska2024treeapproachpaulidecompositionalgorithm}. 
Then a Clifford circuit $C_{1}$ acting on this quantum state performs a covariant transformation of the basis according to
\begin{equation}
C_{1}\rho C_{1}^{\dagger}=\sum_{\bar{P}}a_{\bar{P}}C_{1}\bar{P}C^{\dagger}_{1}.
\end{equation}
This basis transformation under Clifford circuits forms a bijective map from the Pauli group~$\mathcal{P}$ to itself, and for the covariant tracking describes the mapping of logical operators $\bar{P}$ to physical operators $P$.
One could in principle compute this transformation for every element of the exponentially large Pauli group, which would also reveal the inverse map.
Alternatively, one can restrict the tracking to a linear-sized set of generators~$\bar{g}_{m}$ of the Pauli group~$\mathcal{P}$.
In this case, one can still compute how any logical Pauli operator $\bar{P}=\textstyle\prod_{j}\bar{g}_{j}$ maps to a physical Pauli operator via its decomposition into generators as $C_{1}\bar{P}C^{\dagger}_{1}=\textstyle\prod_{j}C_{1}\bar{g}_{j}C_{1}^{\dagger}=P$.
While the restriction to only generators makes the tracking computationally efficient, it also obscures the information required for the inverse map, i.e., mapping physical operators to logical ones.

On the other hand, the \textit{contravariant} Clifford tracking of $C_{1}\rho C_{1}^{\dagger}$ is given by the inverse transformation on the state components as 
\begin{equation}
C_{1}\rho C_{1}^{\dagger}=\sum_{P}a_{C^{\dagger}_{1}PC_{1}} P.
\end{equation}
Consequently, when only tracking generators, the ease of mapping logical operators to physical operators and the difficulty of mapping physical operators to logical operators is reversed with respect to the covariant case.

Besides the difference in use case for the two types of tracking, the conjugation method for iteratively tracking concatenations of circuits also differs.
That is, when appending a second transformation $C_{2}$, we find 
\begin{align}
C_{2}C_{1}\rho C_{1}^{\dagger}C_{2}^{\dagger}&=\sum_{\bar{P} }a_{\bar{P}}C_{2}C_{1}\bar{P}C^{\dagger}_{1}C_{2}^{\dagger}\label{eqn:covariant_tracking}\\
&=\sum_{P}a_{C^{\dagger}_{1}C^{\dagger}_{2}PC_{2}C_{1}}P,\label{eqn:contravariant_tracking}\end{align}
for the covariant and contravariant transformation, respectively. 
Importantly, the covariant Clifford tracking computes the basis change $C_{2}C_{1}\bar{P}C^{\dagger}_{1}C_{2}^{\dagger}$ iteratively by \textit{outward} conjugation of $C_{1}\bar{P}C^{\dagger}_{1}$ with $C_{2}$, while the contravariant Clifford tracking computes $C^{\dagger}_{1}C^{\dagger}_{2}PC_{2}C_{1}$ from $C^{\dagger}_{1}PC_{1}$ iteratively by \textit{inward} conjugation with $C^{\dagger}_{2}$ (see Eq.~\eqref{eq:update_2} and Ref.~\cite{Gidney2021stimfaststabilizer}).
Generally, both the inward and outward conjugation can be used to do covariant and contravariant tracking, but in order to be computationally efficient when appending Clifford circuits (i.e., by iteratively using intermediate results) the outward conjugation is suited for covariant tracking, while the inward conjugation is suited for contravariant tracking.
Using outward conjugation for contravariant tracking one could compute $C^{\dagger}_{1}C^{\dagger}_{2}PC_{2}C_{1}$ by outward conjugation of $C^{\dagger}_{2}PC_{2}$ with $C^{\dagger}_{1}$, but then $C^{\dagger}_{1}PC_{1}$ must be computed separately, resulting in an increasing computational cost for each appended Clifford gate. A similar argument holds for using inward conjugation to do covariant tracking. 
Note that prepending $C_{2}$ to $C_{1}$, rather than appending, swaps the preferred conjugation to inward conjugation for contravariant tracking and outward conjugation for covariant tracking.
For a more thorough discussion on the various choices of iteratively tracking added Clifford circuits and the resulting conjugation required, see Appendix Fig.~\ref{fig:decision_tree}.

Importantly, using the non-preferred conjugation iteratively requires tracking a symplectic minimal generating set of the Pauli group, contrary to the preferred conjugation, where the transformation of one individual Pauli operator can be tracked iteratively.
Comparing the two variants of tracking with the decomposition of the unitary in Eq.~\eqref{eq:en_de_coding}, we see that the contravariant tracking of Pauli operators of the flow tracking yields the logical meaning $\bar{R}_j$ for each physical rotation $R_j$, since $\bar{R}_1=C_{1}^{\dagger}R_{1}C_{1}$, $\bar{R}_2=C_{1}^{\dagger}C_{2}^{\dagger}R_{2}C_{2}C_{1}$, etc.

Although numerous methods track covariant Clifford transformations \cite{AaronsonGottesman2004} or contravariant Clifford transformations \cite{schmitz2023graphoptimizationperspectivelowdepth} (see Appendix Rmk.~\ref{rmk:existing_tableaus} for details), we propose a method to efficiently track both transformations in a single tableau (already hinted at in \cite{Gidney2021stimfaststabilizer}). 
To keep track of both the co- and contravariant transformations of the Pauli group one could naively track them in parallel with both the Clifford and the flow tableau. 
Alternatively, one could track only one of the tableaus and invert it after each time step to get the other (see Appendix Prop.~\ref{prop:inverse_flow_tableau}).
However, using the symplectic property of the tableaus we propose a more efficient method to track both mappings.
For this purpose, we consider the situation that the flow labels $\labelname[X_k]$ and $\labelname[Z_k]$ from contravariant tracking are given and demonstrate how to obtain the inverse transformations from the flow labels directly.
That is, the covariant transformation of the generators is given by
\begin{equation}
\label{eq:inverse_X}
    C\bar{X}_{j}C^\dagger =i^{\eta_{ j}} \prod_{{k}:\,{j \in \labelname[Z_{k}]}}X_{k} \prod_{{k}:\,{ j\in \labelname[X_{k}]}}Z_{k},
\end{equation}
\begin{equation}
\label{eq:inverse_Z}
    C\bar{Z}_{j}C^\dagger =i^{\eta_{\langle j \rangle }} \prod_{{k}:\,{\langle j\rangle\in \labelname[Z_{k}]}}X_{k} \prod_{{k}:\,{\langle j\rangle\in \labelname[X_{k}]}}Z_{k},
\end{equation}
for suitable \textit{Clifford phases} $\eta_{j},\eta_{\langle j\rangle} \in \{0,1,2,3\}$ (see Appendix Cor.~\ref{cor:combined_tableau_permutation}).
These $2n$ Clifford phases can be tracked in addition to the flow labels to avoid the time intensive need of inverting a tableau or the memory intensive tracking of an additional tableau in case both the co- and the contravariant transformations are needed.
The update rule for each Clifford phase can be computed analogous to Eq.~\ref{eq:update_1}, but with the conjugation inverted, i.e. $C\bar{g}_{j}C^{\dagger}$.
The resulting phase factor together with the additional phase obtained from reordering the generators into standard order tells us how to update these Clifford phases. 
Consequently, when appending $H_{t}$, the Clifford phases update according to
\begin{align}
    {\eta_{j}} &\to {\eta_{j}}\oplus2 
    && \text{if } j \in \ell_{X_{t}} \text{ and } j \in \ell_{Z_{t}} \,, \label{eq:H_X_clifford_phase}\\
    {\eta_{\langle j \rangle}} &\to \,{\eta_{\langle j \rangle}\oplus2} && \text{if } \langle j \rangle \in \ell_{X_{t}} \text{ and } \langle j \rangle \in \ell_{Z_{t}} \,, \label{eq:H_Z_clifford_phase}
\end{align}
and when appending $S_{t}$ according to
\begin{align}
    {\eta_{j}} &\to {\eta_{j}}\oplus 1 
    && \text{if } j \in \ell_{Z_{t}}  \,, \label{eq:S_X_clifford_phase}\\
    {\eta_{\langle j \rangle}} &\to \,{\eta_{\langle j \rangle}\oplus 1} && \text{if }  \langle j \rangle \in \ell_{Z_{t}} \,, \label{eq:S_Z_clifford_phase}
\end{align}
where the addition $\oplus$ is modulo $4$. 
Note that the $\text{CNOT}$ gate does not require updating of the Clifford phases, since it acts trivially. 
The time complexity of tracking the flow labels scales as $\mathcal{O}(n)$ per tracked elementary Clifford gate and the additional tracking of Clifford phases only changes the constant factor of the leading term.
This additional tracking adds $4n$ bits of memory requirement, but circumvents the $\mathcal{O}(n^{3})$ operations needed for the inversion of the tableau.
Having access to both the co- and contravariant transformations of the generators proves especially useful when designing quantum circuits in combination with stabilizers.
\section{Labeling stabilizer codes}
The Parity Flow formalism can also track the logical effects of operations on physical qubits in stabilizer codes and simultaneously track the stabilizers that define that code.
This reveals stabilizer-equivalent operations from both a co- and contravariant perspective and thereby flexibility in how a unitary can be applied. 
Furthermore, it shows how the auxiliary qubits in such codes can be used to implement more (commuting) operations in parallel, and thus can lead to significant reductions in quantum circuit depth \cite{klaver2024, Messinger_2023, Fellner2022_applications}.

\begin{figure*}[ht!]
    \centering
    \includegraphics[width=0.95\linewidth]{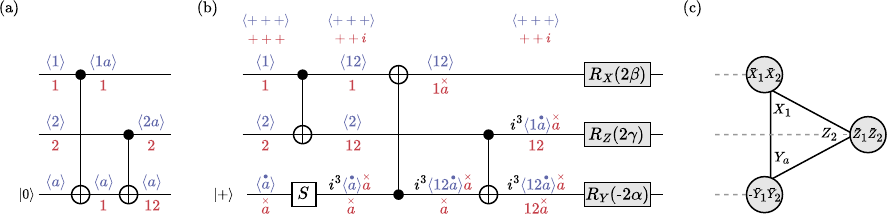}
    \caption{(a) $\mathrm{CNOT}$ circuit to obtain the effect of $\bar Z_{1}\bar Z_{2}$ on an initially empty auxiliary qubit $a$. The Clifford phases are omitted, since the action of the $\mathrm{CNOT}$ gates is trivial.
    (b) Encoding circuit to implement all logical operators of the Heisenberg model $e^{i\alpha\bar{Y}_{1}\bar{Y}_{2}}e^{i\beta\bar{X}_{1}\bar{X}_{2}}e^{i\gamma\bar{Z}_{1}\bar{Z}_{2}}$ in parallel after encoding redundantly with an auxiliary qubit. 
    Dots and crosses on top of an index denote stabilized and destabilized operators, respectively.
    After the encoding circuit, a physical $Y$ rotation on the auxiliary qubit 
    results in a valid logical $-\bar Y_1\bar Y_2$ operator, since
    $\textcolor{xlog}{\trivial{a}}$ 
    can be omitted from the label $\ell_{{X}_{a}}$ for interpreting its logical effect.
    The Clifford phases \textcolor{xlog}{$\langle i^{\eta_{\langle 1\rangle}}\dotsc\rangle$}\textcolor{zlog}{$i^{\eta_1}\dotsc$} are tracked on top of the circuit according to Eq.~\eqref{eq:S_X_clifford_phase}.
    (c) The resulting encoding of the circuit in (b) showing the stabilizer at the time step of the rotations. 
    The circles show the three qubits, labeled with the non-stabilizer violating logical operators (i.e., no crossed indices) corresponding to the local physical operators on the corners of the triangle. 
    The triangle indicates that the stabilizer $\bar{X}_{a}$ is mapped to $X_{1}Z_{2}Y_{a}$ by covariant tracking, revealing that the logical operator on $Y_{a}$ must be equal to the logical operator that corresponds to operator $X_{1}Z_{2}$.
    }
    \label{fig:auxiliary_qubits}
\end{figure*}

We start with the simplest case, where the encoding circuit of a stabilizer code consists only of $\mathrm{CNOT}$ gates\footnote{In this case, updating of Clifford phases is not yet needed, since $\mathrm{CNOT}$s act trivially.} and auxiliary qubits initialized in the $\ket{0}$ state.  
Since an auxiliary qubit initialized in the state $\ket{0}$ is indifferent to a Pauli-$Z$ rotation, we can assign an empty $Z$ label to such an auxiliary qubit at the start of a quantum circuit \cite{Amy_2018}.
Then, any $\mathrm{CNOT}$ gates that target the auxiliary qubit can change its $Z$ label to obtain a logical effect involving other qubits.
For example, the $Z$ operator of an auxiliary qubit~$a$ can obtain the logical effect of the two-qubit operator $\bar Z_{1}\bar Z_{2}$ through the encoding shown in Fig.~\ref{fig:auxiliary_qubits}(a).
Note that the initial $X$ label of the auxiliary qubit is not empty as the corresponding operation $\bar{X}_{a}$ has a non-trivial effect on the quantum state and it is therefore important to follow its presence throughout the labels.
We can describe initialized auxiliary qubits as part of a trivial stabilizer code, where the stabilizers are the operators corresponding to empty labels.
Importantly, we can only keep the labels empty as long as the quantum state remains within this stabilized code space.
The $Z$ stabilizers are tracked by the flow labels according to Eq.~\eqref{eq:inverse_Z}, resulting in $C\bar{Z}_{a}C^{\dagger}= Z_{1}Z_{2}Z_{a}$ for the Clifford circuit in Fig.~\ref{fig:auxiliary_qubits}(a), thereby tracking the code space deformation induced by the Clifford gates. 
Whether a physical non-Clifford operation commutes with the stabilizers can be determined directly from these labels: Any operation on a label which contains a non-omitted auxiliary index (in our example, any label containing $\flowlabel{}{a}{}$) would violate the stabilizer and result in a state outside the code space, which would invalidate the trivial action of the omitted auxiliary label.
Codes obtained by encoding circuits consisting only of $\mathrm{CNOT}$ gates and qubits initialized in the $\ket{0}$ state result in parity codes in which the stabilizers, labels, and label indices are of the same Pauli type~\cite{Messinger_2023}.
Tracking the (empty) labels during the encoding circuit of such codes can reveal a large set of locally implementable logical multi-qubit operations.
As a prominent example, the LHZ layout \cite{Lechner2015, Fellner2022_prl} encodes $k$ logical qubits in $n=k(k+1)/2$ physical qubits, such that each possible pair of logical qubit indices is represented on one of the physical qubit labels. 
This allows one to implement all pairwise logical $R_{ZZ}$ rotations via physical single-qubit $R_{Z}$ rotations in a single time step.

The $Z$ labels of parity codes are always restricted to only products of $Z$ operators.
However, with the flow labels we can generalize parity codes by extending the logical effect of single-qubit rotations from only products of $Z$ operators to products of arbitrary single-qubit Pauli operators and a phase.
Instead of setting the initial labels of auxiliary qubits to empty, we mark the indices corresponding to a stabilized generator by a dot, and the ones corresponding to an anti-commuting generator by a cross. 
Consequently, until a non-Clifford operator is applied that contains a cross in its label, all indices with dots can be omitted.
This has the same effect as setting the labels empty, but additionally allows for tracking the correct phases throughout the circuit.
An example of marked auxiliary qubits is shown in Fig.~\ref{fig:auxiliary_qubits}(b).
There, the logical operators $\bar{X}_{1}\bar{X}_{2}$, $\bar{Y}_{1}\bar{Y}_{2}$ and $\bar{Z}_{1}\bar{Z}_{2}$ are encoded onto physical single-qubit operators that can only be applied in parallel through the use of a stabilizer.
Note that the marked labels handle the specific case where auxiliary qubits are initialized in the $+1$-eigenstate of the single-qubit generators $X$ or $Z$ at the start of the quantum circuit.
However, more general cases where stabilizers are defined mid-circuit or as products of generators, for instance by multi-qubit measurements during the quantum circuit, are also possible (see Appendix Prop.~\ref{prop:multiqubitpauli_withauxies}).
To see how the stabilizer equivalence in Fig.~\ref{fig:auxiliary_qubits}(b) induces the simultaneous encoding of $\bar{X}_{1}\bar{X}_{2}$, $\bar{Y}_{1}\bar{Y}_{2}$ and $\bar{Z}_{1}\bar{Z}_{2}$ we read off $C\bar{X}_{a}C^{\dagger}=i\,X_{1}X_{a}\,Z_{2}Z_{a}=X_{1}Z_{2}Y_{a}$ using Eq.~\eqref{eq:inverse_X} for $\bar{X}_{a}$ with the Clifford phase $\textcolor{zlog}{i}$ and the labels $\textcolor{zlog}{\violate{a}}$ in Fig.~\ref{fig:auxiliary_qubits}(b).
This results in Fig.~\ref{fig:auxiliary_qubits}(c), along with single-qubit labels that do not violate this stabilizer.

This method applies to any stabilized Pauli operator $\bar{P}_{+}$ with quantum state $\ket{\psi}$ for which it is known that $\bar{P}_{+}\ket{\psi}=\ket{\psi}$.
From the covariant transformations in Eqs.~(\ref{eq:inverse_X}) and (\ref{eq:inverse_Z}) applied to the stabilizer in terms of its generators, we find the physical Pauli operator $P_{+}=C\bar{P}_{+}C^{\dagger}$.
This operator induces an equivalence in the logical effect of any two commuting Pauli operators for which ${P}_{1}{P}_{2}={P}_{2}{P}_{1}=P_{+}$ holds. 
Then, from the contravariant tracking in the flow labels, we find the logical effect of these physical Pauli operators as $\bar{P}=C^{\dagger}P_{1}C=C^{\dagger}P_{2}{P}_{+}C$, resulting in ${P}_{1}C\ket{\psi}=P_{2}C\ket{\psi}=C\bar{P}\ket{\psi}$.
Consequently, we can choose to apply a logical rotation of $\bar{P}$ via the physical rotation of $P_{1}$ or the physical rotation of $P_{2}$.
In addition to providing different decompositions [as in Eq.~\eqref{eq:alternating_unitary}] of a given unitary by choosing between Pauli rotations with the same logical effect, the combined tracking also reveals how to perform more operations in parallel, potentially reducing circuit depth, as demonstrated in Refs.~\cite{Fellner2022_prl, klaver2024}.

\section{Summary}
We have presented the Parity Flow formalism, a method for efficiently tracking contravariant transformations of quantum states and operators for appending Clifford circuits.
This formalism adds additional information to quantum circuit diagrams and thereby provides an intuitive way to understand unitaries composed of Pauli rotations and Clifford gates. 
Consequently, it aids in synthesis of such unitaries on real quantum hardware and enables us to find unitary decompositions with reduced total quantum gate depth and count.
Additionally, we highlight the relation between co- and contravariant Clifford transformations and propose a method to efficiently track both types within the Parity Flow formalism.
Lastly, we have shown how the combined tracking helps to utilize auxiliary qubits for parallel operations, and how the labels assist in quantum circuit design by revealing different ways to implement the same logical operations in stabilizer codes.

\section{Acknowledgments}
The authors thank F. Dreier, J. Farnsteiner, J. Kysela, N. Sakharwade and M. Traube for valuable comments and discussions.
This study was supported by NextGenerationEU via FFG and Quantum Austria (FFG Project No.~FO999896208). 
This research was funded in part by the Austrian Science Fund (FWF) under Grant-DOI 10.55776/F71.
This project was funded within the QuantERA II Programme that has received funding from the European Union’s Horizon 2020 research and innovation programme under Grant Agreement No. 101017733. 
This publication has received funding under Horizon Europe Programme HORIZON-CL4-2022-QUANTUM-02-SGA via the project 101113690 (PASQuanS2.1). 
This study was supported by the Austrian Research Promotion Agency (FFG Project No. FO999924030, FFG Basisprogramm).
For the purpose of open access, the authors have applied a CC BY public copyright license to any Author Accepted Manuscript version arising from this submission.

\clearpage

\appendix 
\appendixpage
\renewcommand\thefigure{\thesection.\arabic{figure}} 
\renewcommand\thetable{\thesection.\arabic{table}} 
\setcounter{figure}{0}
\setcounter{table}{0}
\counterwithin*{figure}{section}
\counterwithin*{table}{section}

While the main text focuses on the explanation and usage of the Parity flow labels for circuit design, the appendix provides the interested reader with the mathematical background of the formalism. For this, we aim at a self-contained introduction to the (re)presentation of the Pauli group and the representation of Clifford operators by tableaus in order to explain the different choices made for the formalism in comparison to the existing Clifford tableaus, e.g., \cite{dehaenedemoor2003, AaronsonGottesman2004, Gidney2021stimfaststabilizer, vandaele2022, Winderl2023, schmitz2023graphoptimizationperspectivelowdepth, GossetGrierKerznerSchaeffer2024}.

The first section introduces the chosen presentation (see Rmk.~\ref{rmk:grouppresentation}) of the Pauli group via the flow labels and their connection to the symplectic (re)presentation of the phaseless Pauli strings. The representation of Clifford circuits is covered in the second section, where we classify the representing tableaus into covariant Clifford tableaus and contravariant \emph{flow tableaus}. The tableaus used in some of the literature are compared with respect to this classification in Rmk.~\ref{rmk:existing_tableaus}. The combination of the two tableaus into the \emph{combined flow tableau} leads to an efficient way to combine the tracking with stabilizer codes in a circuit with auxiliary qubits and is covered in the last section.

Note that the proofs use the (combined) flow tableau instead of the flow labels (with Clifford phases). Both notations contain exactly the same information, but while the labels are more suitable for the (possibly even hand-written) annotation of quantum circuits, the tableau notation is more suitable for the mathematical background.

\section{Presentation of Pauli operators}
\subsection{Flow labels and vector-like notation for Pauli operators}\label{sec:flow_labels}

We start with an overview of notations we use for Pauli operators. The Pauli group on one qubit is generated by the Pauli matrices 
\begin{equation}
    Z=\left(\begin{smallmatrix}
        1&0\\0&-1
    \end{smallmatrix}\right)\quad X=\left(\begin{smallmatrix}
        0&1\\1&0
    \end{smallmatrix}\right)\quad Y=\left(\begin{smallmatrix}
        0&-i\\i&0
    \end{smallmatrix}\right),
\end{equation} 
alternatively, it is generated by $X$, $Z$ and $i\idmatrix$. The Pauli group $\bar{\Pauli}=\bar{\Pauli}_n$ on $\npq$ qubits consists of tensor products of Pauli operators on these $\npq$ qubits. In order to have a short hand notation the label tracking in the flow formalism is done in a dense set-theoretic form, where the short hand label
\[
\labelname=\flowlabel{-}{12}{3}
\] 
(``minus X12 Z3'') corresponds to the Pauli operator
\[
\bar{P}=\bar{P}(\labelname)=\logicalpauliop{-}{1,2}{3}, 
\]
where the index of $\textcolor{xlog}{\bar{X}_{\ipq}}$ etc.\ denotes that the single-qubit operator $X$ etc.\ acts on the $\ipq^\text{th}$ qubit and we use a bar on top of the Pauli operators to denote the \emph{logical} operators as will be explained later (p.~\pageref{sec:clifford_circuits}). 
We use the brackets $\flowlabel{}{\;}{}$ to differentiate between the $\logicalpauliop{}{ }{}$- and the $\logicalpauliop{}{}{ }$-labels and add the colors as an additional visual guide. Using either the brackets or the colors would be enough to give a unique meaning to such a \emph{flow label}, but removing brackets \emph{and} color would make the notation ambigious, e.g., $\bar{P}=\logicalpauliop{-}{1,2}{3}$ 
and $\bar{P}'=\logicalpauliop{-}{1}{2,3}$ 
do not correspond to the same label $-123$ but to $\flowlabel{-}{12}{3}$ vs. $\flowlabel{-}{1}{23}$.\footnote{To be precise, the $\logicalpauliop{}{ }{}$-indices $\textcolor{xlog}{\langle\ipq\rangle}$ and the $\logicalpauliop{}{}{ }$-indices $\textcolor{zlog}{\ipq}$ are elements in the disjoint union $\{\flowlabel{}{1}{},\dotsc,\flowlabel{}{\npq}{}\}\sqcup\{\flowlabel{}{}{1},\dotsc,\flowlabel{}{}{\npq}\}$ of two identical copies of $\{1,\dotsc,\npq\}$ and the colors and/or the brackets indicate, whether an index $\ipq$ lives in the first copy, written as $\flowlabel{}{\ipq}{}$, translating to $\logicalpauliop{}{\ipq}{}$, or in the second copy, written as $\flowlabel{}{}{\ipq}$, translating to $\logicalpauliop{}{}{\ipq}$.}

This shorthand heavily relies on the choice of a standard form for Pauli operators: 

Every Pauli operator $\bar{P}$ can be written \emph{uniquely} as a product of a phase, a product of single-qubit Pauli-$\logicalpauliop{}{ }{}$-operators $\logicalpauliop{}{\ipq}{}$ and a product of single-qubit Pauli-$\logicalpauliop{}{}{ }$-operators $\logicalpauliop{}{}{\ipq}$
\begin{equation}\label{eq:standard_form_pauli}
    \logicalp
    = i^\ep \cdot \prod_{\ipq=1}^\npq \logicalpauliop{}{\ipq}{}^{\logicalxi[\ipq]}\cdot\prod_{\ipq=1}^\npq \logicalpauliop{}{}{\ipq}^{\logicalzeta[\ipq]}
    =:i^\ep \bm{\logicalpauliop{}{ }{}}^{\bmlogicalxi} 
    \bm{\logicalpauliop{}{}{ }}^{\bmlogicalzeta},
\end{equation}
which we fix as its \emph{standard form}. This follows from the fact that the Pauli group is generated by $g_0=i\idmatrix, g_{\flowlabel{}{j}{}}=\logicalpauliop{}{\ipq}{}$ and $g_{\flowlabel{}{}{j}}=\logicalpauliop{}{}{\ipq}$ and their commutativity properties.

Its ``\emph{vector form}'' is then the element 
\begin{equation}
    \varrho(\logicalp):=\logicalflowp\in \reprpauli.
\end{equation}
Note that this is not a vector in the strict sense, see Rmk.~\ref{rmk:no-vectorspace} for details. By abuse of notation, the binary vectors $\bmlogicalxi$ resp.\ $\bmlogicalzeta$ can also be interpreted as sets of indices, the $\logicalpauliop{}{ }{}$- and $\logicalpauliop{}{}{ }$-components of the flow label, 
    \begin{equation}\begin{split}
        {\bmlogicalxi}&=(\logicalxi[1],\dotsc,\logicalxi[\npq])\in\mathbb{F}_2^\npq\\
        &\mapsto \mathrm{supp}(\bmlogicalxi):=\{\flowlabel{}{\ipq}{} \vert \logicalxi[\ipq]=1\}=:\bmlogicalxi\subset\{\flowlabel{}{1}{},\dotsc,\flowlabel{}{\npq}{}\},\\
        {\bmlogicalzeta}&=(\logicalzeta[1],\dotsc,\logicalzeta[\npq])\in\mathbb{F}_2^\npq\\
        &\mapsto \mathrm{supp}(\bmlogicalzeta):=\{\flowlabel{}{}{\ipq} \vert \logicalzeta[\ipq]=1\}=:\bmlogicalzeta\subset\{\flowlabel{}{}{1},\dotsc,\flowlabel{}{}{\npq}\}
    \end{split}\end{equation}
where only the indices of the non-trivial Pauli operators that actually appear in the product are written down making the notation dense
\begin{equation}\label{eq:standard_form_pauli_dense_label}
    \logicalp
    = i^\ep \cdot \prod_{\ipq=1}^\npq \logicalpauliop{}{\ipq}{}^{\logicalxi[\ipq]}\cdot\prod_{\ipq=1}^\npq \logicalpauliop{}{}{\ipq}^{\logicalzeta[\ipq]}
    =:i^\ep\prod_{m\in\labelname}g_m= \logicalp(\labelname),
\end{equation}
for the flow label $\labelname=i^\kappa\flowlabel{}{\bm{\xi}}{}\flowlabel{}{}{\bm{\zeta}}$.
In these flow labels, we include the whole phase $i^\ep$ in order to differentiate the phase label from the indices in the $\logicalpauliop{}{ }{}$- and $\logicalpauliop{}{}{ }$-labels, while in the ``vector form'' only the exponent $\ep$ is written down. An overview of the equivalent notations is given in Tab.~\ref{tab:paulinotations}, where $\bm{e_\ipq}$ is the $\ipq^{\text{th}}$ vector of the canonical basis. The flow label notation is optimized for handwriting in analyzing a circuit and the ``vector form'' is optimized for group theory and calculations, while the Pauli operators are the objects of interest. 

\begin{table}[!ht]
    \centering
        \begin{tabular}{p{12mm}|l|l|l|l|l}
        form & 
            example $\logicalp$ &
            $\logicalp=\idmatrix$ &
            $\logicalp=\bar{X}_\ipq$&$\logicalp=\bar{Z}_\ipq$&$\logicalp=\bar{Y}_\ipq$\\[1mm]\hline&&&&&\\[-2ex]
        flow label & 
        $\flowlabel{i}{23}{13}$ &
            &
            $\flowlabel{}{\ipq}{}$ &
            $\flowlabel{}{}{\ipq}$ &
            $\flowlabel{i}{\ipq}{\ipq}$\\[1mm]\hline
        ``vector form'' & 
            $\logicalpaulivec{1}{011}{101}$ &
             $\logicalpaulivec{0}{\bm{0}}{\bm{0}}$ &
             $\logicalpaulivec{0}{\bm{e_{\ipq}}}{\bm{0}}$ &
             $\logicalpaulivec{0}{\bm{0}}{\bm{e_{\ipq}}}$ &
             $\logicalpaulivec{1}{\bm{e_{\ipq}}}{\bm{e_{\ipq}}}$\\[1mm]\hline&&&&&\\[-2ex]
        operator &
            $\logicalpauliop{i}{2,3}{1,3}$ &
            $\idmatrix$ &
            $\logicalpauliop{}{\ipq}{}$&$\logicalpauliop{}{}{\ipq}$ &
            $\logicalpauliop{i}{\ipq}{\ipq}$\\[1mm]
         forms& $i{\color{xlog}\bm{X}^{(011)}}{\color{zlog}\bm{Z}^{(101)}}$&
            $\textcolor{xlog}{\bm{X}_{\ipq}^{\bm{0}}}\textcolor{zlog}{\bm{Z}_{\ipq}^{\bm{0}}}$ & 
            $\textcolor{xlog}{\bm{X}^{\bm{e_{\ipq}}}}\textcolor{zlog}{\bm{Z}^{\bm{0}}}$&
            $\textcolor{xlog}{\bm{X}^{\bm{0}}}\textcolor{zlog}{\bm{Z}^{\bm{e_{\ipq}}}}$ &
            $i\textcolor{xlog}{\bm{X}^{\bm{e_{\ipq}}}}\textcolor{zlog}{\bm{Z}^{\bm{e_{\ipq}}}}$\\[1mm]
        \end{tabular}
    \caption{Overview of the different notations. All notations are equivalent, respecting the group structures.}
    \label{tab:paulinotations}
\end{table}

\subsection{Group presentation of Pauli operators}

We can make the bijective map
\begin{equation}\begin{split}
    \varrho:\bar\Pauli&\longrightarrow \reprpauli\\
    \logicalp = i^\ep \bm{\logicalpauliop{}{ }{}}^{\bmlogicalxi}\bm{\logicalpauliop{}{}{ }}^{\bmlogicalzeta} &\longmapsto \varrho(\logicalp)=\logicalflowp,
\end{split}\end{equation}
into a group isomorphism by defining a compatible group operation $\circledast$ on the group $\reprpauli$, which we abbreviate as $G$ throughout the rest of the appendices.

\begin{definition}
    On $G=\reprpauli$ define the operation $\circledast$ by setting
    \begin{multline}
\logicalflowp[1]\circledast\logicalflowp[2]\\
:=\left(\ep_1\oplus_4\ep_2\oplus_4 2\bmlogicalzeta[1]\cdot\bmlogicalxi[2]^T|\bmlogicalxi[1]\oplus_2\bmlogicalxi[2]|\bmlogicalzeta[1]\oplus_2\bmlogicalzeta[2]\right),
\end{multline}
    where $\oplus_4$ denotes addition modulo $4$ in the ring $\mathbb{Z}_4$, $\oplus_2$ denotes addition modulo $2$ in the field $\mathbb{F}_2$ and $\cdot$ denotes standard matrix multiplication.
\end{definition}

Note that the correction term $2\bmlogicalzeta[1]\cdot\bmlogicalxi[2]^T$ in the phase exponent just mimics the sign change introduced by swapping the Pauli-$\logicalpauliop{}{ }{}$-operators of the second Pauli operator through the Pauli-$\logicalpauliop{}{}{ }$-operators of the first Pauli operator in order to get the standard form in Eq.~\eqref{eq:standard_form_pauli}, e.g.,
\begin{equation}\label{exp:groupop}\begin{split}
&(\logicalpauliop{i}{2,3}{}\overbrace{\logicalpauliop{}{}{1,3})\cdot(\logicalpauliop{-}{1}{}}^{\text{swap giving } (-1)^1}\logicalpauliop{}{}{3})\\
=&(-i)\cdot(-1)^{1}(\logicalpauliop{}{1,2,3}{1,3,3})=\logicalpauliop{i}{1,2,3}{1}\\
\end{split}\end{equation}
can be written as
\begin{equation}\begin{split}
&\Big(i^1 {\color{xlog}\bm{X}^{(011)}}\overbrace{{\color{zlog}\bm{Z}^{(101)}}\Big)\cdot \Big(i^2 {\color{xlog}\bm{X}^{(100)}}}^{\text{swap}}{\color{zlog}\bm{Z}^{(001)}} \Big)\\
=&\;i^{1\oplus_42}\cdot(-1)^{{\color{zlog}(101)}\cdot{\color{xlog}(100)^T}}{\color{xlog}\bm{X}^{(011)\oplus(100)}}{\color{zlog}\bm{Z}^{(101)\oplus(001)}}
\end{split}\end{equation}
\begin{equation}\begin{split}
=&\;i^{1\oplus_42\oplus_4 2{\color{zlog}(101)}\cdot{\color{xlog}(100)^T}}{\color{xlog}\bm{X}^{(011)\oplus(100)}}{\color{zlog}\bm{Z}^{(101)\oplus(001)}}\\
=&\;i^1{\color{xlog}\bm{X}^{(111)}}{\color{zlog}\bm{Z}^{(100)}},
\end{split}\end{equation}
in the ``vector form'' this corresponds to 
\begin{equation}\begin{split}
    &\logicalpaulivec{1}{011}{101}\circledast\logicalpaulivec{2}{100}{001}\\
    =&\logicalpaulivec{1\oplus_4 2\oplus_4 2\textcolor{zlog}{(101)}\textcolor{xlog}{(100)^T}}{011\oplus_2100}{101\oplus_2001}\\
    =&\logicalpaulivec{1}{111}{100},
\end{split}\end{equation}
i.e., only the exponents are written down, which makes the group isomorphism $\varrho$ sort of a generalized logarithm. In the flow label notation this corresponds to the set-theoretic operations in Fig.~\ref{fig:group_op}.

\begin{figure}[!htb]
    \centering
    \begin{picture}(250,45)
    \put(0,0){\includegraphics[width=0.45\linewidth]{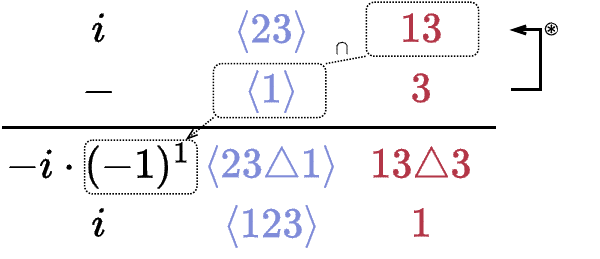}}
    \put(0,45){$a)$}
    \put(110,10){\includegraphics[width=0.55\linewidth]{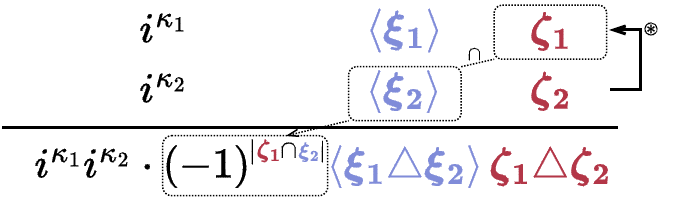}}
    \put(110,45){$b)$}
    \end{picture}
    \caption{The group operation $\circledast$ written in the flow label notation a) in the example of Eq.~\eqref{exp:groupop}, b) in the general case. 
    For $\bmlogicalxi$ and $\bmlogicalzeta$ just take the symmetric difference $\triangle$ of the sets. For the phase multiply the phases with the sign correction, which is $(-1)$ to the power of the number of elements $\vert\bmlogicalzeta[1]\cap\bmlogicalxi[2]\vert$ in the intersection marked by the rounded rectangles (identifying $\textcolor{xlog}{\langle\ipq\rangle}$ and $\textcolor{zlog}{\ipq}$), corresponding to the number of anticommuting swaps in order to produce the standard form of Pauli operators.}\label{fig:group_op}
\end{figure}

\begin{remark}
    For computational purposes it is advantageous to represent the phase exponent $\ep\in\mathbb{Z}_4$ as a 2-bit-number $\ep=\ep_1\cdot 2^1+\ep_0\cdot 2^0=(\ep_1,\ep_0)_2\in\mathbb{F}_2^2$. This corresponds to writing $i^\ep=(-1)^{\ep_1} i^{\ep_0}$ with binary exponents, which is nearly the same standard form as in~\cite{GossetGrierKerznerSchaeffer2024} or~\cite{dehaenedemoor2003}, see Rmk.~\ref{rmk:exponent_binary_number_dehaenedemoor} for details. Nevertheless, the group operation on the phase information remains addition modulo 4, which is the addition with carry-on on $\mathbb{F}_2^2$, where a possible overflow to the third lowest bit is ignored, corresponding to reduction modulo $4=2^2$. Since this might easily lead to confusion with the bitwise operations on the two other components $\mathbb{F}_2^\npq$ of $G$, we choose to use the $\mathbb{Z}_4$-component.
\end{remark}
\begin{proposition}\label{prop:groupG}
    $(G,\circledast)$ is a non-abelian group with identity element $\mathrm{id}_G=\logicalpaulivec{0}{\bm{0}}{\bm{0}}$, the inverse element of $\bm{\varrho}=\logicalflowp\in G$ is given by $\circleddash 
\bm{\varrho}=\circleddash \logicalflowp=\left(-\ep \oplus_4 2\bmlogicalzeta\cdot\bmlogicalxi^T\left\vert\bmlogicalxi\right\vert\bmlogicalzeta\right)$ and the centralizer of $G$ is $Z(G)=\{\logicalpaulivec{\ep}{\bm{0}}{\bm{0}}\vert\ep\in\mathbb{Z}_4\}\cong\mathbb{Z}_4$.
\end{proposition}
\begin{proof}
    The group axioms of associativity $(\bm{\varrho}_1\circledast \bm{\varrho}_2)\circledast \bm{\varrho}_3=\bm{\varrho}_1\circledast (\bm{\varrho}_2\circledast \bm{\varrho}_3)$, identity element $\bm{\varrho}\circledast\mathrm{id}_G=\mathrm{id}_G\circledast \bm{\varrho}=\bm{\varrho}$ and inverse element $\bm{\varrho}\circledast (\circleddash \bm{\varrho})=(\circleddash \bm{\varrho})\circledast \bm{\varrho}=\mathrm{id}_G$ can be shown by plugging in the definition of $\circledast$ and using the well known properties of the standard operations in $\mathbb{F}_2$, $\mathbb{Z}_4$ and the matrix multiplication.

    The centralizer of $G$ consists of all elements $\bm{\varrho}_1$ of $G$, such that $\bm{\varrho}_1$ commutes with \emph{all} elements $\bm{\varrho}_2\in G$, i.e., 
    \[
    \bm{\varrho}_1\circledast \bm{\varrho}_2=\bm{\varrho}_2\circledast \bm{\varrho}_1\qquad\forall \bm{\varrho}_2.
    \]
    Since the addition modulo $2$ in the $\bmlogicalxi$- and the $\bmlogicalzeta$-component is always commutative, the claim now follows from comparing the phase information:
    \[
    \begin{array}{rll}
    \ep_1\oplus_4\ep_2\oplus_4 2\bmlogicalzeta[1]{\bmlogicalxi[2]}^T&=\ep_2\oplus_4\ep_1\oplus_4 2\bmlogicalzeta[2]{\bmlogicalxi[1]}^T&\\
    2\bmlogicalzeta[1]{\bmlogicalxi[2]}^T&=2\bmlogicalzeta[2]{\bmlogicalxi[1]}^T&\forall \bmlogicalxi[2], \bmlogicalzeta[2]\\
    \Leftrightarrow\qquad\bmlogicalxi[1]&=\bmlogicalzeta[1]=\bm{0}.&
    \end{array}
    \]
    Since the centralizer is obviously not the whole group $G$, the group operation $\circledast$ is non-abelian.
\end{proof}
\begin{proposition}\label{prop:group_isomorphism}
    The map 
    \begin{equation}\begin{split}
        \varrho:\bar\Pauli&\longrightarrow G=\reprpauli\\
        \logicalp=i^\ep\bm{\logicalpauliop{}{ }{}}^{\bmlogicalxi}\bm{\logicalpauliop{}{}{ }}^{\bmlogicalzeta}&\longmapsto \logicalflowp
    \end{split}\end{equation}
    is a group isomorphism.
\end{proposition}
\begin{proof}
    For $\varrho$ to be a group homomorphism we have to prove $\varrho(\bar P\bar Q)=\varrho(\bar P)\circledast \varrho(\bar Q)$ for all $\bar P,\bar Q\in\bar\Pauli$, but that is exactly the definition of the operation $\circledast$.
    Since every $\logicalp\in\bar\Pauli$ can be uniquely written as $\logicalp=i^\ep\bm{\logicalpauliop{}{ }{}}^{\bmlogicalxi}\bm{\logicalpauliop{}{}{ }}^{\bmlogicalzeta}$ the map $\varrho$ is bijective and therefore a group isomorphism.
\end{proof} 
\begin{remark}\label{rmk:exponent_binary_number_dehaenedemoor}
    This proposition together with the definition of the group operation $\circledast$ is a variant of Lemma~1 in \cite{dehaenedemoor2003} (see also Lemma~33 in \cite{GossetGrierKerznerSchaeffer2024}). Because of their choice to write the phase exponent as a reversed 2-bit-binary-number their sign correction term contains one correction corresponding to the carry on of addition of 2-bit-numbers while the other correction corresponds to the anticommuting pairs of Pauli $X$- and $Z$-operators.
\end{remark}
\begin{remark}\label{rmk:scalar}
    The action of an element $\logicalpaulivec{\ep'}{\bm{0}}{\bm{0}}\in Z(G)$ on an arbitrary element $\varrho(\logicalp)=(\ep\vert\bmlogicalxi\vert\bmlogicalzeta)\in G$ corresponds to multiplying $i^{\ep'}\idmatrix\in\bar\Pauli$ and $\logicalp$ under this isomorphism, therefore giving the multiplication with a scalar $i^\ep$ in the Pauli group $\bar\Pauli$ (which is \emph{not} a scalar multiplication on the presentation $\reprpauli$, see Rmk.~\ref{rmk:no-vectorspace}). 
    We use the short hand $\ep'\circledast \varrho(\logicalp):=\logicalpaulivec{\ep'}{\bm{0}}{\bm{0}}\circledast (\ep\vert\bmlogicalxi\vert\bmlogicalzeta)=(\ep'\oplus_4\ep\vert\bmlogicalxi\vert\bmlogicalzeta)$ for this action.
    Note that $\ep'\circledast \varrho(\logicalp)=\varrho(\logicalp) \circledast \ep'$.
\end{remark}
\begin{remark}\label{rmk:grouppresentation}
    The group isomorphism $\varrho$ is a \emph{presentation} of the Pauli group in the following (broadened) sense. The Pauli group $\Pauli$ can be presented (in the strict group theoretic sense) as the free group on the symbols $i\idmatrix$, $\pauliop{}{\ipq}{}$, $\pauliop{}{}{\ipq}$ quotiented out by the relations given by their orders, i.e., $(i\idmatrix)^4=\pauliop{}{\ipq}{}^2=\pauliop{}{}{\ipq}^2=\idmatrix$ as well as those given by the (anti-)commutativity properties of the Pauli operators, e.g., $\pauliop{}{\ipq}{}\pauliop{}{}{\ipq}=(i\idmatrix)^2\pauliop{}{}{\ipq}\pauliop{}{\ipq}{}$. While the relations given by the orders and commutativity properties correspond to the commutative addition modulo $4$ resp.\ $2$ in the components $\mathbb{Z}_4$ resp.\ $\mathbb{F}_2$ in $G$, the relations corresponding to the anticommutativity properties get translated into the special operation $\circledast$. We will not use the free group in the following, but we will refer to $\varrho$ as presentation.
\end{remark}

\begin{remark}\label{rmk:no-vectorspace}
We call the presentation $\varrho(\logicalp)=\logicalflowp$ the ``vector form'' since it looks like a vector, even if its first entry is in $\mathbb{Z}_4$. As an algebraic object it is \emph{not} a vector because of the missing scalar multiplication. The group $G=\reprpauli$ being non-commutative can neither be a vector space (over a field) nor a module (over some unital ring) since the distributive laws needed would directly imply that $G$ has to be commutative. When trying to define an operation as near as possible to a scalar multiplication there are some candidates each one with its own draw backs:

\begin{description}
    \item[Idea 1] We can try to push the exponentiation of Pauli operators with integers to the presentation. If we consider $\logicalp^a$ for natural numbers $a\in\mathbb{N}$ and $\logicalp\in\bar\Pauli$ with $\varrho(\logicalp)=\logicalflowp$, then we can of course write 
    \begin{align}
        \logicalp^a&=\underbrace{\logicalp\cdot\cdots\cdot \logicalp}_{a} \mapsto\underbrace{\logicalflowp\circledast\cdots\circledast\logicalflowp}_{a} \nonumber\\
        &=\left(\left.\left.a\ep\oplus_42\cdot\left\lfloor \frac{a}{2}\right\rfloor\bmlogicalzeta\cdot\bmlogicalxi^T\right|a\bmlogicalxi\right|a\bmlogicalzeta\right)\label{eq:scalarmultiplication} .
    \end{align}
    It would be natural to define $a\cdot\logicalflowp$ as this last expression, which works similarly for $\logicalp^{-a}=(\logicalp^{-1})^a$.  
    In that case, we can define $\cdot:\mathbb{Z}\times G\rightarrow G$ and check which axioms of a $\mathbb{Z}$-module our group $G=\reprpauli$ with this operation fulfills. For this operation to be a scalar multiplication, we need to check whether 
    $a(\logicalflowp[1]\circledast\logicalflowp[2])=a\logicalflowp[1]\circledast a\logicalflowp[2]$. This corresponds to the equation $(\logicalp\bar Q)^a=\logicalp^a\bar Q^a$, which in general is not true if $\logicalp$ and $\bar Q$ don't commute. The other three axioms for a scalar multiplication correspond to $\logicalp^a\cdot \logicalp^b=\logicalp^{a+b}$, $(\logicalp^a)^b=\logicalp^{ab}$ and $\logicalp^1=\logicalp$ and are therefore fulfilled $a\logicalflowp$, since $\logicalp$ always commutes with itself.
    \item[Idea 2] We can restrict the above operation to the field $\mathbb{F}_2$ and consider $\cdot:\mathbb{F}_2\times G\rightarrow G$, which corresponds to just looking for $\logicalp^0=\idmatrix$ and $\logicalp^1=\logicalp$. This removes the problem of the first idea, since indeed $(\logicalp\bar Q)^a=\logicalp^a\bar Q^a$ if we only consider $a=0$ and $a=1$. But now we get $\idmatrix=\logicalp^{0}=\logicalp^{1\oplus_21}=\logicalp^1\cdot \logicalp^1=\logicalp^2\overset{\text{i.g.}}{\not=}\idmatrix$, which is only true for Pauli operators with order at most $2$ and false if $\logicalp\in\Pauli$ has $\ord(\logicalp)=4$, e.g., $\logicalp=iX$.
    \item[Idea 3] We can restrict the operation instead to the ring $\mathbb{Z}_4$, which solves the problem of idea 2, but we are back to the problem of idea 1.
    \item[Idea 4] In another direction we could try for a scalar multiplication coming from the centralizer of the group, i.e., $\circledast:\mathbb{Z}_4\times G\rightarrow G$ with $\ep'\circledast\logicalflowp=\logicalpaulivec{\ep'}{\bm{0}}{\bm{0}}\circledast\logicalflowp$, but since this corresponds to multiplying with $i^{\ep'}\idmatrix$ we get $i^{\ep'}\idmatrix\cdot(\logicalp\cdot\bar Q)=(i^{\ep'}\idmatrix\cdot\logicalp)\cdot \bar Q=\logicalp\cdot(i^{\ep'}\idmatrix\cdot \bar Q)\not=(i^{\ep'}\idmatrix\cdot\logicalp)\cdot(i^{\ep'}\idmatrix\cdot \bar Q)$, such that the axiom $\ep'\circledast(\logicalflowp[1]\circledast\logicalflowp[2])=(\ep'\circledast\logicalflowp[1])\circledast(\ep'\circledast\logicalflowp[2])$ is not true in general. Note that in the matrix algebra with matrix addition, matrix multiplication and scalar multiplication this is indeed the well known scalar multiplication $\alpha M$ of a matrix by a scalar. But this fulfills the axioms of scalar multiplication with respect to the matrix addition and not with respect to the matrix multiplication, which is the group operation of the Pauli group.
\end{description}
Having said that, we nevertheless use the definition of a multiple of a group element $a\cdot\logicalflowp$ in Eq.~\eqref{eq:scalarmultiplication}, especially in the cases $a\in\mathbb{F}_2$ or $a\in\mathbb{Z}_4$, but bear in mind that not all rules for a scalar multiplication are fulfilled. 
\end{remark}

Since the group operation $\circledast$ on $G$ is defined in such a way that the presentation $\varrho$ is a group isomorphism, we get the following corollary on the ``vector'' presentation of a product of multiple Pauli operators. For compact formulas here and in the rest of the paper we first give the following

\begin{definition}
    For any row vector $\bm{\varrho}=(\varrho_{1},\dotsc,\varrho_{m})\in\mathbb{F}_2^m$ we define $Q(\bm{\varrho})$ as the \textbf{str}ictly \textbf{up}per triangular part ($\mathrm{strup}$ for short) of the $m\times m$ matrix $\bm{\varrho}^T\bm{\varrho}$, i.e., $Q(\bm{\varrho}):=\mathrm{strup}(\bm\varrho^T\bm\varrho)$
    \begin{equation}
        Q(\bm{\varrho}):=\left(\begin{array}{cccccc}
        0&\varrho_{1}\varrho_{2}&\varrho_{1}\varrho_{3}&\dots&\dots&\varrho_{1}\varrho_{m}  \\
             0&0&\varrho_{2}\varrho_{3}&\dots&\dots&\varrho_{2}\varrho_{m} \\
             0&0&0&\ddots&&\varrho_{3}\varrho_{m} \\
             \vdots& \vdots&\vdots&\ddots&\ddots&\vdots \\
             0& 0& 0&\dots&0&\varrho_{m-1}\varrho_{m} \\
             0&0&0&\dots&0&0 \\
        \end{array}\right).
    \end{equation}
    For a row vector $\bm{\varrho}=(\varrho_{0}\vert\varrho_{1}\dotsc\varrho_{m})\in\mathbb{Z}_4\times\mathbb{F}_2^{m}$ that has a $0^\text{th}$ element, we define $Q(\bm\varrho)$ by the exact above definition, i.e., the $0^\text{th}$ element gets omitted.
    For any Pauli operator $\logicalp\in\bar\Pauli$ with ``vector form'' $\varrho(\logicalp)=\logicalflowp$ we define the triangular $2\npq\times2\npq$-matrix $Q(\logicalp):=Q((\bmlogicalxi\vert\bmlogicalzeta))$.
\end{definition}

\begin{corollary}\label{cor:multi_pauli_product}[see Cor.~34 in \cite{GossetGrierKerznerSchaeffer2024}]
    The ``vector'' presentation $\varrho(\logicalp)$ of the product $\logicalp=\logicalp_{1}\cdots \logicalp_{m}$ of Pauli operators $\logicalp_i\in\bar\Pauli$ with ``vectors'' $\bm{\varrho_i}=\varrho(\logicalp_i)=\logicalpaulivec{\ep_i}{\bm{\xi_i}}{\bm{\zeta_{i}}}$ is given by
    \begin{equation}
        \varrho(\logicalp)=\bigg(
            \sum_i \ep_i \oplus_4 2s 
            \bigg\vert
                \sum_{i}\bmlogicalxi[i]
            \bigg\vert
                \sum_{i}\bmlogicalzeta[i]
        \bigg),
    \end{equation}
    where the sign correction is given by 
    \begin{equation}
        s=(\bmlogicalzeta[1],\dotsc,\bmlogicalzeta[m])\cdot Q(\bm{1}) \cdot \left(\begin{array}{c}{\bmlogicalxi[1]}^T\\\vdots\\{\bmlogicalxi[m]}^T\end{array}\right)
    \end{equation} with $\bm{1}=(1,\dotsc,1)\in\mathbb{F}_2^m$.
\end{corollary}

\begin{remark}
    Note that the formula for the sign correction $s$ is only a formal notation, since the entries of the ``row vector'' $(\bmlogicalzeta[1],\dotsc,\bmlogicalzeta[m])$ are themselves row vectors, while the entries of the ``column vector'' on the right ${\bmlogicalxi_i}^T$ are column vectors themselves. Treating the $\bmlogicalxi[\ipq]$ and $\bmlogicalzeta[\ipq]$ as formal symbols yields
    \begin{equation}\begin{split}
        s&=(0, \bmlogicalzeta[1], \bmlogicalzeta[1]\oplus_2\bmlogicalzeta[2],\dotsc, \bmlogicalzeta[1]\oplus_2\dots\oplus_2\bmlogicalzeta[m-1])\cdot \left(\begin{array}{c}{\bmlogicalxi[1]}^T\\\vdots\\{\bmlogicalxi[m]}^T\end{array}\right)\\
        &=\sum_{i=1}^{m-1}\sum_{j=i+1}^{m}\bmlogicalzeta[i]\bmlogicalxi[j]^T.
    \end{split}\end{equation}
    Alternatively, this double sum is exactly the sum of all entries of the strictly upper triangular matrix
    \begin{equation}\label{eq:commutation_multiple_paulis}
        \mathrm{strup}\left(\begin{array}{ccc}
            \bmlogicalzeta[1]\bmlogicalxi[1]^T & \cdots & \bmlogicalzeta[1]\bmlogicalxi[m]^T \\
            \vdots&&\vdots\\
            \bmlogicalzeta[m]\bmlogicalxi[1]^T & \cdots & \bmlogicalzeta[m]\bmlogicalxi[m]^T,
        \end{array}\right).
    \end{equation}
    Note that the full matrix $(\bmlogicalzeta[i]\bmlogicalxi[j])_{i,j}$ contains all the information on the sign corrections $s_{i,j}$ when writing the product of any pair of two Pauli operators $\logicalp_{i}\cdot\logicalp_{j}$ with $1\leq i,j\leq m$ in the standard form of Eq.~\eqref{eq:standard_form_pauli}. Because of the ordering in the product $\logicalp_{1}\cdots\logicalp_{m}$ we only need those sign corrections corresponding to $i<j$, i.e., the strictly upper part of the matrix. Hence, the sign correction can be calculated by multiplying the triangular matrix by the row vector $\bm{1}$ from the left and by the column vector $\bm{1}^T$ from the right. 
\end{remark}
\begin{proof}
    Since $\varrho$ is a group homomorphism we get
    \begin{equation}
        \begin{split}
            \varrho(\logicalp)&=\varrho(\logicalp_{1}\cdots \logicalp_{m})=\varrho(\logicalp_{1})\circledast \cdots\circledast\varrho(\logicalp_{m})\\
            &=\logicalpaulivec{\ep_1}{\bm{\xi_1}}{\bm{\zeta_{1}}}\circledast \cdots\circledast\logicalpaulivec{\ep_m}{\bm{\xi_m}}{\bm{\zeta_{m}}}.
        \end{split}
    \end{equation}
    Repeatedly using the definition of $\circledast$ gives the linear part 
    \begin{equation}\label{eq:linear_part_Pauli_product}
\bigg(\sum_{i}\ep_{i}\bigg\vert\sum_{i}\bmlogicalxi[i]\bigg\vert\sum_{i}\bmlogicalzeta[i]\bigg)
    \end{equation}
    and the phase correction
    \begin{equation}\label{eq:double_sum_phase_correction}
        2s=2\cdot\sum_{i=1}^{m-1}\sum_{j=i+1}^{m}\bmlogicalzeta[i]\bmlogicalxi[j]^T,
    \end{equation}
    as claimed. 
\end{proof}
\begin{remark}\label{rmk:runtime_pauliproduct}    
    The most expensive part of the calculation of the sign correction corresponds to calculating (approximately half of) the matrix multiplication of two $m\times m$-matrices in order to get Eq.~\eqref{eq:commutation_multiple_paulis}, i.e., the run time in the naive matrix-multiplication is $\mathcal{O}(m^3)$ and with fast matrix-multiplication it is $\mathcal{O}(m^\omega)$, where $2\leq\omega<2{,}371339$ is the matrix multiplication exponent (see Alman et~al.~\cite{alman2024asymmetryyieldsfastermatrix} for the bound and \cite{GossetGrierKerznerSchaeffer2024} for using fast matrix multiplication in the context of Clifford tableaus). Since the calculation of the linear part of $\varrho(\logicalp)$ (Eq.~\eqref{eq:linear_part_Pauli_product}) is $\mathcal{O}(mn)$ we get a run time of $\mathcal{O}(\max(mn,m^\omega))$.
\end{remark}

\subsection{Vector space presentation of Pauli operators without phase}
If the phases of the Pauli operators are not needed, for instance if only $\CNOT$-gates are used, the somewhat complicated group operation $\circledast$ and the problem that we work with only ``nearly vectors'' vanishes, see for example~\cite{Winderl2023}, Vandaele et~al.~\cite{vandaele2022}. This corresponds to the induced presentation $\tilde\varrho$ on the quotient group $\tilde\Pauli:=\bar\Pauli/\langle i\idmatrix\rangle$, which is well defined since the subgroup $\langle i\idmatrix\rangle$ generated by the phases is a normal subgroup. Furthermore, the quotient group is commutative and an $\mathbb{F}_2$-vector space in the canonical way. The following proposition is well known:
\begin{proposition}
    The induced presentation
    \begin{equation}\begin{split}
        \tilde\varrho:\tilde\Pauli&\longrightarrow \mathbb{F}_2^{\npq}\times \mathbb{F}_2^{\npq}\\
        \left[\logicalp\right]=\left[i^\ep \bm{\logicalpauliop{}{ }{}}^{\bmlogicalxi}\bm{\logicalpauliop{}{}{ }}^{\bmlogicalzeta}\right]&\longmapsto (\bmlogicalxi\vert\bmlogicalzeta),
    \end{split}\end{equation}
    where $[\logicalp]$ is the equivalence class of $\logicalp$ in $\tilde\Pauli$, is a vector space isomorphism.
\end{proposition}

\subsection{The $XYZ$-mapping $r$}
Often, see for example the original paper~\cite{AaronsonGottesman2004} on the simulation of stabilizer circuits, another vector-like notation for the Pauli group is used. The chosen standard form for Pauli operators there is
\begin{equation}\label{eqn:XYZ-standard-form}
\logicalp=(-1)^\delta\prod_{j=1}^{\npq}\logicalp_{\ipq}, \text{ where }\logicalp_{\ipq}\in\{\idmatrix, \logicalpauliop{}{\ipq}{}, \bar{Y}_{\ipq}, \logicalpauliop{}{}{\ipq}\}
\end{equation}
and $\logicalp$ is represented by 
\[
r(\logicalp)=\left((x_1,\dotsc,x_{\npq}), (z_1,\dotsc,z_{\npq}),\delta\right),
\]
where $(x_{\ipq},z_{\ipq})$ is $(0,0), (1,0), (1,1)$ resp.\ $(0,1)$ if $\logicalp_{\ipq}$ is $\idmatrix, \logicalpauliop{}{\ipq}{}, \bar{Y}_{\ipq}$ resp.\ $\logicalpauliop{}{}{\ipq}$ and $\delta$ is the sign-exponent. 
The advantage of this notation is that only the Pauli operators which have a non-trivial $+1$-eigenspace are considered (and $\pm\idmatrix$), see also Rmk.~\ref{rmk:properPaulis}. 
These are the ``interesting'' operators in the stabilizer formalism in the sense that they actually define stabilizer states. The disadvantage is that the mapping is not a group isomorphism, since the Pauli operators of the above form do not form a group and the group structure of the full Pauli group given by the multiplication of two operators is not respected by this notation. Let us look more closely into this:

\begin{definition}
    Let $\logicalp\in\bar\Pauli$ be a Pauli operator. The order $\ord(\logicalp)$ of the operator is the smallest positive integer $a$, such that $\logicalp^a=\idmatrix$.

    Denote by $\bar\Pauli_{\leq2}\subsetneq\bar\Pauli$ the set of Pauli operators with $\ord(\logicalp)\leq2$.

    A Pauli operator $\logicalp\in\bar\Pauli_{\leq2}$ with $\logicalp\not=\pm\idmatrix$ is called \emph{proper}. A ``vector'' $\bm\varrho\in G$ is \emph{proper} if it is the image $\varrho(P)$ of a proper Pauli operator $P$.
\end{definition}
\begin{remark}\label{rmk:orderPauli}
    Since $\bar\Pauli$ is a finite group, every Pauli operator has finite order. Actually, the order of any Pauli operator is $1, 2$ or $4$: For every $\logicalp\in\bar\Pauli$ we have
    \begin{equation}\begin{split}
        \logicalp^2=\logicalp\cdot \logicalp&\mapsto 2\cdot\logicalflowp\\
        &=\paulivec{2\ep\oplus_4 2\bmlogicalzeta\cdot\bmlogicalxi^T}{\bm{0}}{\bm{0}}\\
        &\mapsto i^{2\ep\oplus_42\bmlogicalzeta\cdot\bmlogicalxi^T}\idmatrix\\
        &= (-1)^{\ep\oplus_2\bmlogicalzeta\cdot\bmlogicalxi^T}\idmatrix
    \end{split}\end{equation}
    and therefore $\logicalp^4=\idmatrix$ and $\ord(\logicalp)$ divides $4$.
    
    \centerline{\begin{tabular}{c|c|l}
        $\ord(\logicalp)$ & condition & examples \\\hline
        $1$ & $\logicalflowp=(0\vert\textcolor{xlog}{\bm{0}}\vert\textcolor{zlog}{\bm{0}})$ & $\idmatrix$\\\hline
        &&\\[-11pt]
        $2$ & $\ep\oplus_2\bmlogicalzeta\cdot\bmlogicalxi^T=0$ & $\logicalpauliop{}{ }{}=i^0\logicalpauliop{}{ }{}^1\logicalpauliop{}{}{ }^0$, \\
        &$\logicalflowp\not=(0\vert\textcolor{xlog}{\bm{0}}\vert\textcolor{zlog}{\bm{0}})$&$\bar{Y}=i^1\logicalpauliop{}{ }{}^1\logicalpauliop{}{}{ }^1$\\\hline
        &&\\[-11pt]
        $4$ & $\ep\oplus_2\bmlogicalzeta\cdot\bmlogicalxi^T=1$& $i\logicalpauliop{}{ }{}=i^1\logicalpauliop{}{ }{}^1\logicalpauliop{}{}{ }^0$, \\
        &&$-\logicalpauliop{}{ }{ }=i^2\logicalpauliop{}{ }{}^1\logicalpauliop{}{}{ }^1$
    \end{tabular}}

    The proper Pauli operators are exactly the Hermitian Pauli operators $\not=\pm\idmatrix$, i.e., those operators with $\ord(\logicalp)\leq2$ that are not a scalar multiple of the identity.
\end{remark}

It is easy to switch between the two vector-like forms by redistributing factors of $i$ in the two standard forms: 
\begin{proposition}\label{prop:randrho}
    Consider the following diagram of maps of sets
    \[\begin{tikzcd}[every cell/.style={inner xsep=1ex, inner ysep=0.85ex}]
        \bar\Pauli \arrow[r, "\varrho"] \arrow[rd, phantom, "\circlearrowleft", shift right=2pt]                            
            & G = \reprpauli   \\
        \bar\Pauli_{\leq2} \arrow[u, hook, "\iota"] \arrow[r, "r"'] 
            & \mathbb{F}_2^{n}\times \mathbb{F}_2^{n}\times \mathbb{F}_2 \arrow[u, "\psi"']
    \end{tikzcd}\]
    where $\varrho$ and $r$ are the $XZ$-presentation and the $XYZ$-mapping respectively and $\psi:\mathbb{F}_2^{n}\times \mathbb{F}_2^{n}\times \mathbb{F}_2\rightarrow G=\reprpauli$ is given by $\psi(\bmlogicalxi,\bmlogicalzeta,\delta):=(2\delta + \bmlogicalzeta\cdot\bmlogicalxi^T\vert\bmlogicalxi\vert\bmlogicalzeta)$.
    
    Then this diagram commutes, i.e., the concatenation $\psi\circ r$ is the same as the concatenation $\varrho\circ\imath$ of the inclusion $\imath:\bar\Pauli_{\leq2}\hookrightarrow\bar\Pauli$ and the presentation $\varrho$. In particular, the map $\psi$ is injective, but not surjective with image $G_{\leq2}=\varrho(\bar\Pauli_{\leq2})\subsetneq G$.
\end{proposition}
\begin{proof}
    Let $\logicalp\in\bar\Pauli$ be a Pauli operator in the $XYZ$ standard form of Eq.~\eqref{eqn:XYZ-standard-form} which we want to bring into $XZ$ standard form
    \begin{equation}\label{eq:changestdform}\begin{split}
        \logicalp&=(-1)^\delta\prod_{\ipq=1}^{\npq}\logicalp_{\ipq}=i^\ep \prod_{\ipq=1}^{\npq}\logicalpauliop{}{\ipq}{}^{\logicalxi_{\ipq}}\logicalpauliop{}{}{\ipq}^{\logicalzeta[\ipq]}.
    \end{split}\end{equation}
    Since we have
    \begin{equation}
        (\logicalxi[\ipq], \logicalzeta[\ipq]) = \left\{\begin{array}{ll}
            (0,0) & \text{if } \logicalp_{\ipq}=\idmatrix \\
            (1,0) & \text{if } \logicalp_{\ipq}=\logicalpauliop{}{\ipq}{}\\
            (1,1) & \text{if }\logicalp_{\ipq}=\bar{Y}_{\ipq}=i\logicalpauliop{}{\ipq}{\ipq}\\
            (0,1) & \text{if }\logicalp_{\ipq}=\logicalpauliop{}{}{\ipq}
        \end{array}\right.
    \end{equation}
    only the third case needs to be considered in order to compute the phase $\kappa$, this case occurs $\bmlogicalzeta\cdot\bmlogicalxi^T$ times, leading to 
    \begin{equation}\begin{split}
        \logicalp=(-1)^\delta\cdot i^{\bmlogicalzeta\cdot\bmlogicalxi^T} \prod_{\ipq=1}^{\npq}\logicalpauliop{}{\ipq}{}^{\logicalxi[\ipq]}\logicalpauliop{}{}{\ipq}^{\logicalzeta[\ipq]}&=i^{2\delta+\bmlogicalzeta\cdot\bmlogicalxi^T} \prod_{\ipq=1}^{\npq}\logicalpauliop{}{\ipq}{}^{\logicalxi[\ipq]}\logicalpauliop{}{}{\ipq}^{\logicalzeta[\ipq]}.
    \end{split}\end{equation}
    Therefore, the definition of $\psi$ and the fact that both standard forms are unique directly yield the commutativity of the diagram. Since $r$ and $\varrho$ are bijections, the properties of $\iota$, being injective but not surjective, since there exist elements of order $4$ in $\bar\Pauli$, transfer to $\psi$.
\end{proof}

\begin{corollary}\label{cor:forgetful_presentations}
    The forgetful presentations $\tilde{r}:\tilde\Pauli\rightarrow \mathbb{F}_{2}^{\npq}\times \mathbb{F}_{2}^{\npq}$ and $\tilde{\varrho}:\rightarrow \mathbb{F}_{2}^{\npq}\times \mathbb{F}_{2}^{\npq}$ defined on the quotient group $\tilde\Pauli=\bar\Pauli/\langle i\idmatrix\rangle$ induced by $r$ resp.\ $\varrho$, i.e., forgetting the phase factor, coincide.
\end{corollary}
\begin{remark}\label{rmk:phaseless_tilde}
    In the following we often have to refer to the phaseless version of some object, e.g., we need $(\bmlogicalxi\vert\bmlogicalzeta)$ of some $\bm{\varrho}=\logicalflowp$ or we need all entries of a ``matrix'' $M$ but those corresponding to phases. We will denote such a phaseless version of the object with a tilde, e.g., $\tilde{\bm\varrho}$ or $\tilde{M}$.
\end{remark}

\begin{remark}\label{rmk:properPaulis}
    Coming back to the proper Pauli operators, we have seen up to now that the $XYZ$-mapping is defined on Pauli operators with order up to $2$, i.e., $\logicalp^2=\idmatrix$. Since in the stabilizer formalism a stabilizer state $\state$ is characterized by the group $\stab\state$ of those Pauli operators $\logicalp$ that fix the state, i.e., $\logicalp\state=\state=+1\state$, any Pauli operator appearing in a stabilizer group has to have a non-trivial $+1$-eigenspace in the first place. Since $\logicalp$ is unitary, all eigenvalues $\lambda$ have absolute value $1$. And since $\logicalp^4=\idmatrix$, only $\pm1$ and $\pm i$ can be eigenvalues of $\logicalp$.

    The condition that $\ord(\logicalp)\leq2$ means that $\logicalp$ is Hermitian, since then $\logicalp^2=\idmatrix$, and therefore, we get real eigenvalues for $\logicalp$, more precisely $\pm1$. Note that the elements $\pm\idmatrix\in\bar\Pauli_{\leq2}$ by definition are \emph{not} proper Pauli operators, since they have either the whole space or the trivial subspace as $+1$-eigenspace. If you compare this to Sec.~10.5.1 on the stabilizer formalism in~\cite{NielsenChuang2011}, especially Prop.~10.5, this corresponds to the careful avoidance of the element $-\idmatrix$ in the stabilizer group.

    From Prop.~\ref{prop:randrho} we conclude that $\logicalflowp$ corresponds to a proper Pauli operator, iff $\ep\oplus_2\bmlogicalzeta\cdot\bmlogicalxi^T=0$ and $(\bmlogicalxi\vert\bmlogicalzeta)\not=(\textcolor{xlog}{\bm0}\vert\textcolor{zlog}{\bm0})$. This condition than makes sure that we can write the Pauli operator in the $XYZ$ standard form of Eq.~\eqref{eqn:XYZ-standard-form}.
\end{remark}
\begin{remark}
    Note that the $XYZ$-mapping $r$ is not a group isomorphism. On the one hand that is the case, since it is only defined on $\bar\Pauli_{\leq2}\subset\bar\Pauli$, which is not even a group. For example, $\logicalpauliop{}{ }{}$ and $\logicalpauliop{}{}{ }$ are in $\bar\Pauli_{\leq2}$ because $\logicalpauliop{}{ }{}^2=\logicalpauliop{}{}{ }^2=\idmatrix$, but their product $\logicalpauliop{}{ }{ }$ is not in $\bar\Pauli_{\leq2}$, since $(\logicalpauliop{}{ }{ })^2=\logicalpauliop{}{ }{ }\logicalpauliop{}{ }{ }=-\logicalpauliop{}{ }{}^2\logicalpauliop{}{}{ }^2=-\idmatrix$, hence $\ord(\logicalpauliop{}{ }{ })=4$. Equivalently, we could look on the eigenspaces and see that $\logicalpauliop{}{ }{}$ and $\logicalpauliop{}{}{ }$ each have a one-dimensional $+1$-eigenspace (and an orthogonal one-dimensional $-1$-eigenspace), while $\logicalpauliop{}{ }{ }$ has $\pm i$ as eigenvalues. Hence, $\bar\Pauli_{\leq2}$ is not closed under multiplication and not a subgroup. 

    On the other hand, writing $+\bar{Y}$ as $(1,1,0)$ is not compatible with writing $(1,0,0)$ and $(0,1,0)$ for $+\logicalpauliop{}{ }{}$ and $+\logicalpauliop{}{}{ }$ -- at least for the natural addition on $\mathbb{F}_{2}^{n}\times\mathbb{F}_{2}^{n}\times\mathbb{F}_{2}$ -- since
    \begin{equation}
        +\bar{Y}\mapsto(1,1,0)=(1,0,0)\oplus_2(0,1,0)\mapsto \logicalpauliop{}{ }{}\cdot \logicalpauliop{}{}{ }, 
    \end{equation}
    but $\bar{Y}=\logicalpauliop{i}{ }{ }\not=\logicalpauliop{}{ }{ }$.
\end{remark}

\subsection{(Anti-)commutativity}
Recall the following (anti-)commutativity properties of Pauli operators:
\begin{itemize}
    \item Two Pauli operators $\logicalp, \bar{Q}\in\bar\Pauli$ either commute or anticommute. In either of the two standard forms (see Eqns.~\eqref{eq:standard_form_pauli} and \eqref{eqn:XYZ-standard-form}) count the number of qubits, such that the $\ipq^\text{th}$ single-qubit factors of $\logicalp$ and $\bar Q$ anticommute. If this number is even, $\logicalp$ and $\bar{Q}$ commute, if it is odd, they anticommute.
    \item Denote by $\omega:\mathbb{F}_2^{2\npq}\times \mathbb{F}_2^{2\npq}\rightarrow\mathbb{F}_2$ the symplectic inner product given by $\Omega=\left(\begin{array}{c|c}
        0 & \mathbb{I} \\\hline
        \mathbb{I} & 0
    \end{array}\right)$,
    i.e., for $(\bmlogicalxi[1]\vert\bmlogicalzeta[1])$ and $(\bmlogicalxi[2]\vert\bmlogicalzeta[2])\in\mathbb{F}_2^{\npq}\times\mathbb{F}_2^{\npq}=\mathbb{F}_2^{2\npq}$ we define
    \begin{equation}\begin{split}
        \omega((\bmlogicalxi[1]\vert\bmlogicalzeta[1]),(\bmlogicalxi[2]\vert\bmlogicalzeta[2]))&:=(\bmlogicalxi[1]\vert\bmlogicalzeta[1])\Omega(\bmlogicalxi[2]\vert\bmlogicalzeta[2])^T\\
        &\phantom{:}=\bmlogicalzeta[1]\bmlogicalxi[2]^T\oplus \bmlogicalxi[1]\bmlogicalzeta[2]^T
    \end{split}\end{equation}
    The symplectic inner product of $\tilde\varrho(P)$ and $\tilde\varrho(Q)$ for two Pauli operators $\logicalp, \bar Q$ exactly counts the number of anticommuting single-qubit factors of $\logicalp$ and $\bar Q$, hence we have $\logicalp\bar{Q}=(-1)^\nu \bar{Q}\logicalp$ if and only if $\omega(\tilde\varrho(\logicalp),\tilde\varrho(\bar{Q}))=\tilde\varrho(\logicalp)\Omega \tilde\varrho(\bar{Q})^T=\nu$.
\end{itemize}
Note that the sign correction term $\bmlogicalzeta[1]\bmlogicalxi[2]^T$ (omitting the factor $2$) of the $\circledast$-sum $\varrho(\logicalp)\circledast\varrho(\bar Q)=\logicalflowp[1]\circledast\logicalflowp[2]=(\ep_1\oplus_4\ep_2\oplus_42\bmlogicalzeta[1]\bmlogicalxi[2]^T\vert\bmlogicalxi[1]\oplus\bmlogicalxi[2]\vert\bmlogicalzeta[1]\oplus\bmlogicalzeta[2])$ gives the first summand of the symplectic product, this corresponds to bringing the product $\logicalp\bar Q$ to the standard form. The second summand corresponds analogously to bringing the product $\bar Q\logicalp$ to the standard form. The factor $2$ is omitted since we stay in $\mathbb{F}_2$. The sum $\bmlogicalzeta[1]\bmlogicalxi[2]^T\oplus \bmlogicalxi[1]\bmlogicalzeta[2]^T$ then just adds the swaps of pairs $\logicalpauliop{}{\ipq}{}$ and $\logicalpauliop{}{}{\ipq}$ needed in going from $\logicalp\bar Q$ via the standard form to $\bar Q\logicalp$ and therefore, decides whether $\logicalp \bar Q=\bar Q\logicalp$ or $\logicalp\bar Q=-\bar Q\logicalp$. Note also that $\bmlogicalzeta[1]\bmlogicalxi[2]^T=(\bmlogicalxi[1]\vert\bmlogicalzeta[1])\left(\begin{smallmatrix}0&0\\\idmatrix&0\end{smallmatrix}\right)(\bmlogicalxi[2]\vert\bmlogicalzeta[2])^T$ (where $\left(\begin{smallmatrix}0&0\\\idmatrix&0\end{smallmatrix}\right)$ is $U^T$ of \cite{dehaenedemoor2003}). 

In order to simplify the notation, we also write $\omega$ for the concatenations $\omega\circ(\tilde\varrho, \tilde\varrho)$ and $\omega\circ(\tilde r, \tilde r)$, i.e., $\omega(\logicalflowp[1],\logicalflowp[2])$ is defined by forgetting the phase exponents and plugging the results into $\omega$ and analogously for $(\bmlogicalxi[i],\bmlogicalzeta[i],\delta_i)$.

\begin{corollary}\label{cor:proper_product_of_proper}
    Let $\logicalp_1, \logicalp_2\in\bar{\Pauli}$ be two Hermitian Pauli operators. Their product $\logicalp_1\cdot\logicalp_2$ is also Hermitian if and only if $\logicalp_1$ and $\logicalp_2$ commute.
\end{corollary}
\begin{proof}
    Let $\bm{\varrho_i}=\varrho(\logicalp_i)=\logicalflowp[i]$. Then $\kappa_i\oplus_2\bmlogicalzeta[i]\cdot\bmlogicalxi[i]^T=0$ since $\logicalp_i$ is Hermitian. For the product $\bm\varrho=\varrho(\logicalp_1\logicalp_2)=\bm{\varrho_1}\circledast\bm{\varrho_2}=
\left(\ep_1\oplus_4\ep_2\oplus_4 2\bmlogicalzeta[1]\cdot\bmlogicalxi[2]^T|\bmlogicalxi[1]\oplus_2\bmlogicalxi[2]|\bmlogicalzeta[1]\oplus_2\bmlogicalzeta[2]\right)=\logicalflowp$ we get
\begin{equation}\begin{split}
    \kappa+\bmlogicalzeta\cdot\bmlogicalxi^T&= (\kappa_1+\kappa_2+2\bmlogicalzeta[1]\cdot\bmlogicalxi[2]^T)\\
    &\qquad+(\bmlogicalzeta[1]+\bmlogicalzeta[2])\cdot(\bmlogicalxi[1]+\bmlogicalxi[2])^T\\
    &= \kappa_1+\kappa_2+2\bmlogicalzeta[1]\cdot\bmlogicalxi[2]^T\\
    &\qquad+\bmlogicalzeta[1]\cdot\bmlogicalxi[1]^T+\bmlogicalzeta[1]\cdot\bmlogicalxi[2]^T+ \bmlogicalzeta[2]\cdot\bmlogicalxi[1]^T+\bmlogicalzeta[2]\cdot\bmlogicalxi[2]^T\\
    &=(\kappa_1+\bmlogicalzeta[1]\cdot\bmlogicalxi[1]^T)+(\kappa_2+\bmlogicalzeta[2]\cdot\bmlogicalxi[2]^T)\\
    &\qquad+2\bmlogicalzeta[1]\cdot\bmlogicalxi[2]^T+ \bmlogicalzeta[1]\cdot\bmlogicalxi[2]^T+\bmlogicalzeta[2]\cdot\bmlogicalxi[1]^T\\
    &\equiv\bmlogicalzeta[1]\cdot\bmlogicalxi[2]^T+\bmlogicalzeta[2]\cdot\bmlogicalxi[1]^T\qquad\mathrm{mod} (2)\\
    &=\omega(\bm{\varrho_1},\bm{\varrho_2})
\end{split}\end{equation}
Therefore, the Hermitian property $\kappa\oplus_2\bmlogicalzeta\cdot\bmlogicalxi^T=0$ is fulfilled iff $\omega(\bm{\varrho_1},\bm{\varrho_2})=0$, i.e., iff the two Pauli operators commute.
\end{proof}

Later we will need the following well known facts:
\begin{proposition}\label{prop:symplectic_matrices}
    Let $\tilde F=\left(\begin{array}{c|c}
         A& B \\\hline
         C& D
    \end{array}\right)\in\mathbb{F}_2^{2\npq\times2\npq}$ be a symplectic matrix, i.e., $\tilde F\Omega\tilde F^T=\Omega$, with $\npq\times\npq$-blocks $A, B, C$ and $D$.

    Then the inverse matrix can be easily calculated by
    \begin{equation}
        \tilde F^{-1}=\left(\begin{array}{c|c}
         D^T& B^T \\\hline &\\[-11pt]
         C^T& A^T
    \end{array}\right).
    \end{equation}
    The following are equivalent
    \begin{enumerate}
        \item $\tilde F$ is symplectic.
        \item $AD^T+BC^T=\idmatrix$ and the matrices $BA^T$ and $CD^T$ are symmetric.
        \item $\tilde F^{-1}$ is symplectic.
        \item $A^TD+C^TB=\idmatrix$ and the matrices $B^TD$ and $A^TD$ are symmetric.
    \end{enumerate}
\end{proposition}

\begin{remark}\label{rmk:strup_block_form}
    Most often we will use the above block form for matrices like
    \begin{equation}
        \tilde{F}=\left(\begin{array}{c|c}
             \bmlogicalxi[1]&\bmlogicalzeta[1]  \\
             \vdots&\vdots\\
             \bmlogicalxi[\npq]&\bmlogicalzeta[\npq]  \\\hline
             \bmlogicalxi[\npq+1]&\bmlogicalzeta[\npq+1]  \\
             \vdots&\vdots\\
             \bmlogicalxi[2\npq]&\bmlogicalzeta[2\npq]
        \end{array}\right)    = \left(\begin{array}{c|c}
         \textcolor{xlog}{A}& \textcolor{zlog}{B} \\\hline
         \textcolor{xlog}{C}& \textcolor{zlog}{D}
    \end{array}\right).
    \end{equation}
    Then, the matrix in Eq.~\eqref{eq:commutation_multiple_paulis} can be calculated as
    \begin{equation}\begin{split}
        &\phantom{=}\mathrm{strup}\left(\begin{array}{ccc}
            \bmlogicalzeta[1]\bmlogicalxi[1]^T & \cdots & \bmlogicalzeta[1]\bmlogicalxi[2\npq]^T \\
            \vdots&&\vdots\\
            \bmlogicalzeta[2\npq]\bmlogicalxi[1]^T & \cdots & \bmlogicalzeta[2\npq]\bmlogicalxi[2\npq]^T,
        \end{array}\right)\\
        &=\mathrm{strup}\left(\tilde F\cdot \left(\begin{array}{c|c}
         0& 0 \\\hline
         \idmatrix& 0
    \end{array}\right)\cdot \tilde F^T\right)\\
    &=\mathrm{strup}\left(\begin{array}{c|c}
         \textcolor{zlog}{B}\textcolor{xlog}{A}^T& \textcolor{zlog}{B}\textcolor{xlog}{C}^T \\\hline
         &\\[-11pt]
         \textcolor{zlog}{D}\textcolor{xlog}{A}^T& \textcolor{zlog}{D}\textcolor{xlog}{C}^T
    \end{array}\right)\\
    &=\left(\begin{array}{c|c}
         \mathrm{strup}(\textcolor{zlog}{B}\textcolor{xlog}{A}^T)& \textcolor{zlog}{B}\textcolor{xlog}{C}^T \\\hline
         &\\[-11pt]
         0& \mathrm{strup}(\textcolor{zlog}{D}\textcolor{xlog}{C}^T)
    \end{array}\right).
    \end{split}\end{equation}
\end{remark}

\section{Representing Clifford Circuits}\label{sec:clifford_circuits}
Now that we have covered the needed prerequisites on Pauli operators we can turn our attention to the Clifford circuits we would like to track. For a Clifford circuit $C$ -- consisting of Hadamard gates $H$, $S$-gates and $\CNOT$ gates -- we can ask the following two different questions -- using the convention of the main text that \emph{logical} refers to Pauli operators, states etc.\ before the circuit, indicated by a bar in the notation, while \emph{physical} refers to Pauli operators, states, etc.\ after the circuit, written without a bar:
\begin{description}
    \item[Question~1 - logical to physical] Given a logical Pauli operator $\bar{P}$ acting as $\state[\bar\psi]\mapsto\logicalp\state[\bar\psi]$ on the logical state $\state[\bar\psi]$ \emph{before} applying the circuit $C$, what is the induced physical action of $\bar{P}$ \emph{after} the circuit $C$? 

    We want to determine the action $C_\ast \bar{P}$ on the physical state $\state$ (see question marks in Fig.~\ref{fig:pushforward}). The bold arrows show the way how the action of $\bar{P}$ is pushed forward to its induced action $\state\mapsto C\logicalp C^\dagger\state$.

    Algebraically, $C \bar{P} C^\dagger=:C_\ast\bar{P}$ is the \emph{pushforward} of $\bar{P}$ through $C$.
\end{description}
\begin{figure}[!htb]
    \centering
    \includegraphics[width=200pt]{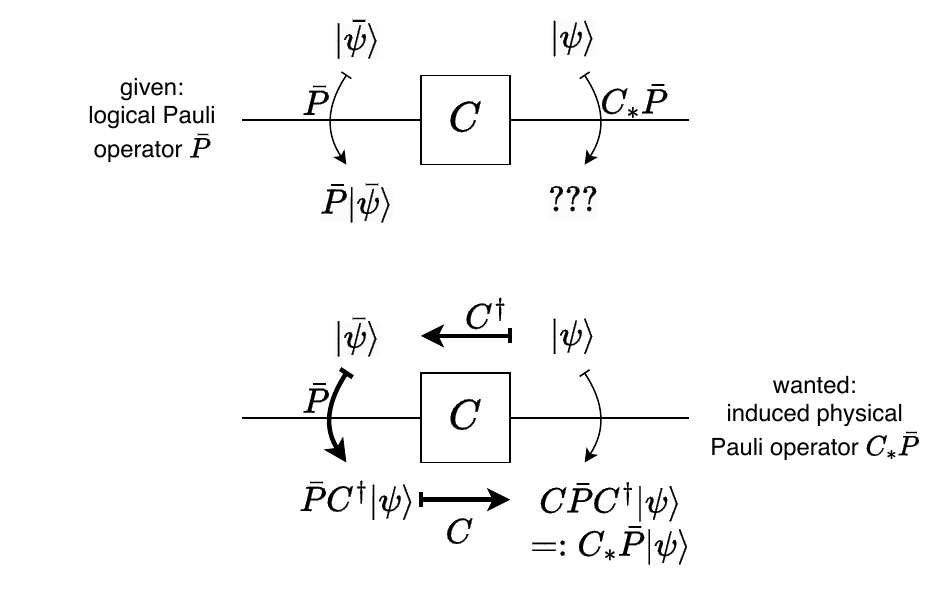}    
    \caption{The action of a logical Pauli operator $\logicalp$ on the logical state $\state[\bar\psi]$ gets pushed forward to the action $C_\ast\logicalp$ on the physical state $\state$ after the Clifford circuit $C$. $C$ acts on $\logicalp$ by left conjugation.}
    \label{fig:pushforward}
\end{figure}
    
\begin{description}
    \item[Question~2 - physical to logical] Given a physical Pauli operator $P$ acting on the physical state $\state$ \emph{after} the Clifford circuit $C$, i.e., $\state=C\state[\bar\psi]$, what is the induced logical action of $P$ \emph{before} the circuit $C$?

    In this variant we want to know the action $C^\ast P$ on the logical state $\state[\bar\psi]$, which we again get through the bold arrows now pulling back the action of $P$ to its induced action on $\state[\bar\psi]$.    

    Algebraically, $C^\dagger P C=:C^\ast P$ is the \emph{pullback} of $P$ through $C$.
\end{description}
\begin{figure}[!htb]
    \centering
    \includegraphics[width=200pt]{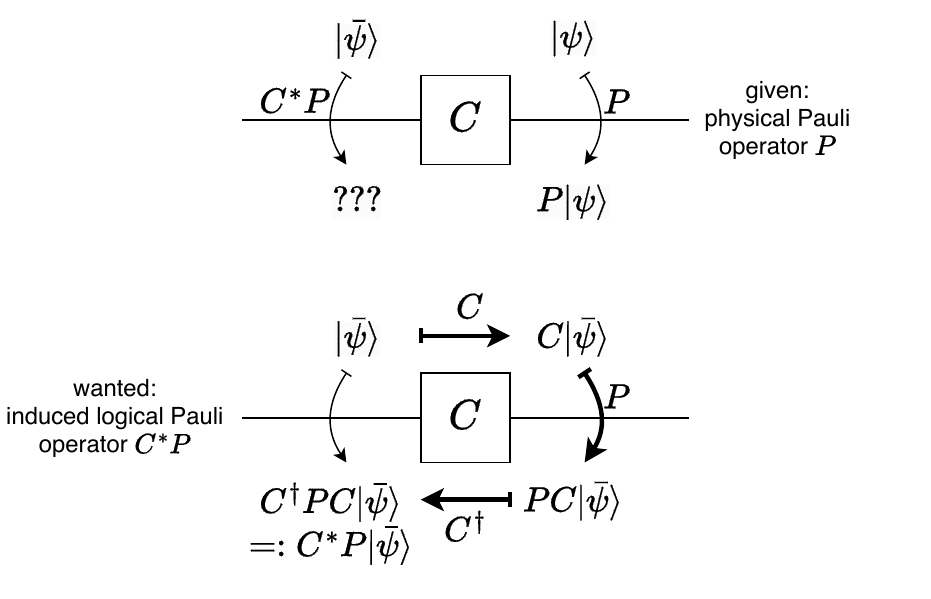}
    \caption{The action of a physical Pauli operator $P$ on the physical state $\vert\psi\rangle$ gets pulled back to the action $C^\ast P$ on the logical state $\vert\bar\psi\rangle$ before the Clifford circuit $C$. $C$ acts on $P$ by right conjugation.}
    \label{fig:pullback}
\end{figure}

For example, question~1 arises in the study of quantum error correction codes, if $\logicalp$ is a stabilizer of a logical state $\state[\bar\psi]$ and we want to know the induced stabilizer after the circuit $C$ in order to check whether some physical Pauli rotation $R_P$ respects the stabilizer. On the other hand, if we apply some physical Pauli rotation $R_P$ after the Clifford circuit $C$ and want to know its logical meaning on the state $\state[\bar\psi]$ we ask question~2.

In this section we will first introduce the action of the Clifford group on the Pauli group from the group theoretic side, identifying the covariant and contravariant conjugations of Pauli operators by Clifford unitaries and introducing the different choices of conjugation type, circuit/inverse circuit and prepend/append. Afterwards the flow and the Clifford tableau are introduced as an analogue to the well known matrix representation of linear maps. As the representation of linear maps highly depends on the chosen (ordered) basis, the tableaus highly depend on the choice of an ordered, symplectic generator set of the Pauli group. In order to have a self contained explanation of the two different tableaus we cover the technical details of this matrix-like representation which is mainly needed since covering the phases in the action destroys the algebraic structure as a vector space (see Rmk.~\ref{rmk:no-vectorspace}). In the remainder of the section the tableaus are categorized into co- and contravariant choices (see Fig.~\ref{fig:decision_tree}) and the practical side of calculating the tableaus in an iterative manner are covered (see Exp.~\ref{exp:heisenberg}).

\subsection{Clifford group and its action on the Pauli group}\label{subsec:clifford_on_pauli}

\subsubsection{Left vs.\ right conjugation}

\begin{definition}
    The \emph{Clifford group} $\Cl=\Cl_n$ on $n$ qubits is the normalizer of the Pauli group, i.e.,
    \begin{equation}
        \Cl=\{C| C \text{ unitary}, C^\dagger P C\in\bar\Pauli \text{ for all }P\in\Pauli\}.
    \end{equation}
\end{definition}
Note that we could have used any of the two types of conjugations
\begin{align}
&\text{left conjugation}& \logicalp &\mapsto C_{\ast}\logicalp= C \logicalp C^\dagger \label{eqn:pushforward}\\
&\text{right conjugation}& P &\mapsto C^{\ast}P = C^\dagger P C\label{eqn:pullback}.
\end{align}
in this definition. But since the normalizer is indeed a group, we have that $C^\dagger\in\Cl$ iff $C\in\Cl$, hence both conjugations lead to the same definition of $\Cl$. We use $\bar\Pauli$ resp.\ $\Pauli$ for the Pauli group to differentiate between the domain and the codomain of the conjugation types, nevertheless they are always the same group in these appendices.

Note that both conjugation types give group homomorphisms, i.e., they are compatible with the multiplication of Pauli operators:
\begin{equation}\label{eq:conjugation-grouphom}\begin{split}
    C_\ast(\logicalp\cdot \bar Q)=C(\logicalp\cdot \bar Q)C^\dagger=C\logicalp C^\dagger C\bar QC^\dagger=C_\ast \logicalp\cdot C_\ast \bar Q\\
    C^\ast(P\cdot Q)=C^\dagger(P\cdot Q)C=C^\dagger P C C^\dagger QC=C^\ast P\cdot C^\ast Q
\end{split}\end{equation}

The seemingly small difference of the left and the right conjugation lead to one significant difference: If we concatenate two Clifford circuits $C_1$ and $C_2$ to find a circuit $C=C_2C_1$, the order of the two individual circuits is \emph{preserved} by the left conjugation
\begin{equation}\label{eq:pushforward_covariant}\begin{split}
    C_\ast \logicalp&=(C_2C_1)_\ast \logicalp=(C_2C_1)\logicalp(C_2C_1)^\dagger\\
    &=C_2C_1\logicalp C_1^\dagger C_2^\dagger=C_2( {C_1}_\ast \logicalp) C_2^\dagger={C_2}_\ast( {C_1}_\ast \logicalp)
\end{split}\end{equation}
but \emph{reversed} by the right conjugation
\begin{equation}\label{eq:pullback_contravariant}\begin{split}
    C^\ast P&=(C_2C_1)^\ast P=(C_2C_1)^\dagger P(C_2C_1)\\
    &=C_1^\dagger C_2^\dagger P C_2C_1=C_1^\dagger( {C_2}^\ast P) C_1={C_1}^\ast({C_2}^\ast P).
\end{split}\end{equation}
In short, the pushforward is \emph{covariant}, i.e., $(C_2C_1)_\ast=(C_2)_\ast(C_1)_\ast$, while the pullback is \emph{contravariant}, i.e., $(C_2C_1)^\ast=(C_1)^\ast(C_2)^\ast$. This gives the following 
\begin{proposition}\label{prop:conjugation-functors}
    The two conjugation types, the pushforward
    \begin{equation}\begin{split}
        {}_\ast:\Cl\times\bar\Pauli&\rightarrow\Pauli\\
        (C,\logicalp)&\mapsto C_\ast \logicalp:=C \logicalp C^\dagger
    \end{split}\end{equation}
    as well as the pullback
    \begin{equation}\begin{split}
        {}^\ast:\Cl\times\Pauli&\rightarrow\bar\Pauli\\
        (C,P)&\mapsto C^\ast P:=C^\dagger PC
    \end{split}\end{equation}
    are covariant group homomorphisms in their second component, i.e., for a fixed Clifford unitary $C\in\Cl$ and two Pauli operators $\logicalp, \bar Q\in\bar\Pauli$ we have $C_\ast(\logicalp\cdot\bar Q)=C_\ast \logicalp\cdot C_\ast \bar Q$ and for $P,Q\in\Pauli$ we have $C^\ast(P\cdot Q)=C^\ast P\cdot C^\ast Q$.

    The pushforward is covariant in the first component, i.e., for two Clifford unitaries $C_1, C_2\in\Cl$ and $\logicalp\in\bar\Pauli$
    \begin{equation}
        (C_2\cdot C_1)_\ast \logicalp={C_2}_\ast ({C_1}_\ast \logicalp),
    \end{equation}
    while the pullback is contravariant in the first component, i.e., for two Clifford unitaries $C_1, C_2\in\Cl$ and $P\in\Pauli$ 
    \begin{equation}
        (C_2\cdot C_1)^\ast P= {C_1}^\ast ({C_2}^\ast P).
    \end{equation}
\end{proposition}

\begin{corollary}\label{cor:proper_paulis_preserved}
    The action of the Clifford group on the Pauli group by right resp.\ left conjugation preserves the set of proper Pauli operators.
\end{corollary}
\begin{proof}
    Let $C\in\Cl$ be a Clifford unitary and $P\in\Pauli$ a proper Pauli operator, i.e., $\ord(P)=2$ and $P\not=-\idmatrix$. Then $(C^\ast P)^2=C^\ast P^2=C^\ast\idmatrix=\idmatrix$ giving $\ord(C^\ast P)\leq 2$. Since $C^\ast(i^\kappa\idmatrix)=i^\kappa\idmatrix$ and $P\not=\pm\idmatrix$ its pullback $C^\ast P\not=\pm\idmatrix$ and therefore, the pullback of a proper Pauli operator $P$ is again proper. The proof for left conjugation, i.e., the pushforward is analogous.
\end{proof}

\begin{corollary}\label{cor:commutativity_paulis_preserved}
    The action of a Clifford unitary $C\in\Cl$ by right conjugation preserves the (anti-)commutativity properties of Pauli operators, i.e., for $P,Q\in\Pauli$ with $PQ=(-1)^\nu QP$ we have $C^\ast PC^\ast Q = (-1)^\nu C^\ast QC^\ast P$. The symplectic inner product $\omega$ is invariant under right conjugation, i.e., $\omega(\varrho(P),\varrho(Q))=\omega(\varrho(C^\ast P),\varrho(C^\ast Q))$. The analogous results are true for the left conjugation.
\end{corollary}
\begin{proof}
    Since right conjugation is a group homomorphism we get
    \begin{equation}\begin{split}
        (C^\ast P)(C^\ast Q) &= C^\ast (PQ)=C^\ast((-1)^\nu QP)\\
        &=(-1)^\nu (C^\ast Q)(C^\ast P).
    \end{split}\end{equation}
    This implies $\omega(\varrho(P),\varrho(Q))=\nu=\omega(\varrho(C^\ast P),\varrho(C^\ast Q))$. The proof for the left conjugation is analogous.
\end{proof}

 In general, we use the notations $\bar{\Pauli}$, $\logicalp$, $\logicalpauliop{}{\ipq}{}$, $\logicalpauliop{}{}{\ipq}$, etc.\ for the logical objects before the circuit and $\Pauli$, $P$, $\physicalpauliop{}{\ipq}{}$, $\physicalpauliop{}{}{\ipq}$, etc.\ for the physical objects after the circuit in order to highlight the direction of the mapping investigated -- physical-to-logical vs.\ logical-to-physical. If we need to differentiate between time steps like in the concatenation of circuits, we use a superscript $\Pauli^{(\itime)}$, $P^{(\itime)}$, $X_{\itq}^{(\itime)}$, $Z_{\itq}^{(\itime)}$, etc.\ for the time $\itime$, in order to not speak about the physical output of the first circuit as the logical input of the second circuit, which might lead to confusion. Therefore, in the case of the concatenation $C_2C_1$ of two circuits, the logical Pauli operators $\bar{\Pauli}$ correspond to $\Pauli^{(0)}$ at time $\itime=0$, while the time steps $i=1,2$ are after the corresponding circuit $C_i$. Fig.~\ref{fig:concatenation} shows the co- and contravariance of the pushforward and the pullback in the circuit model.

\begin{figure}
    \centering
    \includegraphics[width=\linewidth]{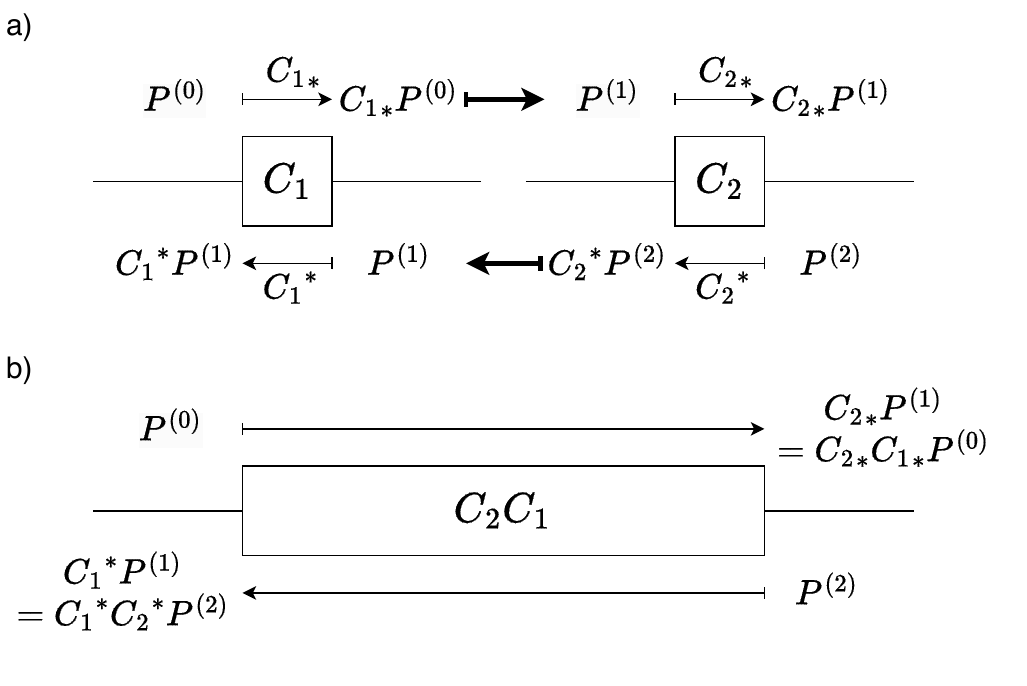}
    \caption{a) Given the two circuits $C_1$ and $C_2$ we can calculate the pushforward (above the circuits) and the pullback (below the circuits). b) The concatenation $C_2C_1$ of the circuits then corresponds to iteratively applying ${C_2}_\ast$ to $P^{(1)}={C_1}_\ast P^{(0)}$ for the pushforward resp.\ iteratively applying ${C_1}^\ast$ to $P^{(1)}={C_2}^\ast P^{(2)}$ for the pullback (see bold arrows in a)). This shows that the pushforward is naturally calculated in an iterative manner in the forward direction of time, while the pullback is naturally calculated iteratively in the backwards direction of time.}
    \label{fig:concatenation}
\end{figure}

\subsubsection{Appending vs.\ prepending gates}
The co- and contravariance of the pushforward resp.\ the pullback lead to different properties with respect to the appending or prepending of gates to a Clifford circuit, see also \cite{Gidney2021stimfaststabilizer, Winderl2023}.

So, if we consider a circuit $C=C_{\itime}\cdots C_2C_1$ having already calculated $C_\ast \logicalp$ and $C^\ast P$ and we now want to \emph{append} a Clifford operator $C_{\itime+1}$, i.e., consider $C_{\itime+1}C$, then in the case of the pushforward we only have to add $C_{\itime+1}$ and $C_{\itime+1}^\dagger$ to the \emph{outside} of the ``sandwich'' 
\begin{equation}
{C_{\itime+1}}_\ast( C_\ast \logicalp)=C_{\itime+1}(C_{\itime}\cdots C_1\logicalp C_1^\dagger \cdots C_{\itime}^\dagger) C_{\itime+1}^\dagger,
\end{equation}
while in the case of the pullback the adding of $C_{\itime+1}^\dagger$ and $C_{\itime+1}$ has to happen on the \emph{inside} of the ``sandwich'' 
\begin{equation}
C^\ast(C_{\itime+1}^\ast P)=C_1^\dagger \cdots C_{\itime}^\dagger (C_{\itime+1}^\dagger P C_{\itime+1}) C_{\itime} \cdots C_{1}.
\end{equation}
Therefore, the calculation of the pushforward is \emph{compatible with appending gates} in the sense that the appending can be naturally done in an iterative manner -- just ``sandwich'' from the outside with the gate to be appended. In contrast, the calculation of the pullback is a priori not compatible with appending gates in the sense that the appending cannot be naturally done in an iterative manner -- the gate to be appended has to be inserted into the ``sandwich''. The opposite is true for the prepending of a gate, i.e., we again consider the circuit $C=C_{\itime}\cdots C_2C_1$ but now we \emph{prepend} the gate $C_0$ to get $CC_0$. A similar calculation as above shows that prepending is naturally iterative for the pullback but not for the pushforward. 

In view of the two mappings logical to physical for the pushforward and physical to logical for the pullback the question of being naturally iterative corresponds to the question which Pauli operators under consideration remain or change by the appending or prepending of a gate. Considering the appending of a gate, the group of logical Pauli operators stays the same, while the group of physical Pauli operators gets moved one time step into the future. This is iterative for the pushforward, since the preimages, i.e., the logical operators, stay the same, while the resulting images, i.e., physical operators, have to be pushed into the future. For the pullback the physical operators are the preimages and these have to be transferred to the next time step, while the images stay in the unchanged logical Pauli group. Therefore, we have to prepend the action of the appended gate, which is not naturally iterative. Considering the prepending of a gate, the situation gets reversed, the logical Pauli group has to be transferred one time step backwards in time, prepending the action of the prepended gate, which is not naturally iterative, while the images of the pullback change, appending the action of the prepended gate, which is naturally iterative. 

\subsubsection{Circuit vs.\ inverse circuit}
Many publications using the Clifford tableau use the inverse Clifford tableau, which is the Clifford tableau of the inverse circuit $C^\dagger$, see \cite{Winderl2023, Gidney2021stimfaststabilizer}. Interchanging $C$ with its inverse $C^\dagger$ obviously interchanges the conjugation types:
\begin{align}
    (C^\dagger)_\ast P = (C^\dagger) P(C^\dagger)^\dagger =C^\dagger P C = C^\ast P\label{eq:pullback_pushforward_inverse1}\\
    \intertext{and vice versa}
    (C^\dagger)^\ast \logicalp = (C^\dagger)^\dagger \logicalp(C^\dagger) =C \logicalp C^\dagger = C_\ast \logicalp.\label{eq:pullback_pushforward_inverse2}
\end{align}
Which also shows that the pullback and the pushforward are inverse to each other:
\begin{equation}\label{eq:pullback_pushforward_inverse3}
    C^\ast(C_\ast \logicalp)=(C^\dagger)_\ast( C_\ast \logicalp)=(C^\dagger C)_\ast \logicalp = \idmatrix_\ast \logicalp = \logicalp
\end{equation}
as well as
\begin{equation}\label{eq:pullback_pushforward_inverse4}
    C_\ast(C^\ast P)=(C^\dagger)^\ast( C^\ast P) \overset{\bullet}{=}(C C^\dagger)^\ast P = \idmatrix^\ast P = P,
\end{equation}
note the reversed ordering of the two matrices at $\bullet$ because of the contravariance.

\begin{corollary}\label{cor:conjugation_isomorphism}
    In the situation of Prop.~\ref{prop:conjugation-functors} the two conjugation types $C^\ast:\Pauli\rightarrow\bar\Pauli$ and $C_\ast:\bar{\Pauli}\rightarrow\Pauli$ are group isomorphisms for each Clifford $C\in\Cl$.
\end{corollary}

We will come back to the concatenation of these choices, appending/prepending, circuit/inverse circuit and left/right conjugation in the section on co- and contravariance of the tableaus, see Fig.~\ref{fig:decision_tree}.

\subsubsection{The action of the elementary Clifford gates}
As is well known, the Clifford group is generated by the \emph{elementary Clifford gates}, the Hadamard gate $H_{\itq}$ and the $S$-gate $S_{\itq}$ on the $\ipq^\text{th}$ qubit and the $\CNOT$ gate $\CNOT[\icq,\itq]$ with the $\icq^\text{th}$ qubit as control and the $\itq^\text{th}$ qubit as target qubit, while the Pauli group is generated by $i\idmatrix, \logicalpauliop{}{\ipq}{}$ and $\logicalpauliop{}{}{\ipq}$. We can then easily calculate the action of the generators $C\in\{\CNOT[\icq,\itq], H_{\itq}, S_{\itq}\}$ of the Clifford group on the generators $\logicalp\in\{i\idmatrix, \logicalpauliop{}{\ipq}{}, \logicalpauliop{}{}{\ipq}\}$ of the Pauli group for the pushforward and analogously for the pullback. For the pushforward we get table~\ref{tab:generatorpushforward} of non-trivial pushforwards.\footnote{We use colors as an additional visual aid (logical operators $\logicalpauliop{}{j}{}, \logicalpauliop{}{}{k}$ and physical operators $\physicalpauliop{}{j}{}$ and $\physicalpauliop{}{}{k}$). While aiming at accessibility of the chosen colors, there is no information lost without them.}
\begin{table}[h!]\begin{tabular}{c||c|c|c}
    $C_\ast \logicalp$ & $\logicalpauliop{}{\itq}{}$ & $\logicalpauliop{}{}{\itq}$ & $\logicalpauliop{}{\icq}{}$ \\\hline\hline
    $\mathrm{CNOT}_{\icq\rightarrow \itq}$ & & $\physicalpauliop{}{}{\itq}\physicalpauliop{}{}{\icq}$ & $\physicalpauliop{}{\icq}{}\physicalpauliop{}{\itq}{}$ \\\hline
    $H_\itq$& $\physicalpauliop{}{}{\itq}$ & $\physicalpauliop{}{\itq}{}$ &    \\\hline
    $S_\itq$& $i\physicalpauliop{}{\itq}{}\physicalpauliop{}{}{\itq}$ & &   \\\hline
\end{tabular}
\caption{Table of non-trivial pushforwards, i.e., the empty cells as well as the cells of all generators $\logicalp$ that do not appear in the table correspond to trivial pushforwards $C_\ast i\idmatrix=i\idmatrix$, $C_\ast {\color{xlog}{\bar X_{\ipq}}}={\color{xphys}{X_{\ipq}}}$, $C_\ast {\color{zlog}{\bar Z_{\ipq}}}={\color{zphys}{Z_{\ipq}}}$.}\label{tab:generatorpushforward}
\end{table}

Since the $\CNOT$- and the $H$-gate are self inverse, the different conjugation types do not matter in these cases, i.e., $H_\ast=H^\ast$ and ${\CNOT}_\ast=\CNOT^\ast$. Hence, the only difference we find in the two conjugation tables is the minus sign for the $S$-gate (and of course the logical-to-physical vs.\ physical-to-logical view), see table~\ref{tab:generatorpullback}. 

\begin{table}[h!]\begin{tabular}{c||c|c|c}
    $C^\ast P$ & $\physicalpauliop{}{\itq}{}$ & $\physicalpauliop{}{}{\itq}$ & $\physicalpauliop{}{\icq}{}$ \\\hline\hline
    &&&\\[-2.3ex]
    $\mathrm{CNOT}_{\icq\rightarrow \itq}$ & & $\logicalpauliop{}{}{\itq}\logicalpauliop{}{}{\icq}$ & $\logicalpauliop{}{\icq}{}\logicalpauliop{}{\itq}{}$ \\\hline
    &&&\\[-2.3ex]
    $H_\itq$& $\logicalpauliop{}{}{\itq}$ & $\logicalpauliop{}{\itq}{}$ &    \\\hline
    &&&\\[-2.3ex]
    $S_\itq$& $-i\logicalpauliop{}{\itq}{}\logicalpauliop{}{}{\itq}$ & &   \\\hline
\end{tabular}
\caption{Table of non-trivial pullbacks, i.e., the empty cells as well as the cells of all generators $P$ that do not appear in the table correspond to trivial pullbacks $C^\ast i\idmatrix=i\idmatrix$, $C^\ast {\color{xphys}{X_{\ipq}}}={\color{xlog}{\bar X_{\ipq}}}$, $C^\ast {\color{zphys}{Z_{\ipq}}}={\color{zlog}{\bar Z_{\ipq}}}$. }\label{tab:generatorpullback}
\end{table}

The action of the elementary Clifford gates on the generators of the Pauli group of course extends to the whole Clifford group, since conjugation is compatible with the multiplication of Pauli operators. For the pushforward, i.e., the left conjugation, we get the actions in Fig.~\ref{fig:elementarypushforward} for the pushforward $C_\ast \logicalp$ of an arbitrary Pauli operator $\logicalp$ represented as $\logicalflowp$ by the following
\begin{proposition}[Sec.~III of \cite{dehaenedemoor2003}]\label{prop:elementarypushforward}
    Let $\logicalp=i^\ep\bm{\logicalpauliop{}{ }{}}^{\bmlogicalxi}\bm{\logicalpauliop{}{}{ }}^{\bmlogicalzeta}\in\bar\Pauli$ be an arbitrary Pauli operator with ``vector form'' $\varrho(\logicalp)=\logicalflowp$ and let $C\in\{\CNOT[\icq,\itq], H_{\itq}, S_{\itq}\}$ be an elementary Clifford gate. Then $\varrho(C_{\ast}\logicalp)$ in the three cases is
    \begin{equation}\begin{split}
        \varrho((\CNOT[\icq,\itq])_\ast \logicalp)&=(\ep\vert \bmlogicalxi\oplus_2 \logicalxi[\icq]\bm{e}_{\itq}\vert \bmlogicalzeta\oplus_2 \logicalzeta[\itq]\bm{e}_{\icq}) \\
        \varrho((H_{\itq})_\ast \logicalp)&=(\ep\oplus_4 2\logicalxi[\itq]\logicalzeta[\itq] \vert \bmlogicalxi\oplus_2        
                                (\logicalxi[\itq]\oplus_2\logicalzeta[\itq])\bm{e}_{\itq} \\
        &\hspace{21mm} \vert \bmlogicalzeta\oplus_2 (\logicalxi[\itq]\oplus_2\logicalzeta[\itq])\bm{e}_{\itq})\\
        \varrho((S_{\itq})_\ast \logicalp)&=(\ep\oplus_4 \logicalxi[\itq]\vert \bmlogicalxi\vert \bmlogicalzeta\oplus_2 \logicalxi[\itq]\bm{e}_{\itq}).\\ 
    \end{split}\end{equation}
\end{proposition}
\begin{proof} The pushforwards of the generators $\logicalpauliop{}{\ipq}{}$ and $\logicalpauliop{}{}{\ipq}$ by $C$ were already given in table~\ref{tab:generatorpushforward} (see \cite[Sec.~10.5.2]{NielsenChuang2011}).

Since the pushforward is a group homomorphism we get
\begin{equation}\begin{split}
        C_\ast \logicalp &= C_\ast \left(i^\ep\prod_{\ipq}\logicalpauliop{}{\ipq}{}^{\logicalxi[\ipq]}\prod_{\ipq}\logicalpauliop{}{}{\ipq}^{\logicalzeta[\ipq]}\right)\\
        &=C_\ast(i^\ep\idmatrix)\prod_{\ipq}(C_\ast \logicalpauliop{}{\ipq}{})^{\logicalxi[\ipq]}\prod_{\ipq}(C_\ast \logicalpauliop{}{}{\ipq})^{\logicalzeta[\ipq]}.
\end{split}\end{equation}
Since most of the actions of $C_\ast$ are trivial, we only have to look at the non-trivial ones of $C$ in the above table and rearrange the product to get the standard form of Eq.~\eqref{eq:standard_form_pauli}. In the $C=\CNOT[\icq,\itq]$ case that means:
\begin{equation}\begin{split}
    {\CNOT}_\ast \logicalp & = i^\ep \blackphysicalpauliop{}{1}{}^{\logicalxi[1]}\cdots (C_\ast \logicalpauliop{}{\icq}{})^{\logicalxi[\icq]}\cdots \blackphysicalpauliop{}{\itq}{}^{\logicalxi[\itq]}\cdots \blackphysicalpauliop{}{\npq}{}^{\logicalxi[\npq]}\\  
    &\phantom{=i^\ep}\cdot \blackphysicalpauliop{}{}{1}^{\logicalzeta[1]}\cdots \blackphysicalpauliop{}{}{\icq}^{\logicalzeta[\icq]}\cdots (C_\ast \logicalpauliop{}{}{\itq})^{\logicalzeta[\itq]}\cdots \blackphysicalpauliop{}{}{\npq}^{\logicalzeta[\npq]}\\
    &= i^\ep \blackphysicalpauliop{}{1}{}^{\logicalxi[1]}\cdots (\blackphysicalpauliop{}{\icq}{}\blackphysicalpauliop{}{\itq}{})^{\logicalxi[\icq]}\cdots \blackphysicalpauliop{}{\itq}{}^{\logicalxi[\itq]}\cdots \blackphysicalpauliop{}{\npq}{}^{\logicalxi[\npq]}\\  
    &\phantom{=i^\ep}\cdot \blackphysicalpauliop{}{}{1}^{\logicalzeta[1]}\cdots \blackphysicalpauliop{}{}{\icq}^{\logicalzeta[\icq]}\cdots (\blackphysicalpauliop{}{}{\itq}\blackphysicalpauliop{}{}{\icq})^{\logicalzeta[\itq]}\cdots \blackphysicalpauliop{}{}{\npq}^{\logicalzeta[\npq]}\\
    &= i^\ep \blackphysicalpauliop{}{1}{}^{\logicalxi[1]}\cdots \blackphysicalpauliop{}{\icq}{}^{\logicalxi[\icq]}\cdots \blackphysicalpauliop{}{\itq}{}^{\logicalxi[\icq]}\blackphysicalpauliop{}{\itq}{}^{\logicalxi[\itq]}\cdots \blackphysicalpauliop{}{\npq}{}^{\logicalxi[\npq]}\\  
    &\phantom{=i^\ep}\cdot \blackphysicalpauliop{}{}{1}^{\logicalzeta[1]}\cdots \blackphysicalpauliop{}{}{\icq}^{\logicalzeta[\icq]}\blackphysicalpauliop{}{}{\icq}^{\logicalzeta[\itq]}\cdots \blackphysicalpauliop{}{}{\itq}^{\logicalzeta[\itq]}\cdots \blackphysicalpauliop{}{}{\npq}^{\logicalzeta[\npq]}\\
    &= i^\ep \blackphysicalpauliop{}{1}{}^{\logicalxi[1]}\cdots \blackphysicalpauliop{}{\icq}{}^{\logicalxi[\icq]}\cdots \blackphysicalpauliop{}{\itq}{}^{\logicalxi[\icq]\oplus_2\logicalxi[\itq]}\cdots \blackphysicalpauliop{}{\npq}{}^{\logicalxi[\npq]}\\      
    &\phantom{=i^\ep}\cdot \blackphysicalpauliop{}{}{1}^{\logicalzeta[1]}\cdots \blackphysicalpauliop{}{}{\icq}^{\logicalzeta[\icq]\oplus_2\logicalzeta[\itq]}\cdots \blackphysicalpauliop{}{}{\itq}^{\logicalzeta[\itq]}\cdots \blackphysicalpauliop{}{}{\npq}^{\logicalzeta[\npq]}\\
    &= i^\ep \bm{\blackphysicalpauliop{}{ }{}}^{\bmlogicalxi\oplus_2\logicalxi[\icq]\bm{e_{\itq}}}\bm{\blackphysicalpauliop{}{}{ }}^{\bmlogicalzeta\oplus_2\logicalzeta[\itq]\bm{e_{\icq}}}.
\end{split}\end{equation}
Since we only rearranged Pauli operators of the same type, the phase stays unchanged, we get an additional $X_{\itq}$ if we had an $X_{\icq}$, which corresponds to addition modulo $2$ of the exponents since $X$ has order $2$, as well as an additional $Z_{\icq}$ if we had an $Z_{\itq}$.

A similar calculation for $H_{\itq}$ leaves most of $\logicalp$ unchanged, but $X_{\itq}^{\xi_{\itq}}$ replaced by $Z_{\itq}^{\xi_{\itq}}$ and $Z_{\itq}^{\zeta_{\itq}}$ replaced by $X_{\itq}^{\zeta_{\itq}}$. Since we have to swap them back in order to get back to the standard form we get an additional $-1$ in the phase iff both operators are present, i.e., iff $\logicalxi[\itq]=\logicalzeta[\itq]=1$. This phase correction is exactly the $\oplus_4 2\logicalxi[\itq]\logicalzeta[\itq]$ of the proposition, while the addition of $(\logicalxi[\itq]\oplus_2\logicalzeta[\itq])\bm{e_{\itq}}$ to $\bmlogicalxi$ and $\bmlogicalzeta$ leads to canceling the existing $\logicalxi[\itq]$ (resp.\ $\logicalzeta[\itq]$) and replacing it by $\logicalzeta[\itq]$ (resp.\ $\logicalxi[\itq]$).

For the $S$-gate $S_{\itq}$ the only non-trivial action is $X_{\itq}\mapsto iX_{\itq}Z_{\itq}$. Hence, $\logicalp$ and ${S_\itq}_\ast\logicalp$ differ iff $\logicalzeta[\itq]=1$. Rearranging ${S_{\itq}}_\ast \logicalp$ to standard form in this case just adds a $1$ to the phase exponent $\ep$ as well as adding an additional $Z_{\itq}$, giving raise to $\oplus_4\logicalxi[\itq]$ in the phase component and $\oplus_2\logicalxi[\itq]\bm{e}_{\itq}$ in the last component of $\varrho({S_\itq}_\ast\logicalp)$.

Therefore, we proved the proposition, see Fig.~\ref{fig:elementarypushforward} for the visualization of these actions as column operations.
\end{proof}

\begin{figure}[h!]
    \centering
    \includegraphics[width=0.75\linewidth]{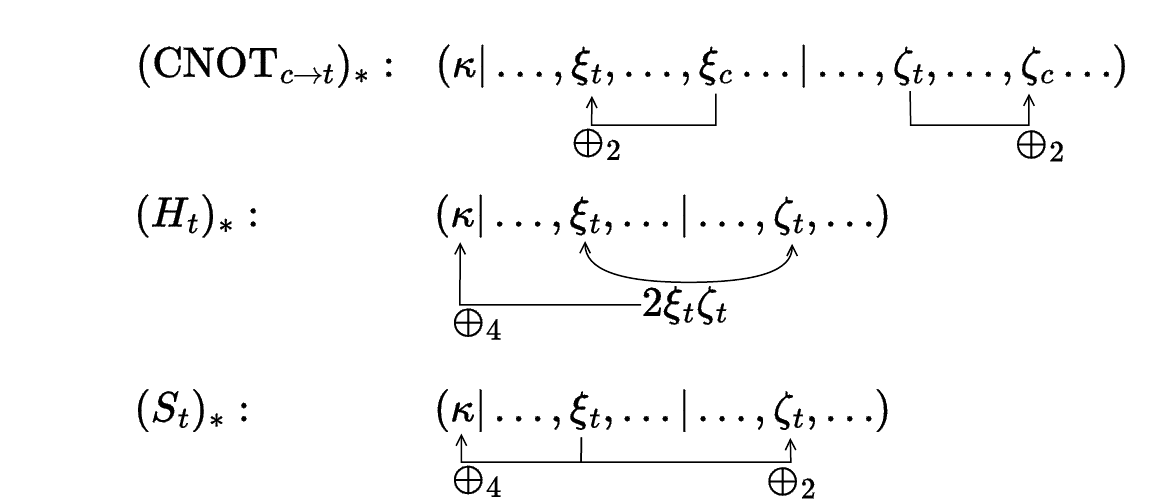}
    \caption{The action of the elementary Clifford gates on an arbitrary element of the Pauli group represented as $\flowp$. For the $\CNOT$-gate, there is no action on the phase exponent~$\ep$. For the $H$-gate we need an additional $(-1)$ in the phase, iff both $X_{\itq}$ \emph{and} $Z_{\itq}$ occur in the Pauli operator. For the $S$-gate we only need to add a $1$ to the phase exponent to account for the factor $i$ iff $X_{\itq}$ occurs in the operator. See Prop.~\ref{prop:elementarypushforward} for the whole computation.}
    \label{fig:elementarypushforward}
\end{figure}
 
Because of the iterative property of appending gates in the case of the pushforward, we can directly interpret the actions of Fig.~\ref{fig:elementarypushforward} in the following way. Let $C$ be a Clifford circuit, $C_{\itime+1}$ an elementary Clifford gate, which we want to append to $C$, and let $\flowp$ be the ``vector form'' of $C_\ast \logicalp$, the action of the Clifford circuit on an arbitrary Pauli operator $\logicalp$. The actions of Fig.~\ref{fig:elementarypushforward} now directly give the ``vector form'' of $(C_{\itime+1}C)_\ast \logicalp$.

Of course, we can write down the action of the pullback of the elementary Clifford gates like in Fig.~\ref{fig:elementarypushforward} for an arbitrary Pauli operator $P$, the only difference being the minus sign for $(S_{\itq})^\ast$. But since the pullback is not iterative for appending gates, we cannot interpret these actions as giving the action of $(C_{\itime+1}C)^\ast P=C^\ast(C_{\itime+1}^\ast P)$, since we would need to calculate $Q=C_{\itime+1}^\ast P$ and plug this into the expression $C^\ast Q$ (see Fig.~\ref{fig:concatenation} with $C_2=C_{\itime+1}$ and $C_1=C$). In order to give a nice way of calculating the pullback for appending a gate we consider the tableau form which corresponds to choosing a nice generator set of the Pauli group (see Cor.~\ref{cor:Flow_append_iterative}).

\subsection{Flow tableau and Clifford tableau}
Our next step is to define the Clifford resp.\ the flow tableau as representations of the action by right resp.\ left conjugation in a ``matrix form''. The reasoning follows the same lines as representing a linear map $\psi:V_1\rightarrow V_2$ by a matrix after choosing (ordered) bases $\mathcal{B}_i$ for the vector spaces $V_i$, i.e., if $\varphi_{\mathcal{B}_1}(\bm{v_1})$ is the row vector representing the input $\bm{v_1}\in V_1$ in the basis $\mathcal{B}_1$ then $ \varphi_{\mathcal{B}_1}(\bm{v_1})\cdot\rowcolrestr[{\mathcal{B}_1}]{M(\psi)}[{\mathcal{B}_2}]$ is the row vector representing the output $\psi(\bm{v_1})\in V_2$ in the basis $\mathcal{B}_2$. Note that the notation $\rowcolrestr[{\mathcal{B}_1}]{M(\psi)}[{\mathcal{B}_2}]$ is chosen such that the subscripts of the vector and the matrix match if the multiplication is in the correct order, i.e., row vector $\cdot$ matrix. These facts can be nicely summarized in the following commutative diagram
\begin{equation}\label{eq:cd_linear_map}
\begin{tikzcd}[column sep=1cm, every cell/.style={inner xsep=1ex, inner ysep=0.85ex}]
    V_1 \arrow[d, "\varphi_{\mathcal{B}_1}"'] \arrow[rd, phantom, "\circlearrowleft", shift left=0ex] \arrow[r, "\psi"] & V_2 \arrow[d, "\varphi_{\mathcal{B}_2}"] \\
    K^{n}\arrow[r, "\cdot {}_{\mathcal{B}_1}M(\psi)_{\mathcal{B}_2}"']& K^{m}
\end{tikzcd}
\end{equation}
where the $\circlearrowleft$ in the middle signifies the commutativity of the diagram, i.e., going first down and then to the right, i.e., $\varphi_{\mathcal{B}_1}(\bm{v_1})\cdot\rowcolrestr[{\mathcal{B}_1}]{M(\psi)}[{\mathcal{B}_2}]$, is the same as going first to the right and then down, i.e. $\varphi_{\mathcal{B}_2}(\psi(\bm{v_1}))$. Since we chose to use row vectors and the multiplication of a row vector from the left to a matrix, the matrix is an $n\times m$-matrix. Recall that the matrix has to have the rows $\varphi_{\mathcal{B}_2}(\psi(\bm{b_i}))$ for $1\leq i\leq n$, i.e., the $i^{\text{th}}$ row is the image of the $i^{\text{th}}$ basis vector $\bm{b_i}$ of $\mathcal{B}_1$ represented in the basis $\mathcal{B}_2$. Therefore, the rows can be labeled by the basis vectors of the basis $\mathcal{B}_1$, whereas the columns can be labeled by the basis vectors of the basis $\mathcal{B}_2$.

Now -- mimicking the choice of a basis -- fix the generators $g_{\ipq}$ of the Pauli group $\Pauli$ as $g_0:= i\idmatrix$, $g_{\ipq}:= X_{\ipq}$ for $1\leq\ipq\leq\npq$ and $g_{\ipq}:= Z_{\ipq-\npq}=Z_{\ipqcheck}$ for $1\leq{\ipqcheck}:=\ipq-\npq\leq\npq$ and analogously for $\bar\Pauli$. Note that these generators are independent, i.e., we can not remove any $g_{\ipq}$ without losing the generating property, and symplectic, i.e., for each $\ipq$ the generators $g_{\ipq}=X_{\ipq}$ and $g_{\ipqhat}=Z_{\ipq}$ with $\ipqhat:=\ipq+\npq$ anticommute and all other pairs of generators commute. Note, that mapping these generators to their presentations $\varrho(g_j)$ in the group $G=\reprpauli$ yields vector like presentations corresponding to the canonical "vectors": $\varrho(g_0)=\varrho(i\idmatrix)=(1\vert\bm{0}\vert\bm{0})$, $\varrho(g_j)=\varrho(X_j)=(0\vert\bm{e}_j\vert\bm{0})$ and $\varrho(g_j)=\varrho(Z_{\ipqcheck})=(0\vert\bm{0}\vert\bm{e}_j)$.

\begin{definition}
    Define the \emph{Clifford tableau} $T(C)$ and the \emph{flow tableau} $F(C)$ of a Clifford unitary $C\in\Cl$ with respect to the presentation $\varrho$ by 
    \begin{equation}\begin{split}
    T(C)&:=\left(\varrho(C_\ast(\logicalp_{\ipq}))\right)_{0\leq\ipq\leq2\npq}\in G^{2\npq+1}\\
        &\phantom{:}=\begin{pNiceArray}{c}[first-col]
    \Block{1-1}{i\idmatrix}&\varrho(C_\ast(i\idmatrix))\\\hline
        &\\[-7pt]
        \Block[color=xlog, draw=white, line-width=0pt, rounded-corners]{1-1}{\bar{X}_{\ipq}}&\varrho(C_\ast \logicalpauliop{}{\ipq}{})\\
        &\\[-7pt]\hline
        &\\[-7pt]
        \Block[color=zlog, draw=white, line-width=0pt, rounded-corners]{1-1}{\bar{Z}_{\ipq}}& \varrho(C_\ast \logicalpauliop{}{}{\ipq})\\
        &\\[-7pt]
    \end{pNiceArray}=
    \begin{pNiceArray}{c|c|c}[first-row,first-col]
        & \eta
        & \Block[color=xphys, draw=white, line-width=0pt, rounded-corners]{1-1}{\bm{X}} 
        & \Block[color=zphys, draw=white, line-width=0pt, rounded-corners]{1-1}{\bm{Z}} \\
        \Block{1-1}{i\idmatrix}&1&\bm0&\bm0\\\hline
        &\\[-7pt]
        \Block[color=xlog, draw=white, line-width=0pt, rounded-corners]{1-1}{\bar{X}_{\ipq}}&\eta_{\ipq}&{\color{xphys}{\bm{x}_{\ipq}}}&{\color{zphys}{\bm{z}_{\ipq}}}\\
        &\\[-7pt]\hline
        &&\\[-7pt]
        \Block[color=zlog, draw=white, line-width=0pt, rounded-corners]{1-1}{\bar{Z}_{\ipq}}& \eta_{\ipqhat}&{\color{xphys}{\bm{x}_{\ipqhat}}}&{\color{zphys}{\bm{z}_{\ipqhat}}}\\[5pt]
    \end{pNiceArray}
    \end{split}\end{equation}
    resp.\ 
    \begin{equation}\begin{split}
    F(C)&:=\left(\varrho(C^\ast(P_{\ipq}))\right)_{0\leq\ipq\leq2\npq}\in G^{2\npq+1}\\
    &\phantom{:}=\begin{pNiceArray}{c}[first-col]
    \Block{1-1}{i\idmatrix}&\varrho(C^\ast(i\idmatrix))\\\hline
        &\\[-7pt]
        \Block[color=xphys, draw=white, line-width=0pt, rounded-corners]{1-1}{{X}_{\ipq}}&\varrho(C^\ast {\color{xphys}{X}_{\ipq}})\\
        &\\[-7pt]\hline
        &\\[-7pt]
        \Block[color=zphys, draw=white, line-width=0pt, rounded-corners]{1-1}{{Z}_{\ipq}}& \varrho(C^\ast {\color{zphys}{Z}_{\ipq}})\\[5pt]
    \end{pNiceArray}=
    \begin{pNiceArray}{c|c|c}[first-row,first-col]
        & \ep
        & \Block[color=xlog, draw=white, line-width=0pt, rounded-corners]{1-1}{\bm{\bar{X}}} 
        & \Block[color=zlog, draw=white, line-width=0pt, rounded-corners]{1-1}{\bm{\bar{Z}}} \\
        \Block{1-1}{i\idmatrix}&1&\bm0&\bm0\\\hline
        &&\\[-7pt]
        \Block[color=xphys, draw=white, line-width=0pt, rounded-corners]{1-1}{{X}_{\ipq}}&\ep_{\ipq}&\bmlogicalxi[\ipq]&\bmlogicalzeta[\ipq]\\
        &&\\[-7pt]\hline
        &\\[-7pt]
        \Block[color=zphys, draw=white, line-width=0pt, rounded-corners]{1-1}{{Z}_{\ipq}}& \ep_{\ipqhat}&\bmlogicalxi[\ipqhat]&\bmlogicalzeta[\ipqhat]\\
        &&\\[-7pt]
    \end{pNiceArray},
    \end{split}\end{equation}
    where $\ipqhat:=\ipq+\npq$.
    
    These definitions give the maps $T:\Cl\rightarrow G^{2\npq+1}$ and $F:\Cl\rightarrow  G^{2\npq+1}$.

    We denote the phaseless versions of the tableaus $F(C)$ resp.\ $T(C)$, omitting the zeroth row and the zeroth column, and the corresponding maps $F$ resp.\ $T$ by a tilde $\tilde{F}$ resp.\ $\tilde{T}$.
\end{definition}
\begin{remark}
    Since the action of $C$ by conjugation on $g_0=i\idmatrix$ is always trivial, the first row contains no information. Nevertheless, we will include this row for ease of later calculations (compare the bar-notation of Thm.~1 in \cite{dehaenedemoor2003}).

    In the Clifford tableau we use the notation ${\color{xphys}{\bm{x}}}$ and ${\color{zphys}{\bm{z}}}$ for the exponents of $\bm{\physicalpauliop{}{ }{}}$ and $\bm{\physicalpauliop{}{}{ }}$ in order to be more comparable to the original Clifford tableau in \cite{AaronsonGottesman2004} as well as to distinguish the entries of the Clifford tableau from those of the flow tableau. 
    
    The remaining difference to the original Clifford tableau is now only the different choice of presentation $\varrho$ giving the phase exponents $\eta_{\ipq}$ instead of the sign of the phase in the $XYZ$-notation $r$, see Eq.~\eqref{eqn:XYZ-standard-form} or Prop.~\ref{prop:randrho}. 

    The flow tableau gives rise to the flow labels $\labelname[P]^C$ by switching from the ``vector'' form to the label form, writing the two labels $\labelname[\physicalpauliop{}{j}{}]^C$ and $\labelname[\physicalpauliop{}{}{j}]^C$ above and below the $j^\text{th}$ qubit line.
\end{remark}

\begin{remark}
    Storing a Clifford or a flow tableau in a binary way, i.e., storing the $\mathbb{Z}_4$-components as 2-bit binary numbers results in $(2+2n)\cdot(2n+1)=\mathcal{O}(n^2)$ bits.
\end{remark}

The fact that replacing $C\in\Cl$ by its inverse $C^\dagger$ interchanges the two conjugation types, see Eqs.~\eqref{eq:pullback_pushforward_inverse1}-\eqref{eq:pullback_pushforward_inverse4}, directly implies

\begin{corollary}\label{cor:tableau_inverse_clifford}
    For an arbitrary Clifford unitary $C\in\Cl$ the flow tableau of the inverse $C^\dagger$ is  $F(C^\dagger)=T(C)$ and its Clifford tableau is $T(C^\dagger)=F(C)$.
\end{corollary}

Therefore, the flow tableau is an inverse Clifford tableau in the literature, see e.g., \cite{Gidney2021stimfaststabilizer}.

For further reference we state the flow and the Clifford tableaus of the elementary Clifford gates:
\begin{corollary}\label{cor:tableaus_elementary_cliffords}
\begin{equation}\begin{array}{lll}
    F(H_{\itq})\!\!&=T(H_{\itq})\!\!&= \!\!\left(\begin{array}{c|c}
         \scriptstyle 1&\scriptstyle  \bm{0} \\\hline
         &\\[-13pt]
         \scriptstyle \bm{0}^T&\scriptstyle E_{\itq,\itqhat}
    \end{array}\right)\\[1pt]
    F(S_{\itq})\!\!&&=\!\!\left(\begin{array}{c|c}
         \scriptstyle 1&\scriptstyle  \bm{0} \\\hline
         &\\[-13pt]
         \scriptstyle 3\bm{e_\itq}^T&\scriptstyle D_{\itq,\itqhat} 
    \end{array}\right)\\[1pt]
    &\phantom{=}\;T(S_{\itq})\!\!&=\!\!\left(\begin{array}{c|c}
         \scriptstyle 1&\scriptstyle  \bm{0} \\\hline
         &\\[-13pt]
         \scriptstyle \bm{e_\itq}^T&\scriptstyle D_{\itq,\itqhat} 
    \end{array}\right)\\[1pt]
        F(\CNOT[\icq,\itq])\!\!\!\!&=T(\CNOT[\icq,\itq])\!\!\!\!&=\!\!
         \left(\begin{array}{c|c|c}
             \scriptstyle 1&\scriptstyle \bm{0}& \scriptstyle \bm{0} \\\hline
             &\\[-13pt]
             \scriptstyle \!\bm{0}^T\!\!&\scriptstyle  \!D_{\icq,\itq}\!&\scriptstyle  0 \\\hline
             &\\[-13pt]
             \scriptstyle \!\bm{0}^T\!\!&\scriptstyle  0 &\!\scriptstyle  D_{\itq,\icq}\!
        \end{array}\right)
\end{array}\end{equation}
    Here the square matrix $D_{i,j}=\idmatrix+\vert\bm{e_i}\rangle\langle\bm{e_j}\vert$ of the appropriate sizes $\npq$ resp.\ $2\npq$ is the identity matrix with an additional $1$ at the $(i,j)^{\text{th}}$ entry, while the square matrix $E_{i,j}=\idmatrix+\vert\bm{e_i}\rangle\langle\bm{e_i}\vert+\vert\bm{e_i}\rangle\langle\bm{e_j}\vert+\vert\bm{e_j}\rangle\langle\bm{e_i}\vert+\vert\bm{e_j}\rangle\langle\bm{e_j}\vert$ of size $2\npq$ is the identity matrix with the $(i,i)^{\text{th}}$ and $(j,j)^{\text{th}}$ elements on the diagonal set to $0$ and additional $1$'s at the $(i,j)^{\text{th}}$ and $(j,i)^{\text{th}}$ entry. 
\end{corollary}
\begin{proof}
    This just amounts to plugging in the results of the tables~\ref{tab:generatorpullback} and \ref{tab:generatorpushforward} into the Clifford and the flow tableau definition.
\end{proof}

Recall the symplectic inner product $\omega:\mathbb{F}_2^{2\npq}\times\mathbb{F}_2^{2\npq}\rightarrow \mathbb{F}_2$ from the subsection on (anti-)commutativity of Pauli operators which is defined by $\Omega=\left(\begin{array}{cc}
    0 & \idmatrix \\
    \idmatrix & 0
\end{array}\right)$ and extended to $G$ by disregarding the phase.

\begin{definition}
    A ``matrix'' $M$ in $G^{2\npq+1}$ is \emph{symplectic} if the phaseless matrix $\tilde{M}$ obtained from $M$ by removing the $0^\text{th}$ column and the $0^\text{th}$ row is symplectic, i.e, $\tilde{M}\Omega \tilde{M}^T=\Omega$.
    
    A ``matrix'' $M$ in $G^{2\npq+1}$ is a \emph{Clifford tableau} resp.\ a \emph{flow tableau} if there exists a Clifford unitary $C\in\Cl$ with $T(C)=M$ resp.\ $F(C)=M$.
    
    A ``matrix'' $M$ in $G^{2\npq+1}$ is a \emph{proper} tableau if $M$ is a Clifford tableau or, equivalently, if $M$ is a flow tableau.
\end{definition}

\begin{remark}
    We sometimes write $M\Omega M^T=\Omega$ for the symplectic property of a ``matrix'' $M\in G^{2\npq+1}$, but note that this matrix equation is not defined in the strict sense and we actually have to use the phaseless matrix $\tilde{M}$ of $M$.
\end{remark}

\begin{remark}
Since $\Cl$ is a group and hence contains $C^\dagger$ for every $C\in\Cl$ Cor.~\ref{cor:tableau_inverse_clifford} implies that a ``matrix'' $M$ is a Clifford tableau iff it is a flow tableau.    
\end{remark}

\begin{proposition}\label{prop:proper_tableau} (see Fact 35 in \cite{GossetGrierKerznerSchaeffer2024}) 
    A ``matrix'' $M\in G^{2\npq+1}$ is proper iff $M$ is symplectic, the $0^\text{th}$ row $\bm{m_0}$ is the standard ``vector'' $\bm{e_0}=(1\vert\bm{0}\vert\bm{0})\in Z(G)$ and the other rows $\bm{m_{\ipq}}$ of $M$ are proper elements in $G$.
\end{proposition}

\begin{proof}
    ``$\Leftarrow$'' By Cor.~\ref{cor:proper_paulis_preserved} the conjugation of a proper Pauli operator by a Clifford unitary gives a proper Pauli operator again. Since $\logicalpauliop{}{\ipq}{}$ and $\logicalpauliop{}{}{\ipq}$ are proper (i.e., Hermitian and $\not=\pm\idmatrix$), their pullbacks via a Clifford $C\in\Cl$ are proper and therefore, all rows but the $0^\text{th}$ of $F(C)$ are proper. The generators $g_{\ipq}=\logicalpauliop{}{\ipq}{}$ and $g_{\ipqhat}=\logicalpauliop{}{}{\ipq}$, $\ipq=1,\dotsc,\npq$, $\ipqhat=\npq+\ipq$, are symplectic, since almost all pairs of generators $g_{\ipq}, g_{\ipq'}$ commute but $g_{\ipq}, g_{\ipqhat}$ which anticommute, i.e., in ``matrix'' form we get
    \begin{equation}
        \left(\begin{array}{c|c|c}
            1&\bm{0} & \bm{0} \\\hline
             \bm{0}& \idmatrix &0\\\hline
             \bm{0} &0&\idmatrix
        \end{array}\right)\cdot \left(\begin{array}{c|c}
    0 & \idmatrix \\\hline
    \idmatrix & 0
\end{array}\right)\cdot\left(\begin{array}{c|c|c}
            1&\bm{0} & \bm{0} \\\hline
             \bm{0}& \idmatrix &0\\\hline
             \bm{0} &0&\idmatrix
        \end{array}\right)^T=\left(\begin{array}{c|c}
    0 & \idmatrix \\\hline
    \idmatrix & 0
\end{array}\right).
    \end{equation}
    By Cor.~\ref{cor:commutativity_paulis_preserved} these (anti-)commutativity properties are preserved under pullback, therefore, 
    \begin{equation}
            \begin{pNiceArray}{c|c|c}[first-row,first-col]
        & \ep
        & \Block[color=xlog, draw=white, line-width=0pt, rounded-corners]{1-1}{\bm{\bar{X}}} 
        & \Block[color=zlog, draw=white, line-width=0pt, rounded-corners]{1-1}{\bm{\bar{Z}}} \\
        \Block{1-1}{i\idmatrix}&1&\bm0&\bm0\\\hline
        \Block[color=xphys, draw=white, line-width=0pt, rounded-corners]{1-1}{{X}_{\ipq}}&\ep_{\ipq}&\bmlogicalxi[\ipq]&\bmlogicalzeta[\ipq]\\\hline
        \Block[color=zphys, draw=white, line-width=0pt, rounded-corners]{1-1}{
        {Z}_{\ipq}}& \ep_{\ipqhat}&\bmlogicalxi[\ipqhat]&\bmlogicalzeta[\ipqhat]\\
    \end{pNiceArray}  \cdot\left(\begin{array}{c|c}
         0&\idmatrix  \\\hline
         \idmatrix& 0
    \end{array}\right) \cdot \begin{pNiceArray}{c|c|c}
        1&\bm0&\bm0\\\hline
        \ep_{\ipq}&\bmlogicalxi[\ipq]&\bmlogicalzeta[\ipq]\\\hline
         \ep_{\ipqhat}&\bmlogicalxi[\ipqhat]&\bmlogicalzeta[\ipqhat]\\
    \end{pNiceArray}^T=\left(\begin{array}{c|c}
         0&\idmatrix  \\\hline
         \idmatrix& 0
    \end{array}\right).
    \end{equation}
    Note that the $0^\text{th}$ row $\bm{m_0}$ of the flow tableau corresponds to the ignored generator $g_0=i\idmatrix$ with $\varrho(g_0)=\bm{e_0}$. $g_0$ commutes with all other generators $g_j$ and having order $4$ is obviously not proper. Therefore, every flow tableau $M=F(C)$ is symplectic, has $0^\text{th}$ row $\bm{e_0}$ and all other rows are proper.
     
    ``$\Rightarrow$'' Let $M\in G^{2\npq+1}$ be symplectic with $0^\text{th}$ row $\bm{e_0}$ and all other rows $\bm{m_j}$ proper. This is essentially Thm.~4 from \cite{dehaenedemoor2003}, which constructs a concrete Clifford circuit such that $T(C)=M$, see also \cite{Winderl2023} for a version using only the connectivity of a given hardware.
\end{proof}

Since the flow tableau $F(C)$ contains the images of a set of independent generators $g_{\ipq}$ of the Pauli group $\Pauli$ under the pullback $C^\ast$ we can calculate the pullback $C^\ast P$ of an arbitrary Pauli operator $P\in\Pauli$ by reading off appropriate pullbacks of the generators $g_{\ipq}$ in the flow tableau $F(C)$ and multiplying them. This can be represented as a ``vector-matrix''-multiplication induced by $\circledast$:
\begin{definition}
    Let $\bm\varrho=\paulivec{\varrho_{0}}{\varrho_{1}\dotsc\varrho_{\npq}}{\varrho_{\npq+1}\dotsc\varrho_{2\npq}}\in G$ be a ``vector'' and $M\in G^{2\npq+1}$ a ``matrix'' with rows $\bm{m_{i}}\in G$. The operation $\bm\varrho\circledast M$ is defined as the induced operation of $\circledast$, i.e.,
    \[
        \bm\varrho\circledast M:= \varrho_{0}\cdot \bm{m_{0}}\circledast\varrho_{1}\cdot \bm{m_{1}}\circledast\dotsc \circledast\varrho_{2\npq}\cdot \bm{m_{2\npq}},
    \]
    with $\cdot$ defined in Eq.~\eqref{eq:scalarmultiplication}.
\end{definition}
\begin{remark}\label{rmk:vector_matrix_non_scalarproduct}
    Of course, since $\circledast$ is not commutative, the ordering of the $\circledast$-summands in the definition is crucial. For example permuting the entries of $\bm\varrho$ and the rows of $M$ in the same manner will in general give a different product. Additionally, recall the problem of the ``nearly'' scalar product of Rmk.~\ref{rmk:no-vectorspace}, which implies that this definition has to be taken with care: 
    
    Since the elements $\varrho_{\ipq}$ for $\ipq=1,\dotsc,2\npq$ are in $\mathbb{F}_2$, $\varrho_{\ipq}\cdot \bm{m_{\ipq}}$ in these cases just corresponds to including those rows $\bm{m_{\ipq}}$ into the $\circledast$-sum that have $\varrho_{\ipq}=1$. As explained in the idea 2 of the remark on the ``nearly'' scalar product, this might lead to problems if we had two identical rows $\bm{m_{\ipq}}=\bm{m_{\ipq'}}$ with $\varrho_{\ipq}=\varrho_{\ipq'}=1$ since then the intuitive calculation $\varrho_{\ipq}\cdot \bm{m_{\ipq}}\circledast\varrho_{\ipq'}\cdot \bm{m_{\ipq'}}=(1\oplus_2 1)\bm{m_{\ipq}}=\paulivec{0}{\bm0}{\bm0}$ would be false if $\bm{m_{\ipq}}$ had order $4$. 
    
    The element $\varrho_0$ is in $\mathbb{Z}_4$, which sets this summand apart and makes it less prone to intuitively use scalar product rules for $\varrho_{0}\cdot \bm{m_{0}}\circledast\varrho_{\ipq}\cdot \bm{m_{\ipq}}$, which are false in general. Especially, the case of idea 1 in the remark corresponding to $(PQ)^a\not=P^aQ^a$ for non-commutative $P$ and $Q$ cannot arise, since we only have one summand with a $\mathbb{Z}_4$-factor. Bear in mind that $\varrho_{0}\cdot\bm{m_{0}}$ corresponds to the exponentiation $P_{0}^{\varrho_{0}}$ of the Pauli operator $P_{0}=\varrho^{-1}(\bm{m_{0}})$, it is \emph{not} the $\circledast$-product of $\varrho_{0}$ and $\bm{m_{0}}$ corresponding to the product $i^{\varrho_{0}}\idmatrix\cdot P_{0}$ (see idea~4 of Rmk.~\ref{rmk:no-vectorspace}), especially $\varrho_{0}\cdot\bm{e_0}=(\varrho_{0}\cdot 1\vert\bm{0}\vert\bm{0})$ which is not equal to $\varrho_{0}\circledast\bm{e_0}=(\varrho_{0}\oplus_41\vert\bm{0}\vert\bm{0})$.

    Hence, as the above definition is stated, it is well-defined, but simplifying the $\circledast$-sum with scalar product rules and the known rules from vector-matrix-multiplication is in general false. From Prop.~\ref{prop:proper_tableau} we know that a proper ``matrix'' $M=F(C)$ is symplectic, which implies that the rows of $M$ correspond to a set of independent generators of the Pauli group. This excludes the problem of two identical ``vectors'' in the computation of the $\circledast$-sum. In addition, the $0^\text{th}$ row is just $\bm{e_0}$ and therefore in the centralizer of $G$, while all other rows $\bm{m_{\ipq}}$ are proper, i.e., correspond to Hermitian Pauli operators. Therefore, the calculation of the product $\bm\varrho\circledast F(C)$ is not prone to errors out of false scalar multiplication rules.
\end{remark}

\begin{corollary}\label{cor:vector_matrix_multiplication}
    Let $\bm\varrho=\paulivec{\varrho_{0}}{\varrho_{1}\dotsc\varrho_{\npq}}{\varrho_{\npq+1}\dotsc\varrho_{2\npq}}\in G$ be a ``vector'' and $M\in Z(G)\times G^{2\npq}$ a ``matrix'' with rows $\bm{m_{0}}=\logicalpaulivec{\ep_{0}}{\bm{0}}{\bm{0}}\in Z(G)$ and $\bm{m_{i}}=\logicalflowp[i]\in G$. Then
    \begin{equation}
        \bm\varrho\circledast M=2s\circledast \bm\varrho\cdot M,
    \end{equation}
    where $\cdot$ is the standard vector-matrix-multiplication and $s=(\bmlogicalzeta[1]\dotsc         \bmlogicalzeta[2\npq])Q(\bm\varrho)\left(\begin{array}{c}
            \bmlogicalxi[1]^T\\
            \vdots\\
            \bmlogicalxi[2\npq]^T
    \end{array}\right)$.
\end{corollary}
\begin{proof}
    The standard vector-matrix-multiplication is well-defined in this situation, since in
    \begin{equation}\label{eq:standard_matrix_vector_product}
        \bm\varrho\cdot M=\left(\sum_{i}\varrho_{i}\ep_{i}\left\vert \left(\sum_{i}\varrho_{i}\logicalxi[i,j]\right)_j\right.\left\vert \left(\sum_{i}\varrho_{i}\logicalzeta[i,j]\right)_j\right.\right)
    \end{equation}
    the first component is well-defined in $\mathbb{Z}_4$ and the sums in the other two components are well-defined in $\mathbb{F}_2$ since $\bmlogicalxi[0]$ and $\bmlogicalzeta[0]$ vanish.
    
    From Cor.~\ref{cor:multi_pauli_product} we know the result of a $\circledast$-sum, hence if we set $i_0=0$ and denote by $i_t$ those $i>0$ with $\varrho_{i}=1$ we get
    \begin{equation}\begin{split}
        \bm\varrho\circledast M&=\varrho_{0}\bm{m_{0}}\circledast \varrho_{1}\bm{m_{1}}\circledast\dots \circledast  \varrho_{2\npq}\bm{m_{2\npq}}\\
        &=\varrho_{i_0}\bm{m_{i_0}}\circledast \varrho_{i_1}\bm{m_{i_1}}\circledast\dots \circledast  \varrho_{i_k}\bm{m_{i_k}}\\
        &= \left(
            \sum_{t} \ep_{i_t} \oplus_4 2s 
            \left\vert
                \sum_{t}\bmlogicalxi[i_t]
            \right.\left\vert
                \sum_{t}\bmlogicalzeta[i_t]
            \right.
        \right).\\
        &= 2s\circledast \left(
            \sum_{t} \ep_{i_t} 
            \left\vert
                \sum_{t}\bmlogicalxi[i_t]
            \right.\left\vert
                \sum_{t}\bmlogicalzeta[i_t]
            \right.
        \right),
    \end{split}\end{equation}
    where we factored out the element $2s=\logicalpaulivec{2s}{\bm0}{\bm0}$, which is in the centralizer of $G$.
    This expression without the $2s$ is the same as $\bm\varrho\cdot M$ in Eq.~\eqref{eq:standard_matrix_vector_product}, i.e., we are left with showing that  
    \begin{equation}\label{eq:phasecorrectionterm}\begin{split}
        s&=(\bmlogicalzeta[i_1]\dotsc\bmlogicalzeta[i_k])\cdot Q(\bm{1})\cdot\left(\begin{array}{c}
           \bmlogicalxi[i_1]^T\\
         \vdots\\ 
         \bmlogicalxi[i_k]^T
    \end{array}\right)\\
    &=(\bmlogicalzeta[1]\dotsc\bmlogicalzeta[2\npq])\cdot Q(\bm\varrho)\cdot\left(\begin{array}{c}
           \bmlogicalxi[1]^T\\
         \vdots\\ 
         \bmlogicalxi[2\npq]^T
    \end{array}\right).
    \end{split}\end{equation}
    By definition
    \begin{equation}
        Q=Q(\bm\varrho)=\left(\begin{array}{cccccc}
             0&\varrho_{1}\varrho_{2}&\varrho_{1}\varrho_{3}&\dots&\dots&\varrho_{1}\varrho_{2\npq}  \\
             0&0&\varrho_{2}\varrho_{3}&\dots&\dots&\varrho_{2}\varrho_{2\npq} \\
             0&0&0&\ddots&&\varrho_{3}\varrho_{2\npq} \\
             \vdots& \vdots&\vdots&\ddots&\ddots&\vdots \\
             0& 0& 0&\dots&0&\varrho_{2\npq-1}\varrho_{2\npq} \\
             0&0&0&\dots&0&0 \\
        \end{array}\right),
    \end{equation}
    therefore, the formal quadratic form with $Q$ calculates the sum of all phase contributions $(\varrho_{i}\bmlogicalzeta[i])(\varrho_{j}\bmlogicalxi[j])^T$ with $0<i<j$, showing that Eq.~\eqref{eq:phasecorrectionterm} is true.
\end{proof}

\begin{remark}
    Alternatively the sign correction $s$ can be written as 
    \begin{equation}\label{eq:sign_correction_strup}
        s=\bm\varrho\cdot
                \mathrm{strup}\left(\begin{array}{ccc}
            \bmlogicalzeta[1]\bmlogicalxi[1]^T & \cdots & \bmlogicalzeta[1]\bmlogicalxi[2\npq]^T \\
            \vdots&&\vdots\\
            \bmlogicalzeta[2\npq]\bmlogicalxi[1]^T & \cdots & \bmlogicalzeta[2\npq]\bmlogicalxi[2\npq]^T,
        \end{array}\right)\cdot
        \bm\varrho^T,
    \end{equation}
    where the stated product in the strict sense is undefined, instead of $\bm\varrho\in G$ we have to take its phaseless part $\bm{\tilde\varrho}=(\varrho_{1}\dotsc\varrho_{\npq}\vert\varrho_{\npq+1}\dotsc\varrho_{2\npq})\in\mathbb{F}_{2}^\npq\times \mathbb{F}_{2}^\npq$.
\end{remark}

\begin{corollary}
    Let $P\in\Pauli$ be a proper Pauli operator, $\bm\varrho=\varrho(P)$ and $M\in G^{2\npq+1}$ a proper ``matrix''. Then $\bm\varrho\circledast M$ is again proper.
\end{corollary}

\begin{proof}
    Since $M$ is proper, by Prop.~\ref{prop:proper_tableau} the $0^\text{th}$ row $\bm{m_0}=\bm{e_0}$ and all other rows $\bm{m_i}$ are proper. W.l.o.g.\ we may reorder the $\circledast$-summands of $\bm\varrho\circledast M$ in pairs $\varrho_i\bm{m_i}\circledast\varrho_{\hat\imath}\bm{m_{\hat\imath}}$, since this might at most lead to a sign change. Since $M$ is symplectic these pairs of two rows anticommute iff $\varrho_i\varrho_{\hat\imath}=1$, i.e., iff $P$ contains a factor $Y_i$. By Cor.~\ref{cor:proper_product_of_proper} this leads to a non-proper $\circledast$-sum. Reducing the phase exponent of $\varrho_0\bm{m_0}$ by one and increasing the phase exponent of $\varrho_i\bm{m_i}\circledast\varrho_{\hat\imath}\bm{m_{\hat\imath}}$ by $1$ gives a proper result for this pair and does not change the end result. Since $P$ is proper, Prop.~\ref{prop:randrho} implies that the reduced phase exponent is even and the following $\circledast$-sums of pairs are all proper and commute with each other, therefore, $\bm\varrho\circledast M$ is Hermitian. 

    Since by the symplecticity of $M$ the rows are independent, the phaseless part of $\bm\varrho\circledast M$ may not be ther zero vector, therefore it is proper.
\end{proof}

\begin{proposition}\label{prop:representation_pullback_flow}
    Let $\varrho(P)$ be the ``vector'' of an arbitrary Pauli operator $P\in\Pauli$, $\bm{f_{i}}=(\ep_{i}\vert \bmlogicalxi[i]\vert\bmlogicalzeta[i])$ the $i^{\text{th}}$ row of the flow tableau $F(C)$ for a Clifford unitary $C\in\Cl$. 
    
    Then we have
    \begin{equation}
        \varrho(C^\ast P)=\varrho(P)\circledast F(C) = 2s \circledast \varrho(P)\cdot F(C),
    \end{equation}
    where $\cdot$ is the regular matrix multiplication and the sign correction $s$ can be formally written as 
    \begin{equation}
        s=(\bmlogicalzeta[1],\dotsc,\bmlogicalzeta[2\npq])\cdot Q(P) \cdot \left(\begin{array}{c}{\bmlogicalxi[1]}^T\\\vdots\\{\bmlogicalxi[2\npq]}^T\end{array}\right).
    \end{equation}
\end{proposition}

\begin{proof} 
    Let $\bm{\varrho}:=\varrho(P)=\physicalflowp$ be the ``vector'' of $P$, such that $P=i^\eta \bm{\physicalpauliop{}{ }{}}^{\bmphysicalxi} \bm{\physicalpauliop{}{}{ }}^{\bmphysicalzeta}$. Since $C^\ast$ is a homomorphism we get
    \begin{equation}
            C^\ast P=(C^\ast(i\idmatrix))^\eta \cdot \prod_{\ipq=1}^\npq (C^\ast\physicalpauliop{}{\ipq}{})^{\physicalxi[\ipq]}\cdot\prod_{\ipq=1}^\npq (C^\ast\physicalpauliop{}{}{\ipq})^{\physicalzeta[\ipq]}.
    \end{equation}
    Applying the group isomorphism $\varrho$ therefore yields
    \begin{multline}
            \varrho(C^\ast P)=\eta\varrho(C^\ast(i\idmatrix)) \circledast \physicalxi[1]\varrho(C^\ast\physicalpauliop{}{1}{})\circledast \dots\circledast \physicalxi[\npq]\varrho(C^\ast\physicalpauliop{}{\npq}{})\\ 
            \circledast
            \physicalzeta[1]\varrho(C^\ast\physicalpauliop{}{}{1})\circledast\cdots \circledast
            \physicalzeta[\npq]\varrho(C^\ast\physicalpauliop{}{}{\npq})
    \end{multline}
    Comparing this with the definition of $F(C)$, which gives $\bm{f_{0}}=\varrho(C^\ast(i\idmatrix))$, $\bm{f_{i}}=\varrho(C^\ast\physicalpauliop{}{i}{})$ and $\bm{f_{\hat\imath}}=\varrho(C^\ast\physicalpauliop{}{}{i})$, as well as with the definition of the ``vector-matrix''-multiplication gives the first claim $\varrho(C^\ast P)=\varrho(P)\circledast F(C)$ of the proposition. The second claim than follows from Cor.~\ref{cor:vector_matrix_multiplication} and the definition of $Q(P)$.
\end{proof}

\begin{remark}
    The notation of the sign correction $s$ as a formal quadratic form gives us a good grip on the run time of the calculation of $\varrho(P)\circledast F$: In the worst case all $\varrho_i=1$ giving $2n^2-n$ non-zero entries in $Q$, each leading to the calculation of a scalar product of two vectors of length $\npq$. Therefore, the calculation is $\mathcal{O}(\npq^3)$ in the worst case. On the other hand, if some $\varrho_i=0$ the structure of $Q$ shows us that the whole $i^{\text{th}}$ row as well as the whole $i^{\text{th}}$ column is zero. Therefore, the calculation is $\mathcal{O}(h^2\npq)$, where $h$ is the Hamming norm of $\bm{\varrho}$, i.e., the number of non-zero entries. For example if you want to use the flow tableau to deduce the logical action of a physical $R_X$- or $R_Z$-rotation, you just read off one row of the tableau, which corresponds to a Hamming norm of $1$ and no sign correction term at all. If you consider a physical $R_Y$-rotation you only need the $X$- and $Z$-row of the flow tableau, corresponding to a Hamming norm of $2$. Hence, for the physical single-qubit rotations $R_P$ the calculation of $\varrho(P)\circledast F$ is just $\mathcal{O}(\npq^2)$ for the vector-matrix-multiplication.

    For large Hamming weights $h$ the sign correction is more expensive than the vector-matrix-multiplication. Using fast matrix multiplication and the formula in Eq.~\eqref{eq:sign_correction_strup}
    gives a run time of $\mathcal{O}(\npq^\omega)$ (instead of the above naive $\mathcal{O}(\npq^3)$, see Rmk.~\ref{rmk:runtime_pauliproduct}).
\end{remark}

The analogous proposition holds for the Clifford tableau:

\begin{proposition}\label{prop:representation_pushforward_Clifford}[see Thm.~1 in \cite{dehaenedemoor2003}, Sec.~2.3.1 in \cite{Gidney2021stimfaststabilizer}, Lemma~37 in \cite{GossetGrierKerznerSchaeffer2024}]
    Let $\varrho(\logicalp)$ be the ``vector'' of an arbitrary Pauli operator $\logicalp\in\bar\Pauli$, $\bm{t_{i}}=\physicalpaulivec{\eta_{i}}{\bm{x_{i}}}{\bm{z_{i}}}$ the $i^{\text{th}}$ row of the Clifford tableau $T(C)$ for a Clifford unitary $C\in\Cl$. 
    
    Then we have
    \begin{equation}
        \varrho(C_\ast \logicalp)=\varrho(\logicalp)\circledast T(C) = 2s \circledast \varrho(\logicalp)\cdot T(C),
    \end{equation}
    where $\cdot$ is the regular matrix multiplication and the sign correction $s$ can be formally written as 
    \begin{equation}
        s=(\physicalzeta[1],\dotsc,\physicalzeta[2\npq])\cdot Q(\logicalp) \cdot \left(\begin{array}{c}\physicalxi[1]^T\\\vdots\\\physicalxi[2\npq]^T\end{array}\right).
    \end{equation}
\end{proposition}

These two representation propositions give the desired analogues to the commutative diagram~\eqref{eq:cd_linear_map} for linear maps:
\begin{equation}\label{eq:cd_pull_back_push_forward}
\begin{tikzcd}[column sep=1cm, every cell/.style={inner xsep=1ex, inner ysep=0.85ex}]
    \bar\Pauli \arrow[d, "\varrho"'] \arrow[rd, phantom, "\circlearrowleft", shift left=0ex]  & \Pauli \arrow[l, "C^\ast"'] \arrow[d, "\varrho"] \\
    G& G \arrow[l, "\circledast F(C)"]
\end{tikzcd}\qquad\begin{tikzcd}[column sep=1cm, every cell/.style={inner xsep=1ex, inner ysep=0.85ex}]
    \bar\Pauli \arrow[d, "\varrho"'] \arrow[rd, phantom, "\circlearrowleft", shift left=0ex] \arrow[r, "C_\ast"] & \Pauli \arrow[d, "\varrho"] \\
    G\arrow[r, "\circledast T(C)"']& G
\end{tikzcd}
\end{equation}
Note the different directions of the pullback and the pushforward, mapping physical to logical Paulis resp.\ logical to physical Paulis.

\subsection{Covariance and contravariance of tableaus}
Until this point we could have perfectly done without the flow tableau, just working with the Clifford tableau $T(C)$ and switching to the inverse Clifford tableau $T(C^\dagger)=F(C)$ if needed. But in this subsection we want to address the covariance and contravariance of different usages of the (inverse) Clifford tableaus, where the differentiation between the Clifford and the flow tableau helps to categorize the different tableau variants. 

We start with the composition of two tableaus for which we need the $\circledast$-product of two ``matrices'':
\begin{definition}
    For ``matrices'' $M, M'\in G^{2\npq+1}$ with rows $\bm{m_{i}}$ resp.\ $\bm{m'_{i}}$ we define
    \begin{equation}\begin{split}
        M'\circledast M&:=\left(\begin{array}{c}
        \bm{m'_{0}}\circledast M \\
        \bm{m'_{1}}\circledast M \\
             \vdots\\
        \bm{m'_{2\npq}}\circledast M.
        \end{array}\right)
    \end{split}\end{equation}
\end{definition}

In the cases of interest we get a simplified way of calculating the ``matrix'' product by using Cor.~\ref{cor:vector_matrix_multiplication}:

\begin{corollary}\label{cor:matrixproduct_nice_zeroth_row}
    Let $M,M'\in Z(G)\times G^{2\npq}$ be ``matrices'' with rows $\bm{m_{i}}$ resp.\ $\bm{m'_{i}}$ such that $\bm{m_{0}}=\ep_{0}\bm{e_0}\in Z(G)$ resp.\ $\bm{m'_{0}}=\ep'_{0}\bm{e_0}\in Z(G)$. Then
    \begin{equation}\begin{split}
        M'\circledast M&=\left(\begin{array}{c}
             \ep_{0}\ep'_{0}\bm{e_0} \\
             2s'_{1} \circledast \bm{m'_{1}}\cdot M \\
             \vdots\\
             2s'_{2\npq} \circledast \bm{m'_{2\npq}}\cdot M
             \end{array}\right)\\
        &=:\left(\begin{array}{c}
             0  \\
             2s'_{1}\\
             \vdots\\
             2s'_{2\npq}              
        \end{array}\right)\circledast M'\cdot M,
    \end{split}\end{equation}
    where $\cdot$ is standard matrix multiplication, $s'_{i}$ is the sign correction of the product $\bm{m'_{i}}\circledast M$ and the operation $\circledast$ of a column vector and a matrix is the row-wise $\circledast$-product.
\end{corollary}
\begin{proof}
    We only have to prove the claim on the $0^\text{th}$ row, the other rows follow directly from Cor.~\ref{cor:vector_matrix_multiplication} since $M\in Z(G)\times G^{2\npq}$. For the $0^\text{th}$ row we get $\bm{m'_{0}}\circledast M = \ep'_{0}\bm{m_{0}}\circledast0\bm{m_1}\circledast\dots\circledast0\bm{m_{2\npq}}=\ep'_{0}\ep_{0}\bm{e_{0}}=0\circledast \ep'_{0}\ep_{0}\bm{e_{0}}$, where $\ep'_{0}\ep_{0}\bm{e_{0}}$ is exactly the $0^{\text{th}}$ row of the standard matrix product $M'M$.
\end{proof}

\begin{proposition}\label{prop:concatenation_flow}
    Let $C, C'\in\Cl$ be two Clifford unitaries and $F(C)$ resp.\ $F(C')$ their flow tableaus with rows $\bm{f_{i}}$ and $\bm{f'_{i}}$ respectively. Then the flow tableau of the concatenation $C'C$ is $F(C'C)= F(C')\circledast F(C)$, i.e.,
        \begin{equation}\begin{split}
        F(C'C)&=\left(\begin{array}{c}
             \bm{e_0} \\
             2s'_{1} \circledast \bm{f'_{1}}\cdot F(C) \\
             \vdots\\
             2s'_{2\npq} \circledast \bm{f'_{2\npq}}\cdot F(C)
             \end{array}\right)\\
        &=\left(\begin{array}{c}
             0  \\
             2s'_{1}\\
             \vdots\\
             2s'_{2\npq}              
        \end{array}\right)\circledast F(C')\cdot F(C),
    \end{split}\end{equation}
    with sign corrections $s'_{i}$ coming from $\bm{f'_{i}}\circledast F(C)$.
\end{proposition}
\begin{remark}
    The sign corrections $s'_i$ can be calculated by the formal quadratic form 
    \begin{equation}
    s'_{i}=(\bmlogicalzeta[1]\dotsc\bmlogicalzeta[2\npq])\cdot Q(\bm{f'_{i}})\cdot\left(\begin{array}{c}
         \bmlogicalxi[1]^T  \\
         \vdots\\
         \bmlogicalxi[2\npq]^T
    \end{array}\right)
    \end{equation}
    with the $\bmlogicalxi$ and $\bmlogicalzeta$ of $F(C)$, i.e., $\bm{f_{i}}=\logicalflowp[i]$, or via 
        \begin{equation}\label{eq:signcorrection_concatenation}
        s'_{i}=\bm{f'_{i}}\cdot
                \mathrm{strup}\left(\begin{array}{ccc}
            \bmlogicalzeta[1]\bmlogicalxi[1]^T & \cdots & \bmlogicalzeta[1]\bmlogicalxi[2\npq]^T \\
            \vdots&&\vdots\\
            \bmlogicalzeta[2\npq]\bmlogicalxi[1]^T & \cdots & \bmlogicalzeta[2\npq]\bmlogicalxi[2\npq]^T,
        \end{array}\right)\cdot
        {\bm{f'_{i}}}^T.
    \end{equation}
    In the worst case we have some row $\bm{f'_i}=\logicalpaulivec{\ep_{i}}{1\dotsc 1}{1\dotsc 1}$ in $F(C')$, then we need all products $\bmlogicalzeta[i]\bmlogicalxi[j]^T$ with $0<i<j$ giving a runtime of $\mathcal{O}(\npq^\omega)$ for the sign corrections. In the best case all rows $\bm{f'_{i}}$ have Hamming weight 1 and no sign corrections are needed.
\end{remark}

\begin{proof}
    Using the definition of the flow tableau and the contravariance of $C^\ast$ yields
    \begin{equation}\begin{split}
        F(C'C)&=\left(\begin{array}{c}
             \varrho((C'C)^\ast g_{0})  \\
             \vdots\\
             \varrho((C'C)^\ast g_{2\npq})
        \end{array}\right)=\left(\begin{array}{c}
             \varrho(C^\ast({C'}^\ast g_{0}))  \\
             \vdots\\
             \varrho(C^\ast({C'}^\ast g_{2\npq}))
        \end{array}\right).
    \end{split}\end{equation}
    By Prop.~\ref{prop:representation_pullback_flow} with $P={C'}^\ast g_{i}$ for the flow tableau $F(C)$ and the definition of $F(C')$ we get
    \begin{equation}\begin{split}
        F(C'C)&=\left(\begin{array}{c}
             \varrho({C'}^\ast g_{0})\circledast F(C)  \\
             \vdots\\
             \varrho({C'}^\ast g_{2\npq})\circledast F(C)
        \end{array}\right)
        =\left(\begin{array}{c}
             \bm{f'_{0}}\circledast F(C)  \\
             \vdots\\
             \bm{f'_{2\npq}}\circledast F(C)
        \end{array}\right)
    \end{split}\end{equation}
    which by definition is $F(C')\circledast F(C)$. Since the flow tableaus are proper with $0^{\text{th}}$-row $\bm{e_0}$ the claim follows from Cor.~\ref{cor:matrixproduct_nice_zeroth_row}.
\end{proof}

Analogously we get

\begin{proposition}[see Thm.~2 in \cite{dehaenedemoor2003}, Sec.~2.3.2 in \cite{Gidney2021stimfaststabilizer}, Thm.~38 in \cite{GossetGrierKerznerSchaeffer2024}]\label{prop:concatenation_clifford}
        Let $C, C'\in\Cl$ be two Clifford unitaries and $T(C)$ and $T(C')$ their Clifford tableaus with rows $\bm{t_{i}}$ and $\bm{t'_{i}}$ respectively. Then the Clifford tableau of the concatenation $C'C$ is $T(C'C)=T(C)\circledast T(C')$, i.e.,
        \begin{equation}\begin{split}
        T(C'C)&= \left(\begin{array}{c}
             \bm{e_0}\\
             2s_{1} \circledast \bm{t_{1}}\cdot T(C') \\
             \vdots\\
             2s_{2\npq} \circledast \bm{t_{2\npq}}\cdot T(C').
        \end{array}\right)\\
        &=\left(\begin{array}{c}
             0  \\
             2s_{1}\\
             \vdots\\
             2s_{2\npq}              
        \end{array}\right)\circledast T(C)\cdot T(C')
    \end{split}\end{equation}
    where the sign corrections $s_{i}$ are those of the product $\bm{t_{i}}\circledast T(C')$.
\end{proposition}
\begin{proof}
    Analogous to the proof of Prop.~\ref{prop:concatenation_flow}, using the covariance of the pushforward, i.e., $(C'C)_\ast=C'_\ast C_\ast$.
\end{proof}

\begin{corollary}[see Lemma~5 in \cite{AaronsonGottesman2004}, Sec.~II in \cite{dehaenedemoor2003}]
    The group homomorphisms $T:\Cl\rightarrow G^{2n+1}$ and $F:\Cl\rightarrow G^{2n+1}$ are group isomorphisms onto their image, which is in both cases just the subgroup of proper tableaus.
\end{corollary}

The concatenation $C'C$ of two Clifford circuits can be seen in the following commutative diagrams for the flow tableau
\begin{equation}\label{eq:cd_composition_pull_back}
\begin{tikzcd}[column sep = 1cm, row sep = 1cm, every cell/.style={inner xsep=1ex, inner ysep=0.85ex}] 
\Pauli^{(0)} \arrow[d, "\varrho"] \arrow[dr, phantom, "\circlearrowleft", shift left=1ex] 
& |[alias=TM]| \Pauli^{(1)} \arrow[d, "\varrho"] \arrow[l, "C^\ast"] \arrow[dr, phantom, "\circlearrowleft", shift left=1ex] 
& \Pauli^{(2)} \arrow[d, "\varrho"] \arrow[l, "{C'}^\ast"] \arrow[ll, "(C'C)^\ast"', bend right, shift right,""{name=TT}]
\\
G 
& |[alias=BM]|G \arrow[l, "\circledast F(C)"'] 
& G \arrow[l, "\circledast F(C')"'] \arrow[ll, "\circledast F(C'C)=\circledast (F(C')\circledast F(C))", bend left, shift left,""{name=BB}]
\arrow[phantom, from=TM, to=TT, "\circlearrowleft"]
\arrow[phantom,from=BB,to=BM, "\circlearrowleft"']
\end{tikzcd}
\end{equation}
and for the Clifford tableau
\begin{equation}\label{eq:cd_composition_push_forward}
\begin{tikzcd}[column sep = 1cm, row sep = 1cm, every cell/.style={inner xsep=1ex, inner ysep=0.85ex}] 
\Pauli^{(0)} \arrow[d, "\varrho"] \arrow[r, "C_\ast"'] \arrow[rr, "(C'C)_\ast", bend left, shift left,""{name=TT}]   \arrow[dr, phantom, "\circlearrowleft", shift left=1ex]
& |[alias=TM]| \Pauli^{(1)} \arrow[d, "\varrho"] \arrow[r, "{C'}_\ast"']\arrow[dr, phantom, "\circlearrowleft", shift left=1ex] 
& \Pauli^{(2)} \arrow[d, "\varrho"] 
\\
G \arrow[r, "\circledast T(C)"] \arrow[rr, "\circledast T(C'C)=\circledast (T(C)\circledast T(C'))"', bend right, shift right,""{name=BB}] 
& |[alias=BM]| G \arrow[r, "\circledast T(C')"]
& G      
\arrow[phantom, from=TM, to=TT, "\circlearrowleft"]
\arrow[phantom,from=BB,to=BM, "\circlearrowleft"']
\end{tikzcd}
\end{equation}

\begin{remark}
    Note again the two different types of concatenation. For the contravariant pullback $(C'C)^\ast=C^\ast {C'}^\ast$ a ``vector'' $\bm\varrho\in G$ corresponding to a physical Pauli operator $P$ in $\Pauli=\Pauli^{(2)}$ gets pulled back through $C'C$ either by $\circledast$-multiplying $F(C'C)$ from the right, i.e., $\bm\varrho\circledast F(C'C)=\bm\varrho\circledast (F(C')\circledast F(C))$ or, alternatively, by firstly $\circledast$-multiplying with $F(C')$ from the right and then by $F(C)$ from the right, i.e., $(\bm\varrho\circledast F(C'))\circledast F(C)$
    
    In contrast, for the covariant pushforward $(C'C)_\ast=C'_\ast C_\ast$ a ``vector'' $\bm\varrho\in G$ corresponding to a logical Pauli operator $\logicalp$ in $\bar\Pauli=\Pauli^{(0)}$ gets pushed forward through $C'C$ either by $\circledast$-multiplying $T(C'C)$ from the right, i.e., $\bm\varrho\circledast T(C'C)=\bm\varrho\circledast(T(C)\circledast T(C'))$ or, alternatively, by firstly $\circledast$-multiplying with $T(C)$ from the right and then by $T(C')$ from the right, i.e., $(\bm\varrho\circledast T(C))\circledast T(C')$. 

    Therefore, we see again that the pushforward is naturally iterative, firstly the Clifford tableau $T(C)$ of the first circuit $C$ gets applied to $\bm\varrho$ and secondly the Clifford tableau $T(C')$ of the second circuit $C'$ gets applied to the former result. On the other hand, the pullback is not naturally iterative, since we firstly have to apply the second flow tableau $F(C')$ and only afterwards may apply the first flow tableau $F(C)$.

    The above view also shows that the apparent change of ordering in $T(C'C)=T(C)\circledast T(C')$ is just an artificial one since we have to look from the point of view of $\bm\varrho$, where $\bm\varrho\circledast T(C)\circledast T(C')$ shows that the first Clifford tableau $T(C)$ is the first ``matrix'' to be applied. Vice versa, the apparent non-change of ordering in $F(C'C)=F(C')\circledast F(C)$ is also just artificial, since in $\bm\varrho\circledast F(C')\circledast F(C)$ the second flow tableau $F(C')$ is the first ``matrix'' to be applied. 
    
    We could have avoided this artificial ordering problem by choosing to present Pauli operators $P$ by the column ``vector'' $\varrho(P)^T$, working with the transposed tableaus $T(C)^T$ and $F(C)^T$, where we put the presentations of the conjugated generators as columns of the tableau and using multiplication of column ``vectors'' by ``matrices'' from the left. We nevertheless decided to write row ``vectors'' and multiplication by ``matrices'' from the right since this makes the translation between the flow labels written directly onto the wires of a circuit and the flow tableau more intuitive. After all, the flow labels on a horizontal wire are row-like. 
\end{remark}

\begin{corollary}
    In the situations of Prop.~\ref{prop:concatenation_flow} and~\ref{prop:concatenation_clifford} for every $\bm\varrho\in G$ the following holds:
    \begin{equation}\begin{split}
        \bm\varrho\circledast(F(C')\circledast F(C))&=(\bm\varrho\circledast F(C'))\circledast F(C)\\
        \bm\varrho\circledast(T(C)\circledast T(C'))&=(\bm\varrho\circledast T(C))\circledast T(C').
    \end{split}\end{equation}
\end{corollary}

\begin{remark}
    Note that in general $\circledast$ is not associative in this way, i.e., if $\bm{\varrho}\in G$ and $M, M'\in G^{2\npq +1}$, in general
    \begin{equation}
        \bm{\varrho}\circledast (M'\circledast M)\not=(\bm{\varrho}\circledast M')\circledast M.
    \end{equation}
    This follows from the fact that the multiplication $a\bm{\varrho}$ from Eq.~\eqref{eq:scalarmultiplication} is not a scalar multiplication. For example:
    \begin{equation}\begin{split}
    &(\bm{\varrho}\circledast M') \circledast M \\
       =&\left((0\vert1\vert1)\circledast\left(\begin{array}{c|c|c}
             1&0&0\\\hline
             0&1&1\\\hline
             0&1&1 
        \end{array}\right)\right)\circledast\left(\begin{array}{c|c|c}
             1&0&0\\\hline
             1&1&0\\\hline
             0&0&1 
        \end{array}\right)\\
    =&(2\vert0\vert0)\circledast\left(\begin{array}{c|c|c}
             1&0&0\\\hline
             1&1&0\\\hline
             0&0&1 
        \end{array}\right)\\
    =&(2\vert0\vert0),
    \end{split}\end{equation}
    while
    \begin{equation}\begin{split}
    &\bm{\varrho}\circledast (M' \circledast M)\\
       =&(0\vert1\vert1)\circledast\left(\left(\begin{array}{c|c|c}
             1&0&0\\\hline
             0&1&1\\\hline
             0&1&1 
        \end{array}\right)\circledast\left(\begin{array}{c|c|c}
             1&0&0\\\hline
             1&1&0\\\hline
             0&0&1 
        \end{array}\right)\right)\\
    =&(0\vert1\vert1)\circledast\left(\begin{array}{c|c|c}
             1&0&0\\\hline
             1&1&1\\\hline
             1&1&1 
        \end{array}\right)\\
    =&(1\oplus_41\oplus_42\vert0\vert0)=(0\vert0\vert0).
    \end{split}\end{equation}
    Here, $M'$ is not proper because its second and third row coincide, while $M$ is not proper because its second row corresponds to the order $4$ Pauli operator $iX$.
\end{remark}

\begin{remark}
    The calculation of $F(C')\circledast F(C)$ has a run time of $\mathcal{O}(\npq^\omega)$, since this is the run time of the standard matrix product $F(C')F(C)$ as well as the run time of the sign correction in Eq.~\eqref{eq:signcorrection_concatenation}. See \cite{GossetGrierKerznerSchaeffer2024} for more details.
\end{remark}

We have already seen that the flow and the Clifford tableau are inverse to each other in the sense that $F(C)=T(C^\dagger)$ and $T(C)=F(C^\dagger)$. Using the symplectic property of the flow tableau $F(C)$, which gives an easy formula for the inverse of the phaseless block $\tilde F(C)$ and the formula for the phases of Prop.~\ref{prop:concatenation_flow} we get
\begin{proposition}\label{prop:inverse_flow_tableau}
    Given a Clifford unitary $C\in\Cl$ and its flow tableau 
    \[
    F=F(C)=\left(\begin{array}{c|c}
             1&\bm{0}\\\hline
             &\\[-11pt]
             \bm{\ep}&\tilde F
        \end{array}\right)=\left(\begin{array}{c|c|c}
             1&\multicolumn{2}{c}{\bm{0}}\\\hline
             \multirow{2}{7pt}{$\bm{\ep}$}&\textcolor{xlog}{A}&\textcolor{zlog}{B}  \\\cline{2-3}
             &\textcolor{xlog}{C}&\textcolor{zlog}{D} 
        \end{array}\right)
    \]
    in block form. The inverse flow tableau $F(C^\dagger)$, which is the Clifford tableau $T(C)$, is given by
    \begin{equation}
        F'=F(C^\dagger)= \left(\begin{array}{c|c|c}
             1&\multicolumn{2}{c}{\bm{0}}\\\hline
             \multirow{2}{7pt}{$\bm{\ep}'$}&&\\[-11pt]
             &\textcolor{zlog}{D}^T&\textcolor{zlog}{B}^T  \\\cline{2-3}
             &&\\[-11pt]
             &\textcolor{xlog}{C}^T&\textcolor{xlog}{A}^T 
        \end{array}\right)=\left(\begin{array}{c|c}
             1&\bm{0}\\\hline
             &\\[-11pt]
             \bm{\ep'}&\tilde F'
        \end{array}\right),
    \end{equation}
    where the phases $\bm{\ep}'$ of the inverse tableau are
    \begin{equation}
        \bm{\ep'}=-\tilde F'\bm{\ep}\oplus_42\bm{s}
    \end{equation}
    with sign correction
    \begin{equation}
        \bm{s}=\mathrm{diag^T}\left(\tilde F'\cdot\mathrm{strup}\left(\begin{array}{c|c}
             \textcolor{zlog}{B}\textcolor{xlog}{A}^T&\textcolor{zlog}{B}\textcolor{xlog}{C}^T  \\\hline
             &\\[-11pt]
             \textcolor{zlog}{D}\textcolor{xlog}{A}^T&\textcolor{zlog}{D}\textcolor{xlog}{C}^T
        \end{array}\right)\cdot\tilde {F'}^T\right),
    \end{equation}
    where $\mathrm{diag^T}$ is the column vector of the diagonal entries of the matrix product.
\end{proposition}

\begin{proof}
    That the inverse flow tableau $F(C^\dagger)$ is the Clifford tableau $T(C)$ follows from Cor.~\ref{cor:tableau_inverse_clifford}.
    
    From Prop.~\ref{prop:concatenation_flow} we know that $F(C^\dagger)\circledast F(C)=F(C^\dagger C)=F(\idmatrix)=\idmatrix\in G^{2\npq+1}$. Cor.~\ref{cor:matrixproduct_nice_zeroth_row} then gives $\left(\begin{array}{c}
         0  \\\hline
          2\bm{s}
    \end{array}\right)\circledast F'\cdot F=\idmatrix$. For the phaseless parts it follows that $\tilde F'\cdot \tilde F=\idmatrix\in\mathbb{F}_2^{2\npq\times2\npq}$ and since $\tilde F$ is symplectic by Prop.~\ref{prop:proper_tableau} we get $\tilde F'=\tilde F^{-1}=\left(\begin{array}{c|c}
         \textcolor{zlog}{D}^T&\textcolor{zlog}{B}^T  \\\hline
         &\\[-11pt]
         \textcolor{xlog}{C}^T& \textcolor{xlog}{A}^T
    \end{array}\right)$ from Prop.~\ref{prop:symplectic_matrices}.

    The sign corrections $s_i$, $i=1,\dotsc,2\npq$, come from $\bm{f'_i}\circledast F$, where $\bm{f'_i}$ is the $i^\text{th}$ row of $F'$, i.e., by Rmk.~\ref{rmk:strup_block_form} we get
    \begin{equation}\begin{split}
        s_i&=\bm{f'_i}\cdot\mathrm{strup}\left(\begin{array}{c|c}
             \textcolor{zlog}{B}\textcolor{xlog}{A}^T&\textcolor{zlog}{B}\textcolor{xlog}{C}^T  \\\hline
             &\\[-11pt]
             \textcolor{zlog}{D}\textcolor{xlog}{A}^T&\textcolor{zlog}{D}\textcolor{xlog}{C}^T
        \end{array}\right){\bm{f'_i}}^T,
    \end{split}\end{equation}
    which are exactly the diagonal entries of 
    \begin{equation}
        \tilde F'\cdot\mathrm{strup}\left(\begin{array}{c|c}
             \textcolor{zlog}{B}\textcolor{xlog}{A}^T&\textcolor{zlog}{B}\textcolor{xlog}{C}^T  \\\hline
             &\\[-11pt]
             \textcolor{zlog}{D}\textcolor{xlog}{A}^T&\textcolor{zlog}{D}\textcolor{xlog}{C}^T
        \end{array}\right)\cdot\tilde {F'}^T.
    \end{equation}
    From the $0^\text{th}$ column of $\left(\begin{array}{c}
         0  \\\hline
          2\bm{s}
    \end{array}\right)\circledast F'\cdot F=\idmatrix$ we finally get the following linear equation system over $\mathbb{Z}_4$
    \begin{equation}\begin{split}
        \left(\begin{array}{c}
         0  \\\hline
          2\bm{s}
    \end{array}\right)\circledast F'\cdot \left(\begin{array}{c}
         1  \\\hline
          \bm{\ep}
    \end{array}\right)&=\left(\begin{array}{c}
             1  \\\hline
             \bm{0} 
        \end{array}\right)\\
    \left(\begin{array}{c}
         0 \oplus_4 1 \\\hline
         \\[-11pt]
          2\bm{s} \oplus_4 \bm{\ep'}\oplus_4 \tilde F'\cdot \bm{\ep}\\
    \end{array}\right)&=\left(\begin{array}{c}
             1  \\\hline
             \bm{0} 
        \end{array}\right)\\
     \bm{\ep'}&=-(2\bm{s} \oplus_4 \tilde F'\cdot \bm{\ep})\\
     \bm{\ep'}&=-\tilde F'\cdot \bm{\ep}\oplus_4 2\bm{s},
    \end{split}\end{equation}
    since $\mathbb{Z}_4$ is a commutative ring with $-2s_i=2s_i\in\mathbb{Z}_4$ for $s_i\in\mathbb{F}_2$.
\end{proof}

\begin{remark}
    This shows that the inverting of a tableau has run time $\mathcal{O}(\max(n^\omega, n^2))=\mathcal{O}(n^\omega)$, coming from the calculation of the phases. The inverse of the phaseless part is just transposing and swapping blocks, which is $\mathcal{O}(n^2)$, if there is even need to realize the transposed and swapped matrix at all.
\end{remark}

Analogously we get
\begin{corollary}[see Thm.~3 in \cite{dehaenedemoor2003}, Sec.~2.3.3 in \cite{Gidney2021stimfaststabilizer}]\label{cor:inverse_clifford_tableau}
    Given a Clifford unitary $C\in\Cl$ and its Clifford tableau 
    \[
    T=T(C)=\left(\begin{array}{c|c}
             1&\bm{0}\\\hline
             &\\[-11pt]
             \bm{\ep}&\tilde T
        \end{array}\right)=\left(\begin{array}{c|c|c}
             1&\multicolumn{2}{c}{\bm{0}}\\\hline
             \multirow{2}{7pt}{$\bm{\ep}$}&\textcolor{xphys}{A}&\textcolor{zphys}{B}  \\\cline{2-3}
             &\textcolor{xphys}{C}&\textcolor{zphys}{D}
        \end{array}\right)
    \]
    in block form. The inverse Clifford tableau $T(C^\dagger)$, which is the flow tableau $F(C)$, is given by
    \begin{equation}
        T'=T(C^\dagger)= \left(\begin{array}{c|c|c}
             1&\multicolumn{2}{c}{\bm{0}}\\\hline
             \multirow{2}{7pt}{$\bm{\ep}'$}&&\\[-11pt]
             &\textcolor{zphys}{D}^T&\textcolor{zphys}{B}^T  \\\cline{2-3}
             &&\\[-11pt]
             &\textcolor{xphys}{C}^T&\textcolor{xphys}{A}^T 
        \end{array}\right)=\left(\begin{array}{c|c}
             1&\bm{0}\\\hline
             &\\[-11pt]
             \bm{\ep'}&\tilde T'
        \end{array}\right),
    \end{equation}
    where the phases $\bm{\ep}'$ of the inverse tableau are
    \begin{equation}
        \bm{\ep'}=-\tilde T'\bm{\ep}\oplus_42\bm{s}
    \end{equation}
    with sign correction
    \begin{equation}
        \bm{s}=\mathrm{diag^T}\left(\tilde T'\cdot\mathrm{strup}\left(\begin{array}{c|c}
             \textcolor{zphys}{B}\textcolor{xphys}{A}^T&\textcolor{zphys}{B}\textcolor{xphys}{C}^T  \\\hline
             &\\[-11pt]
             \textcolor{zphys}{D}\textcolor{xphys}{A}^T&\textcolor{zphys}{D}\textcolor{xphys}{C}^T 
        \end{array}\right)\cdot\tilde{T'}^T\right),
    \end{equation}
    where $\mathrm{diag^T}$ is the column vector of the diagonal entries of the matrix product.
\end{corollary}

\begin{remark}
    For another approach to calculate the inverse tableau by directly using commutativity properties of the generators $g_i$ of the Pauli group see Sec.~2.3.3 of \cite{Gidney2021stimfaststabilizer}. In the above proof, these commutativity properties are covered by the symplectic property of the phaseless tableaus.
\end{remark}

The results on the inverse flow tableau and the inverse Clifford tableau can be summarized in the following commutative diagram. Each square of maps is commutative, we omit the commutativity symbol $\circlearrowleft$ for better readability. Each downwards arrow is the presentation $\varrho$, the letter $\varrho$ is also omitted.
\begin{equation}\label{eq:cd_inverse_tableaus}
\begin{tikzcd}[column sep=4ex, row sep=3ex, every cell/.style={inner xsep=1ex, inner ysep=0.85ex}, background color=white, crossing over clearance=1.5ex]
    &  & \bar\Pauli \arrow[rrrr, "C_\ast" {xshift=20pt}, "{C^\dagger}^\ast"' {xshift=20pt}] \arrow[dd] &  &  &  & \Pauli \arrow[dd] \\
    \bar\Pauli \arrow[rru, equal] \arrow[dd] &  &  &  & \Pauli \arrow[rru, equal] \arrow[llll, "C^\ast"' {xshift=-20pt}, "C^\dagger_\ast" {xshift=-20pt}, crossing over]  &  & \\
    &  & G \arrow[rrrr, "\circledast T(C)" {xshift=20pt}, "\circledast F(C^\dagger)"' {xshift=20pt}]  &  &  &  & G \\
    G \arrow[rru, equal] &  &  &  & G \arrow[rru, equal] \arrow[llll, "\circledast F(C)"' {xshift=-10pt}, "\circledast T(C^\dagger)" {xshift=-10pt}] \arrow[from=uu, crossing over]&  &                  
\end{tikzcd}
\end{equation}
Note the inverse directions of the arrows in the front and in the back of the diagram. 

Now, we have everything in our hands to categorize tableaus into two classes, the covariant and the contravariant class. Let $C_1, C_2\in\Cl$ be two Clifford unitaries. Remember from the subsection Clifford group and its action on the Pauli group(p.~\pageref{subsec:clifford_on_pauli}) that there are three relevant choices:
\begin{description}
    \item[Choice 1] By default, the concatenation of $C_1$ and $C_2$ is appending $C_2$ to $C_1$, i.e., $C=C_2C_1$. But we could also prepend $C_2$ to get $C=C_1C_2$.
    \item[Choice 2] By default, the circuit $C$ is considered. But we could also opt for the inverse circuit $C^\dagger$.
    \item[Choice 3] The original Clifford tableau used the pushforward of Pauli operators, i.e., the left conjugation $C_\ast\logicalp=C\logicalp C^\dagger$. But we could also use the pullback, i.e., the right conjugation $C^\ast P=C^\dagger PC$.
\end{description}

In all three choices, the first version (append, circuit, pushforward) is covariant, and the second version (prepend, inverse circuit, pullback) is contravariant. For choice~3 that is just Prop.~\ref{prop:conjugation-functors}. For choice~2 it is well known that $C^\dagger=(C_2C_1)^\dagger=C_1^\dagger C_2^\dagger$, i.e., taking the inverse is contravariant, while just taking $C=C_2C_1$ is trivially covariant, i.e., does not change the ordering of the subcircuits. For choice~1 there is an apparent change of ordering of the subcircuits, but that is just the artificial one coming from the convention of concatenation of maps. From the point of view of the state $\ket{\psi}$, we first apply $C_1$ to $\ket{\psi}$ and afterwards $C_2$ in the append case, which is covariant, while for the prepend case $C_2$ is applied first, making it contravariant.

The above considerations give the binary tree of combinations of choices in Fig.~\ref{fig:decision_tree} and the proof of the following proposition.

\begin{figure*}
    \centering
    \includegraphics[width=\textwidth]{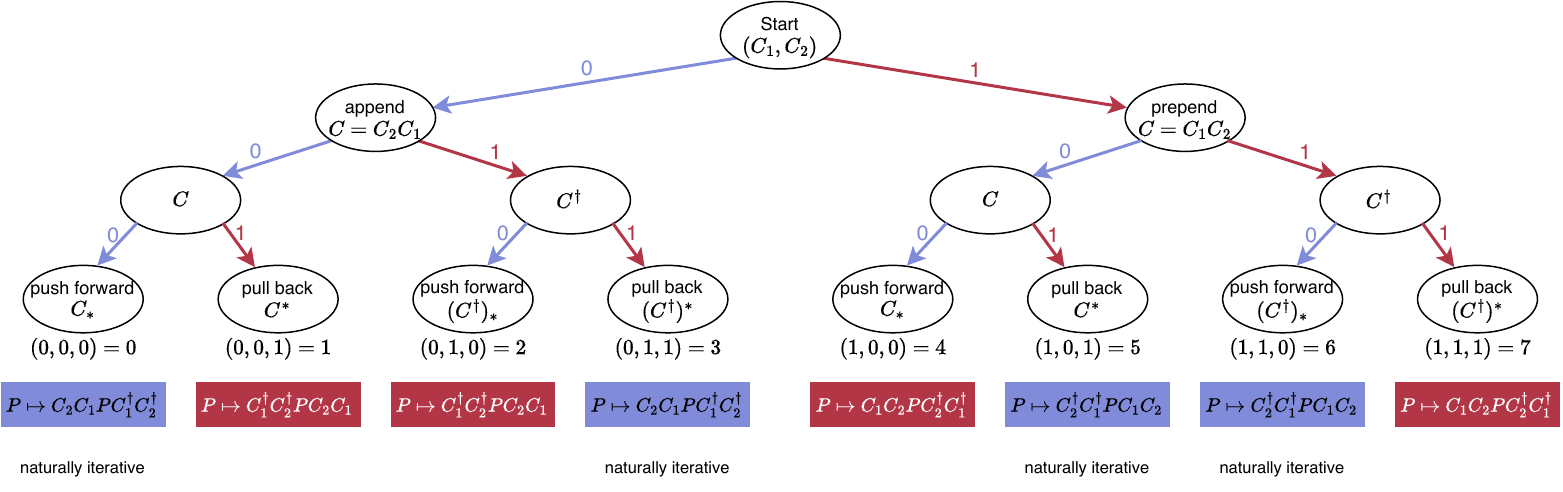}
    \caption{In the decision tree blue arrows (lighter color, left child node, weight $0$) indicate a covariant and red arrows (darker color, right child node, weight $1$) a contravariant choice. Each leaf is named as a binary triple $(i_1, i_2, i_3)$ encoding the path from the root to the leaf. An even number of contravariant steps is covariant, while an odd number of contravariant steps is contravariant. The rectangles at the leaves of the tree show the resulting mapping of Pauli operators and are marked blue resp.\ red if this mapping is co- resp.\ contravariant, i.e., if $C_1$ is the first conjugation on $P$ resp.\ $C_2$ is the first conjugation on $P$.}
    \label{fig:decision_tree}
\end{figure*}

\begin{proposition}\label{prop:functor}
    Let $\mathcal{F}_{1,i_1}:\Cl\times\Cl\rightarrow\Cl$, $i_1\in\mathbb{F}_2$, be the appending $\mathcal{F}_{1,0}(C_1, C_2):= C_2\cdot C_1$ resp.\ prepending $\mathcal{F}_{1,1}(C_1, C_2):= C_1\cdot C_2$ of the Clifford unitary $C_2$ to $C_1$, $\mathcal{F}_{2,i_2}:\Cl\rightarrow\Cl$, $i_2\in\mathbb{F}_2$, be the circuit itself $\mathcal{F}_{2,0}(C)=C$ resp.\ the inverse circuit $\mathcal{F}_{2,1}(C)=C^\dagger$ of a Clifford unitary and $\mathcal{F}_{3,i_3}:\Cl\rightarrow\Aut(\Pauli)$, $i_3\in\mathbb{F}_2$, the left resp.\ right conjugation, i.e., $\mathcal{F}_{3,0}(C):\Pauli\overset{\cong}{\rightarrow}\Pauli$, $P\mapsto C_\ast P=CPC^\dagger$ resp.\ $\mathcal{F}_{3,1}(C):\Pauli\overset{\cong}{\rightarrow}\Pauli$, $P\mapsto C^\ast P=C^\dagger PC$. 
    The concatenation 
    \begin{equation}
        \Cl\times\Cl\xrightarrow{\mathcal{F}_{1,i_1}}\Cl \xrightarrow{\mathcal{F}_{2,i_2}}\Cl\xrightarrow{\mathcal{F}_{3,i_3}}\Aut(\Pauli)
    \end{equation}
    given by a triple $(i_1, i_2, i_3)\in\mathbb{F}_2^3$ is covariant iff the sum $i_1\oplus_2i_2\oplus_2i_3$ is $0$ and contravariant iff the sum is $1$.
\end{proposition}

\begin{remark}\label{rmk:existing_tableaus}
    Most versions of Clifford tableaus in the literature can be categorized with the help of Fig.~\ref{fig:decision_tree} in combination with the additional choice of presentation, i.e., our presentation $\varrho$ based on the $XZ$ standard form of Eq.~\eqref{eq:standard_form_pauli} or a variant thereof or the mapping $r$ based on the $XYZ$ standard form of Eq.~\eqref{eqn:XYZ-standard-form} or a variant thereof. 
    
    For example, Dehaene and De Moor in~\cite{dehaenedemoor2003} use a variant of the $\varrho$-presentation (decomposing the phase $i^\ep=i^\delta(-1)^\epsilon$, i.e., using the reversed binary number of our phase exponent $\ep\in\mathbb{Z}_4$), the appending of Clifford circuits, the circuit itself and the left conjugation, corresponding to the decision-triple $(0,0,0)$ and, therefore, the covariant tracking of the Clifford tableau. 
    
    In the original Clifford tableau paper~\cite{AaronsonGottesman2004} by Aaronson and Gottesman, they use the same choice $(0,0,0)$ but the $XYZ$-mapping $r$ instead of $\varrho$, hence, also the covariant tracking by the Clifford tableau. 
    
    On the contrary, Gidney in his paper~\cite{Gidney2021stimfaststabilizer} on \texttt{stim} does not consider the left conjugation but the right conjugation, especially compare the Sec.~4.1 on The Asymptotic Benefits of Backwards Thinking. While he considers the different choices append/prepend and circuit/inverse circuit as well as mentions the possibility of combining the Clifford tableau and its inverse, his main choices in Sec.~4.2 on Tracking the Tableau seem to be prepending, inverse circuit and right conjugation, i.e., the triple $(1,1,1)$ giving a contravariant tableau, named inverse Clifford tableau by him (in the $XYZ$-mapping $r$). 
    
    Gosset et~al.\ in \cite{GossetGrierKerznerSchaeffer2024} mostly follow Dehaene and De Moor also giving a covariant tracking (in the reversed binary number $XZ$-presentation). 
    
    Windler et~al.\ in \cite{Winderl2023} mostly follow Aaronson and Gottesman, i.e., using the choices $(0,0,0)$ of the covariant tracking in the $XYZ$-mapping $r$ but applied to the Clifford tableau of the inverse circuit as initialization, which gets reduced to the identity tableau leading to a hardware-aware decomposition into elementary Clifford gates. They do point out that they could also use the prepending of the inverse elementary Clifford gates when initializing with the tableau of the circuit itself. In our decision tree in Fig.~\ref{fig:decision_tree} their two versions correspond to the choices $(0,0,0)$ and $(1,1,0)$, which are both covariant and, hence, give equivalent tableaus. 
    
    Schmitz et~al.\ in \cite{schmitz2023graphoptimizationperspectivelowdepth} choose a slightly different approach by working on the Pauli space $\bar{\Pauli}/\langle i\idmatrix\rangle$, hence, using the forgetful presentation $\tilde{\varrho}=\tilde{r}$ of Cor.~\ref{cor:forgetful_presentations}. They track only symplectic matrices which then correspond to the equivalence class of a Clifford operator up to multiplication with Pauli operators, recovering the missing phases in a later step.  Following a path in their Pauli Frame Graph corresponds to iteratively appending Clifford gates to a circuit $C$ while using the right conjugation, corresponding to the triple $(0,0,1)$ and therefore to the phaseless part of the contravariant flow tableau.

    In our opinion, the categorization into the covariant Clifford tableau and the contravariant flow tableau gives a better grip on the question which variants of the known tableaus are actually equivalent and which are not, also giving rise to a systematic way to vary the interpretation of such a tableau as needed. For instance, the classical Clifford tableau as left conjugation on the circuit itself with the append operation (left-most leaf $(0,0,0)$ of the decision tree) corresponds to the logical to physical question of the circuit $C$ and is equivalent to the physical to logical question for the inverse circuit $C^\dagger$, which corresponds to the tableau with left conjugation and append for the inverse circuit (leaf $(0,1,1)$). But it is not equivalent to the flow tableau corresponding to the choice $(0,0,1)$ with iteratively appending Clifford gates to the circuit itself and the right conjugation, which holds the information for the physical to logical question for the circuit $C$ itself and is equivalent to the inverse Clifford tableau of Gidney corresponding to the choice $(1,1,1)$ with the same conjugation type but prepending for the inverse circuit.
\end{remark}

\subsection{Iterative calculation of the tableaus}

After this categorization result we come back to the iterative calculation of the flow tableau. Remember from the section on the actions of the Clifford group on the Pauli group that the action of the left conjugation on the Pauli group is naturally iterative because of its covariant nature (see Fig.~\ref{fig:elementarypushforward}), which directly transfers to the Clifford tableau:

\begin{proposition}\label{prop:Clifford_iteratively}
Let $C'\in\{\CNOT[\icq,\itq], H_{\itq}, S_{\itq}\}$ be an elementary Clifford gate and $C\in\Cl$ be an arbitrary Clifford unitary with Clifford tableau 
\begin{equation}
    T(C)=\left(\begin{array}{c|c|c}
         1&\bm{0} &\bm{0} \\\hline
         \eta_{\ipq}& \bmphysicalxi[\ipq] & \bmphysicalzeta[\ipq]\\\hline
         \eta_{\ipqhat}& \bmphysicalxi[\ipqhat] & \bmphysicalzeta[\ipqhat]
    \end{array}\right).
\end{equation}
Then the Clifford tableau $T(C'C)$ after appending $C'$ to $C$ is
\begin{equation}
    T(C'C)=
        \left(\begin{array}{c}
             0\\\hline  
             \\[-11pt]
             \bm{0}^T\\\hline
            \\[-11pt]
             \bm{0}^T 
        \end{array}\right)\circledast T(C)\cdot\left(\begin{array}{c|c|c}
             1&\bm{0}& \bm{0} \\\hline
             \\[-11pt]
             \bm{0}^T& D_{\icq,\itq}& 0 \\\hline
             \\[-11pt]
             \bm{0}^T& 0 & D_{\itq,\icq}
        \end{array}\right)
\end{equation}
if $C'=\CNOT[\icq,\itq]$,
\begin{equation}
    T(C'C)=
        \left(\begin{array}{c}
             0\\\hline  
             (\physicalxi[\ipq,\itq]\physicalzeta[\ipq,\itq])_{\ipq}\\\hline
             (\physicalxi[\ipqhat,\itq]\physicalzeta[\ipqhat,\itq])_{\ipqhat}
        \end{array}\right)\circledast T(C)\cdot
        \left(\begin{array}{c|c}
             1&\bm{0}  \\\hline
             \\[-11pt]
             \bm{0}^T& E_{\itq,\itqhat}
        \end{array}\right)
\end{equation}
if $C'=H_{\itq}$, with column vector $(\physicalxi[\ipq,\itq]\physicalzeta[\ipq,\itq])_{\ipq}$ the componentwise product of the $\itq^{\text{th}}$ and $\itqhat^{\text{th}}$ columns of $T(C)$, and 
\begin{equation}
    T(C'C)=
        \left(\begin{array}{c}
             0\\\hline
             \\[-11pt]
             \bm{0}^T\\\hline
             \\[-11pt]
             \bm{0}^T 
        \end{array}\right)\circledast T(C)\cdot\left(\begin{array}{c|c}
             1&\bm{0}  \\\hline
             \\[-11pt]
             \bm{e_{\itq}}^T& D_{\itq,\itqhat}
        \end{array}\right)
\end{equation}
if $C'=S_{\itq}$. 
\end{proposition}

\begin{proof} 
    We already know from Prop.~\ref{prop:concatenation_clifford} and Cor.~\ref{cor:tableaus_elementary_cliffords} that 
    \begin{equation}
        T(C'C)=2\bm{s}\circledast T(C)\cdot T(C'),
    \end{equation}
    where $s_0=0$ and the entries $s_{\ipq}$ are the sign corrections coming from the product of the $\ipq^\text{th}$ row of $T(C)$ with $T(C')$ where 

\begin{equation}
    T(H_{\itq})= \left(\begin{array}{c|c|c}
        \scriptstyle 1& \scriptstyle\bm{0}&\scriptstyle \bm{0} \\\hline&\\[-13pt]
         \scriptstyle\bm{0}^T&\scriptstyle\idmatrix + \ket{\bm{e_t}}\bra{\bm{e_t}}& \scriptstyle\ket{\bm{e_t}}\bra{\bm{e_t}}\\\hline &\\[-13pt]
         \scriptstyle\bm{0}^T&\scriptstyle\ket{\bm{e_t}}\bra{\bm{e_t}}& \scriptstyle\idmatrix + \ket{\bm{e_t}}\bra{\bm{e_t}}
    \end{array}\right),
\end{equation}
\begin{equation}
    T(S_{\itq})=\left(\begin{array}{c|c|c}\scriptstyle
         1& \scriptstyle\bm{0} &\scriptstyle\bm{0} \\\hline&\\[-13pt]
         \scriptstyle\bm{e_\itq}^T&\scriptstyle\idmatrix & \scriptstyle\ket{\bm{e_t}}\bra{\bm{e_t}}\\\hline&\\[-13pt]
         \scriptstyle\bm{0}^T&\scriptstyle0&\scriptstyle\idmatrix
    \end{array}\right)
\end{equation}
and
\begin{equation}
    T(\CNOT[\icq,\itq])=
         \left(\begin{array}{c|c|c}\scriptstyle
             1&\scriptstyle\bm{0}&\scriptstyle \bm{0} \\\hline&\\[-13pt]
             \scriptstyle\bm{0}^T& \scriptstyle\idmatrix + \ket{\bm{e_{\icq}}}\bra{\bm{e_{\itq}}}& \scriptstyle0 \\\hline&\\[-13pt]
            \scriptstyle \bm{0}^T& \scriptstyle0 & \scriptstyle\idmatrix + \ket{\bm{e_{\itq}}}\bra{\bm{e_{\icq}}}
        \end{array}\right).
\end{equation}
A careful calculation of the matrices 
\begin{equation}
    \mathrm{strup}\left(\begin{array}{c|c}
             \textcolor{zphys}{B}\textcolor{xphys}{A}^T&\textcolor{zphys}{B}\textcolor{xphys}{C}^T  \\\hline
             &\\[-11pt]
             \textcolor{zphys}{D}\textcolor{xphys}{A}^T&\textcolor{zphys}{D}\textcolor{xphys}{C}^T 
        \end{array}\right)
\end{equation}
for the block matrix form $T(C')=\left(\begin{array}{c|c|c}
     1& \bm{0} &\bm{0}\\\hline
     \eta_{\ipq}& \textcolor{xphys}{A}&\textcolor{zphys}{B}\\\hline
     \eta_{\ipqhat}&\textcolor{xphys}{C}&\textcolor{zphys}{D}
\end{array}\right)$ shows that the sign corrections are zero for the cases $C'=S_\itq$ and $C'=\CNOT[\icq,\itq]$, while for $C'=H_{\itq}$ we get
\begin{equation}\begin{split}
    &\left(\begin{array}{c|c}
             \textcolor{zphys}{B}\textcolor{xphys}{A}^T&\textcolor{zphys}{B}\textcolor{xphys}{C}^T  \\\hline
             &\\[-11pt]
             \textcolor{zphys}{D}\textcolor{xphys}{A}^T&\textcolor{zphys}{D}\textcolor{xphys}{C}^T 
        \end{array}\right)\\
    =&\left(\begin{array}{c|c}
             (\ket{\bm{e_t}}\bra{\bm{e_t}})(\idmatrix+\ket{\bm{e_t}}\bra{\bm{e_t}})^T& (\ket{\bm{e_t}}\bra{\bm{e_t}})(\ket{\bm{e_t}}\bra{\bm{e_t}})^T  \\\hline
             &\\[-11pt]
             (\idmatrix+\ket{\bm{e_t}}\bra{\bm{e_t}})(\idmatrix+\ket{\bm{e_t}}\bra{\bm{e_t}})^T&(\idmatrix+\ket{\bm{e_t}}\bra{\bm{e_t}})(\ket{\bm{e_t}}\bra{\bm{e_t}})^T 
        \end{array}\right)\\
        =&\left(\begin{array}{c|c}
             0& \ket{\bm{e_t}}\bra{\bm{e_t}}  \\\hline
             &\\[-11pt]
             \idmatrix+\ket{\bm{e_t}}\bra{\bm{e_t}}&0
        \end{array}\right)
\end{split}\end{equation}
and, therefore, for the $\ipq^{\text{th}}$ row $\physicalflowp[\ipq]$ of $T(C)$ the sign correction is
\begin{equation}\begin{split}
    &(\bmphysicalxi[\ipq]\vert\bmphysicalzeta[\ipq])\cdot\mathrm{strup}\left(\begin{array}{c|c}
             0& \ket{\bm{e_t}}\bra{\bm{e_t}}  \\\hline
             &\\[-11pt]
             \idmatrix+\ket{\bm{e_t}}\bra{\bm{e_t}}&0
        \end{array}\right)\cdot \left(\begin{array}{c}
             \bmphysicalxi[\ipq]^T\\\hline\\[-11pt] \bmphysicalzeta[\ipq]^T
        \end{array}\right)\\
        &(\bmphysicalxi[\ipq]\vert\bmphysicalzeta[\ipq])\cdot\left(\begin{array}{c|c}
             0& \ket{\bm{e_t}}\bra{\bm{e_t}}  \\\hline
             &\\[-11pt]
             0&0
        \end{array}\right)\cdot \left(\begin{array}{c}
             \bmphysicalxi[\ipq]^T\\\hline \\[-11pt]\bmphysicalzeta[\ipq]^T
        \end{array}\right)\\
        =&\physicalxi[\ipq,\itq]\physicalzeta[\ipq,\itq]
\end{split}\end{equation}
as claimed, i.e., we get a sign change in a row iff the multiplication $T(C)\cdot 
        \left(\begin{array}{c|c}
             1&\bm{0}  \\\hline
             &\\[-11pt]
             \bm{0}^T& E_{\itq,\itqhat}
        \end{array}\right)$ swaps two $1$'s in that row.
\end{proof}

\begin{remark}
    As the matrices $D_{i,j}$ and $E_{i,j}$ are the well known matrices used in Gauß elimination for elementary column resp.\ row operations these actions can be nicely written in the usual form as column operations (compare Fig.~\ref{fig:elementarypushforward}), since the Clifford tableaus of the elementary Clifford gates multiplied from the right correspond to column operations:
    \begin{description}
        \item[$H_{\itq}$] Swap the $\itq^\text{th}$ and $\itqhat^\text{th}$ column of $T(C)$, in case you swap two $1$'s in some row, $\oplus_4$-add a $2$ to the phase.
    \begin{equation}\label{eq:col_ops_H}
        \begin{pNiceArray}{c|cc|cc}[first-row,first-col, last-row]
        & \eta
        & \Block[color=xphys, draw=white, line-width=0pt, rounded-corners]{1-1}{X_\ipq} 
        & \Block[color=xphys, draw=white, line-width=0pt, rounded-corners]{1-1}{X_\itq} 
        & \Block[color=zphys, draw=white, line-width=0pt, rounded-corners]{1-1}{Z_\ipq}
        & \Block[color=zphys, draw=white, line-width=0pt, rounded-corners]{1-1}{Z_\itq}\\
        \Block{1-1}{i\idmatrix}&1&0&0&0&0\\\hline
        \Block[color=xlog, draw=white, line-width=0pt, rounded-corners]{1-1}{\begin{array}{c}
              \vdots\\[0mm]\bar{X}_{i}\\[-2mm]\vdots
        \end{array}}&\eta_{i}&{\color{xphys}{{x}_{i,\ipq}}}&{\color{xphys}{{x}_{i,\itq}}}&{\color{zphys}{{z}_{i,\ipq}}}&{\color{zphys}{{z}_{i,\itq}}}\\\hline
        \Block[color=zlog, draw=white, line-width=0pt, rounded-corners]{1-1}{
        \begin{array}{c}
              \vdots\\[0mm]\bar{Z}_{i}\\[-2mm]\vdots
        \end{array}
        }& \eta_{\hat{\imath}}&{\color{xphys}{{x}_{\hat\imath,\ipq}}}&{\color{xphys}{{x}_{\hat\imath,\itq}}}&{\color{zphys}{{z}_{\hat\imath,\ipq}}}&{\color{zphys}{{z}_{\hat\imath,\itq}}}\\
        &\Block{1-5}{
            \begin{array}{c}
                \hspace*{13mm}\xleftrightarrow[\hspace*{12mm}]{\text{swap}}\\
                \xleftarrow{\hspace*{15mm}}{\scriptstyle 2\textcolor{xphys}{x_{\bullet,\itq}}\textcolor{zphys}{z_{\bullet,\itq}}}
         \end{array}}
    \end{pNiceArray}
    \end{equation}
            \item[$S_{\itq}$] $\oplus_2$-add the $\itq^\text{th}$ column of $T(C)$ onto the $\itqhat^\text{th}$ and $\oplus_4$-add the $\itq^\text{th}$ column of $T(C)$ onto the phase exponent column.
        \begin{equation}\label{eq:col_ops_S}
        \begin{pNiceArray}{c|cc|cc}[first-row,first-col, last-row]
        & \eta
        & \Block[color=xphys, draw=white, line-width=0pt, rounded-corners]{1-1}{X_\ipq} 
        & \Block[color=xphys, draw=white, line-width=0pt, rounded-corners]{1-1}{X_\itq} 
        & \Block[color=zphys, draw=white, line-width=0pt, rounded-corners]{1-1}{Z_\ipq}
        & \Block[color=zphys, draw=white, line-width=0pt, rounded-corners]{1-1}{Z_\itq}\\
        \Block{1-1}{i\idmatrix}&1&0&0&0&0\\\hline
        \Block[color=xlog, draw=white, line-width=0pt, rounded-corners]{1-1}{\begin{array}{c}
              \vdots\\[0mm]\bar{X}_{i}\\[-2mm]\vdots
        \end{array}}&\eta_{i}&{\color{xphys}{{x}_{i,\ipq}}}&{\color{xphys}{{x}_{i,\itq}}}&{\color{zphys}{{z}_{i,\ipq}}}&{\color{zphys}{{z}_{i,\itq}}}\\\hline
        \Block[color=zlog, draw=white, line-width=0pt, rounded-corners]{1-1}{
        \begin{array}{c}
              \vdots\\[0mm]\bar{Z}_{i}\\[-2mm]\vdots
        \end{array}
        }& \eta_{\hat{\imath}}&{\color{xphys}{{x}_{\hat\imath,\ipq}}}&{\color{xphys}{{x}_{\hat\imath,\itq}}}&{\color{zphys}{{z}_{\hat\imath,\ipq}}}&{\color{zphys}{{z}_{\hat\imath,\itq}}}\\
        &\Block{1-5}{
            \begin{array}{c}
                \hspace*{13mm}\xrightarrow[\hspace*{12mm}]{}\\
                \xleftarrow{\hspace*{12mm}}{}\hspace*{12mm}
         \end{array}}
    \end{pNiceArray}
    \end{equation}
    \item[${\mathrm{CNOT}}_{\icq\rightarrow\itq}$] $\oplus_2$-add the $\icq^\text{th}$ column of $T(C)$ onto the $\itq^\text{th}$ and $\oplus_2$-add the $\itqhat^\text{th}$ column of $T(C)$ onto the $\icqhat^\text{th}$.
        \begin{equation}\label{eq:col_ops_CNOT}
        \begin{pNiceArray}{c|ccc|ccc}[first-row,first-col, last-row]
        & \eta
        & \Block[color=xphys, draw=white, line-width=0pt, rounded-corners]{1-1}{X_\ipq} 
        & \Block[color=xphys, draw=white, line-width=0pt, rounded-corners]{1-1}{X_\icq} 
        & \Block[color=xphys, draw=white, line-width=0pt, rounded-corners]{1-1}{X_\itq} 
        & \Block[color=zphys, draw=white, line-width=0pt, rounded-corners]{1-1}{Z_\ipq}
        & \Block[color=zphys, draw=white, line-width=0pt, rounded-corners]{1-1}{Z_\icq}
        & \Block[color=zphys, draw=white, line-width=0pt, rounded-corners]{1-1}{Z_\itq}\\
        \Block{1-1}{i\idmatrix}&1&0&0&0&0&0&0\\\hline
        \Block[color=xlog, draw=white, line-width=0pt, rounded-corners]{1-1}{\begin{array}{c}
              \vdots\\[0mm]\bar{X}_{i}\\[-2mm]\vdots
        \end{array}}&\eta_{i}&{\color{xphys}{{x}_{i,\ipq}}}&{\color{xphys}{{x}_{i,\icq}}}&{\color{xphys}{{x}_{i,\itq}}}&{\color{zphys}{{z}_{i,\ipq}}}&{\color{zphys}{{z}_{i,\icq}}}&{\color{zphys}{{z}_{i,\itq}}}\\\hline
        \Block[color=zlog, draw=white, line-width=0pt, rounded-corners]{1-1}{
        \begin{array}{c}
              \vdots\\[0mm]\bar{Z}_{i}\\[-2mm]\vdots
        \end{array}
        }& \eta_{\hat{\imath}}&{\color{xphys}{{x}_{\hat\imath,\ipq}}}&{\color{xphys}{{x}_{\hat\imath,\icq}}}&{\color{xphys}{{x}_{\hat\imath,\itq}}}&{\color{zphys}{{z}_{\hat\imath,\ipq}}}&{\color{zphys}{{z}_{\hat\imath,\icq}}}&{\color{zphys}{{z}_{\hat\imath,\itq}}}\\
        &\Block{1-7}{
            \begin{array}{c}
                \hspace*{14mm}\xrightarrow[\hspace*{6mm}]{}\hspace*{14mm}\xleftarrow{\hspace*{6mm}}{}
         \end{array}}
    \end{pNiceArray}
    \end{equation}
    \end{description}
\end{remark}

This gives the well known iterative calculation of the Clifford tableau:

\begin{corollary}\label{cor:Clifford_append_iterative}
    Let $C=C_{\itimeend}\cdots C_1\in\Cl$ be a decomposition into elementary Clifford gates $C_{\itime}\in\{H_{\itq}, S_{\itq}, \CNOT[\icq,\itq]\}$. Initialize the Clifford tableau $T_0=T(\idmatrix)=\idmatrix\in G^{2\npq+1}$ and for each $\itime=1,\dotsc,\itimeend$ update the tableau by appending $C_\itime$ to $C_{\itime-1}\cdots C_1$ via the column operations in Eqs.~\eqref{eq:col_ops_H}--$\,$\eqref{eq:col_ops_CNOT} setting $T_{\itime}=T_{\itime-1}\circledast T(C_{\itime})$. Then $T_{\itime}$ is the Clifford tableau $T(C_{\itime}\cdots C_1)$ of the subcircuit $C_{\itime}\cdots C_1$ until time step $\itime$.
\end{corollary}

\begin{proposition}\label{prop:Flow_iteratively}
Let $C'\in\{\CNOT[\icq,\itq], H_{\itq}, S_{\itq}\}$ be an elementary Clifford gate and $C\in\Cl$ be an arbitrary Clifford unitary with flow tableau 
\begin{equation}
    F(C)=\left(\begin{array}{c|c|c}
         1&\bm{0} &\bm{0} \\\hline
         \ep_{\ipq}& \bmlogicalxi[\ipq] & \bmlogicalzeta[\ipq]\\\hline
         \ep_{\ipqhat}& \bmlogicalxi[\ipqhat] & \bmlogicalzeta[\ipqhat]
    \end{array}\right).
\end{equation}
Then the flow tableau $F(C'C)$ after appending $C'$ to $C$ is
\begin{equation}
    F(C'C)=
        \left(\begin{array}{c}
             0\\\hline  \\[-11pt]
             2\bmlogicalzeta[\icq]\bmlogicalxi[\itq]^T\bm{e_{\icq}}^T\\\hline
             \\[-11pt]
             2\bmlogicalzeta[\icqhat]\bmlogicalxi[\itqhat]^T\bm{e_{\itq}}^T 
        \end{array}\right)\circledast \left(\begin{array}{c|c|c}
             1&\bm{0}& \bm{0} \\\hline\\[-11pt]
             \bm{0}^T& D_{\icq,\itq}& 0 \\\hline\\[-11pt]
             \bm{0}^T& 0 & D_{\itq,\icq}
        \end{array}\right)\cdot F(C)
\end{equation}
if $C'=\CNOT[\icq,\itq]$,
\begin{equation}
    F(C'C)=
        \left(\begin{array}{c}
             0\\\hline  \\[-11pt]
             \bm{0}^T\\\hline\\[-11pt]
             \bm{0}^T
        \end{array}\right)\circledast 
        \left(\begin{array}{c|c}
             1&\bm{0}  \\\hline
             &\\[-11pt]
             \bm{0}^T& E_{\itq,\itqhat}
        \end{array}\right)\cdot F(C)
\end{equation}
if $C'=H_{\itq}$ and 
\begin{equation}
    F(C'C)=
        \left(\begin{array}{c}
             0\\\hline  \\[-11pt]
             2\bmlogicalzeta[\itq]\bmlogicalxi[\itqhat]^T\bm{e_{\itq}}^T\\\hline\\[-11pt]
             \bm{0}^T 
        \end{array}\right)\circledast\left(\begin{array}{c|c}
             1&\bm{0}  \\\hline&\\[-11pt]
             3\bm{e_{\itq}}^T& D_{\itq,\itqhat}
        \end{array}\right) \cdot F(C)
\end{equation}
if $C'=S_{\itq}$. 
\end{proposition}

\begin{proof}
    By Cor.~\ref{cor:tableaus_elementary_cliffords} we know the flow tableau $F(C')$ of the elementary Clifford gates and by Prop.~\ref{prop:concatenation_flow} we know that the flow tableau of the concatenation is $F(C'C)=2\bm{s}\circledast F(C')\cdot F(C)$, hence, we only have to prove the sign correction vectors whose entries come from the multiplications $\bm{f'_{i}}\circledast F(C)$, where $\bm{f'_{i}}$ is the $i^\text{th}$ row of $F(C')$. In almost all cases $\bm{f'_i}$ has weight $1$ and therefore the matrix $Q(\bm{f'_i})$ is zero, giving a sign correction of $0$. This shows the vanishing sign correction in the case $C'=H_{\itq}$. 

    In the case $C'=S_{\itq}$ we have $\bm{f'_{\itq}}=(3\vert\bm{e_{\itq}}\vert\bm{e_{\itq}})$ as the only row with Hamming weight greater than $1$ and $\bm{f'_{i}}\circledast F(C)=3\cdot\bm{e_0}\circledast 1\cdot\logicalflowp[\itq]\circledast 1\cdot\logicalflowp[\itqhat]$, which has the claimed sign correction term of $2\bmlogicalzeta[\itq]\cdot\bmlogicalxi[\itqhat]^T$.

    In the $\mathrm{CNOT}$ case we are left with the two rows $\bm{f'_{\icq}}=(0\vert\bm{e_{\icq}}\oplus \bm{e_{\itq}}\vert\bm{0})$ and $\bm{f'_{\itqhat}}=(0\vert\bm{0}\vert\bm{e_{\itq}}\oplus \bm{e_{\icq}})$. For $\bm{f'_{\icq}}$ we get a sign correction of $s_{\icq}=\bmlogicalzeta[\icq]\bmlogicalxi[\itq]^T$ if $\icq<\itq$ and $s_{\icq}=\bmlogicalzeta[\itq]\bmlogicalxi[\icq]^T$ if $\itq<\icq$. Since $F(C)$ is symplectic we have $0=\omega(\bm{f_{\icq}},\bm{f_{\itq}})=\omega((\bmlogicalxi[\icq]\vert\bmlogicalzeta[\icq]),(\bmlogicalxi[\itq]\vert\bmlogicalzeta[\itq]))=\bmlogicalzeta[\icq]\bmlogicalxi[\itq]^T\oplus_2\bmlogicalzeta[\itq]\bmlogicalxi[\icq]^T$ and, hence, in both cases $s_{\icq}=\bmlogicalzeta[\icq]\bmlogicalxi[\itq]^T$ giving the sign correction in the middle part, the sign correction in the lower part can be calculated analogously.
\end{proof}

\begin{remark}
    As in the case of the Clifford tableau, the actions of the elementary flow tableaus can be written as ``matrix'' operations, but since they act from the left they correspond to row operations:
    \begin{description}
        \item[$\CNOT$] In the $\color{xphys}{X}$ block $\circledast$-add the $\itq^{\text{th}}$ row to the $\icq^{\text{th}}$ row and in the $\color{zphys}{Z}$ block $\circledast$-add the $\itqhat^{\text{th}}$ row to the $\icqhat^{\text{th}}$ row. All other rows stay unchanged.
        \begin{equation}\label{eq:row_ops_CNOT}
        \begin{pNiceArray}{c|c|c}[first-row,first-col,last-col]
        & \ep
        & \Block[color=xlog, draw=white, line-width=0pt, rounded-corners]{1-1}{{\bar{X}_{\ipq}}} 
        & \Block[color=zlog, draw=white, line-width=0pt, rounded-corners]{1-1}{{\bar{Z}_{\ipq}}} & \\
        &1&0&0&\\\hline
        \Block[color=xphys, draw=white, line-width=0pt, rounded-corners]{1-1}{{X}_{\icq}}&\ep_{\icq}&\bmlogicalxi[\icq]&\bmlogicalzeta[\icq]&\Block{2-1}{\begin{array}{c}
             \multirow{2}{10pt}{$\Big\rceil \circledast$}\\
             \phantom{f}
        \end{array}}\\
        \Block[color=xphys, draw=white, line-width=0pt, rounded-corners]{1-1}{{X}_{\itq}}&\ep_{\itq}&\bmlogicalxi[\itq]&\bmlogicalzeta[\itq]&\\\hline
        \Block[color=zphys, draw=white, line-width=0pt, rounded-corners]{1-1}{{Z}_{\icq}}& \ep_{\icqhat}&\bmlogicalxi[\icqhat]&\bmlogicalzeta[\icqhat]&\Block{2-1}{\begin{array}{c}
             \multirow{2}{10pt}{$\Big\rfloor \circledast$}\\
             \phantom{f}
        \end{array}}\\
        \Block[color=zphys, draw=white, line-width=0pt, rounded-corners]{1-1}{{Z}_{\itq}}& \ep_{\itqhat}&\bmlogicalxi[\itqhat]&\bmlogicalzeta[\itqhat]&\\
    \end{pNiceArray}
    \end{equation}
    \item[$H_{\itq}$] Swap the $\itq^{\text{th}}$ row with the $\itqhat^{\text{th}}$ row. All other rows stay unchanged.
        \begin{equation}\label{eq:row_ops_H}
        \begin{pNiceArray}{c|c|c}[first-row,first-col,last-col]
        & \ep
        & \Block[color=xlog, draw=white, line-width=0pt, rounded-corners]{1-1}{{\bar{X}_{\ipq}}} 
        & \Block[color=zlog, draw=white, line-width=0pt, rounded-corners]{1-1}{{\bar{Z}_{\ipq}}} & \\
        &1&0&0&\\\hline
        \Block[color=xphys, draw=white, line-width=0pt, rounded-corners]{1-1}{{X}_{\itq}}&\ep_{\itq}&\bmlogicalxi[\itq]&\bmlogicalzeta[\itq]&\Block{2-1}{\begin{array}{c}
             \multirow{2}{5pt}{$\big]$ }\\
        \end{array}}\\\hline
        \Block[color=zphys, draw=white, line-width=0pt, rounded-corners]{1-1}{{Z}_{\itq}}& \ep_{\itqhat}&\bmlogicalxi[\itqhat]&\bmlogicalzeta[\itqhat]&\\
    \end{pNiceArray}
    \end{equation}
    \item[$S_{\itq}$] $\circledast$-add the $\itqhat^{\text{th}}$ row to the $\itq^{\text{th}}$ row. All other rows stay unchanged.
    \begin{equation}\label{eq:row_ops_S}
        \begin{pNiceArray}{c|c|c}[first-row,first-col,last-col]
        & \ep
        & \Block[color=xlog, draw=white, line-width=0pt, rounded-corners]{1-1}{{\bar{X}_{\ipq}}} 
        & \Block[color=zlog, draw=white, line-width=0pt, rounded-corners]{1-1}{{\bar{Z}_{\ipq}}} & \\
         &1&0&0&
        \Block{3-1}{
        \begin{array}{cc}
             &\\
             \multirow{2}{10pt}{$\big\rceil\circledast$}&\vert\circledast3\\
             \phantom{f} &
        \end{array}
        }\\\hline
        \Block[color=xphys, draw=white, line-width=0pt, rounded-corners]{1-1}{{X}_{\itq}}&\ep_{\itq}&\bmlogicalxi[\itq]&\bmlogicalzeta[\itq]&\\\hline
        \Block[color=zphys, draw=white, line-width=0pt, rounded-corners]{1-1}{{Z}_{\itq}}& \ep_{\itqhat}&\bmlogicalxi[\itqhat]&\bmlogicalzeta[\itqhat]&\\
    \end{pNiceArray}
    \end{equation}
    \end{description}
    The $\circledast$-adding of one row onto another is adding from the right as indicated by the notation of the operation on the right of the ``matrix''. Remember that the $\circledast$-operation is just the standard sum of the rows (i.e., $\oplus_4$-sum in the flow phase column and $\oplus_2$-sum else), where the flow phase gets added an additional $2$ iff the $\color{zlog}{\bar{Z}}$ entries of the first summand anticommute with the $\color{xlog}{\bar{X}}$ entries of the second summand, i.e., iff $\bmlogicalzeta[1]\bmlogicalxi[2]^T=1$.

    As stated in the section on the action of the elementary Clifford gates on the Pauli group this is not naturally an iterative action on the Pauli operators themselves but relies on the use of a tuple of symplectic generators, which result in the row operations that do not act on each Pauli operator in each row individually but act on the generating tuple combining resp.\ swapping elements in the tuple of symplectic generators.
\end{remark}

\begin{corollary}\label{cor:Flow_append_iterative}
    Let $C=C_{\itimeend}\cdots C_1\in\Cl$ be a decomposition into elementary Clifford gates $C_{\itime}\in\{H_{\itq}, S_{\itq}, \CNOT[\icq,\itq]\}$. Initialize the flow tableau $F_0=F(\idmatrix)=\idmatrix\in G^{2\npq+1}$ and for each $\itime=1,\dotsc,\itimeend$ update the tableau by appending $C_\itime$ to $C_{\itime-1}\cdots C_1$ via the row operations in Eqs.~\eqref{eq:row_ops_CNOT}--$\,$\eqref{eq:row_ops_S} setting $F_{\itime}= F(C_{\itime})\circledast F_{\itime-1}$. Then $F_{\itime}$ is the flow tableau $F(C_{\itime}\cdots C_1)$ of the subcircuit $C_{\itime}\cdots C_1$ until time step $\itime$.
\end{corollary}

\begin{corollary}\label{cor:cnot_only}
    Let $C=C_{\itimeend}\cdots C_1\in\Cl$ be a circuit of $\CNOT$-gates only, i.e., $C_{\itime}=\CNOT[\icq_{\itime},\itq_{\itime}]$. Then the flow tableau $F(C)$ has the following block form:
    \begin{equation}
        F(C)=\left(\begin{array}{c|c|c}
             1&\bm{0}& \bm{0} \\\hline &&\\[-11pt]
             \bm{0}^T& \textcolor{xlog}{A}&\textcolor{zlog}{0}\\\hline&&\\[-11pt]
             \bm{0}^T&\textcolor{xlog}{0}&\textcolor{zlog}{D}
        \end{array}\right),
    \end{equation}
    where $\textcolor{zlog}{D}$ is an invertible $\npq\times\npq$ matrix and $\textcolor{xlog}{A}=\textcolor{zlog}{D}^{-T}$.
\end{corollary}
\begin{proof}
    The initialization
    \begin{equation}
        F_0=\idmatrix=\left(\begin{array}{c|c|c}
             1 & \bm{0}&\bm{0}  \\\hline&&\\[-11pt]
             \bm{0}^T & \idmatrix & 0 \\\hline&&\\[-11pt]
             \bm{0}^T & 0 & \idmatrix
        \end{array}\right)\in G^{2\npq+1}
    \end{equation}
    has the claimed block structure and the action of 
    \begin{equation}
        F(\CNOT[\icq_{\itime},\itq_{\itime}])=\left(\begin{array}{c|c|c}
             1&\bm{0}& \bm{0} \\\hline&&\\[-11pt]
             \bm{0}^T& D_{\icq_{\itime},\itq_{\itime}}& 0 \\\hline&&\\[-11pt]
             \bm{0}^T& 0 & D_{\itq_{\itime},\icq_{\itime}}
        \end{array}\right)
    \end{equation}
    in each step does preserve it, since the standard matrix multiplication respects the block form and the sign corrections $2\bmlogicalzeta[\icq_{\itime}]\bmlogicalxi[\itq_{\itime}]^T$ and $2\bmlogicalzeta[\icqhat_{\itime}]\bmlogicalxi[\itqhat_{\itime}]^T$ vanish because of the blocks $\textcolor{xlog}{B}$ and $\textcolor{zlog}{C}$ being $0$. Therefore, each flow tableau $F_{\itime}$ has the claimed block structure. 

    By Prop.~\ref{prop:proper_tableau} we know that each flow tableau $F_{\itime}$ is symplectic, i.e., $F_{\itime}\Omega F_{\itime}^T=\Omega$ giving
    \begin{equation}\begin{split}
        \left(\begin{array}{c|c}
              \textcolor{xlog}{A}&0\\\hline
             0&\textcolor{zlog}{D}
        \end{array}\right)\left(\begin{array}{c|c}
              0&\idmatrix\\\hline
             \idmatrix&0
        \end{array}\right)\left(\begin{array}{c|c}
              \textcolor{xlog}{A}^T&0\\\hline&\\[-11pt]
             0&\textcolor{zlog}{D}^T
        \end{array}\right)&=\left(\begin{array}{c|c}
              0&\idmatrix\\\hline
             \idmatrix&0
        \end{array}\right)\\
        \textcolor{zlog}{D}\textcolor{xlog}{A}^T=\textcolor{xlog}{A}\textcolor{zlog}{D}^T&=\idmatrix,
    \end{split}\end{equation}
    i.e., $\textcolor{xlog}{A}=\textcolor{zlog}{D}^{-T}$.
    \end{proof}

\begin{remark}
    In the flow labels this means that in a $\CNOT$-circuit the labels of $\textcolor{xlog}{X}$-type $\flowlabel{}{\bm{\xi_{\ipq}}}{}$  corresponding to the block $\textcolor{xlog}{A}$ and the $\textcolor{zlog}{Z}$-type $\flowlabel{}{}{\bm{\zeta_{\ipq}}}$ corresponding to block $\textcolor{zlog}{D}$ never mix and the phase always stays trivial $+1$. In addition, the $\textcolor{xlog}{X}$-type labels can always be calculated from the $\textcolor{zlog}{Z}$-type labels as their transposed inverse in tableau form:
    \begin{equation}
        \left(\begin{array}{c}
             \bmlogicalxi[1]  \\
              \vdots\\
              \bmlogicalxi[\npq]
        \end{array}\right)=        \left(\begin{array}{c}
             \bmlogicalzeta[1]  \\
              \vdots\\
              \bmlogicalzeta[\npq]
        \end{array}\right)^{-T}.
    \end{equation}
\end{remark}

Let us now come back to the case of general Clifford circuits and look how the iterative calculation of the flow tableau works in an example:

\begin{example}\label{exp:heisenberg}
    Remember the circuit in Fig.~\ref{fig:full_circuit}a) of the main text which represents the dynamics of the 1D-Heisenberg model (see Peng et~al.~\cite{pengetal2022heisenberg}) shown in Fig.~\ref{fig:exp_heisenberg_combined_flow}. The \emph{Clifford phases} above the circuit will be introduced in the next section, so we will ignore them in this example. 

    \begin{figure*}
    \centering\includegraphics[width=\textwidth]{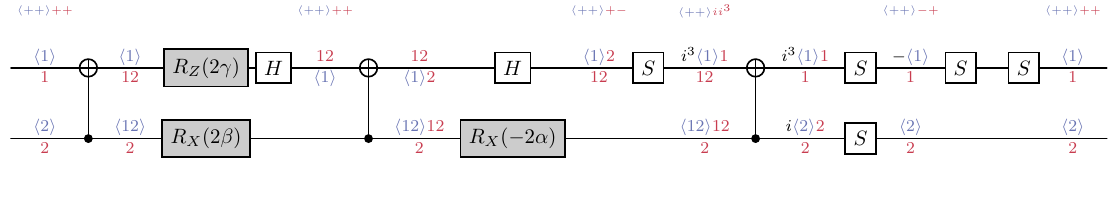}
   \caption{The 1D-Heisenberg example with Clifford phases above the circuit (which will be introduced in Def.~\ref{def:combined_tableau]}). The Clifford gates are tracked via the flow labels, while the gray colored physical rotations are not. Using the ``physical to logical'' direction, their respective logical meaning can be read off of the corresponding flow labels at that time step. $R_{\textcolor{zphys}{Z}}(2\gamma)$ corresponds to the $\textcolor{zphys}{Z_1}$-label $\textcolor{zlog}{12}$ and, therefore, to the logical rotation $\bar{R}_{\textcolor{zlog}{\bar{Z}_1\bar{Z}_2}}(2\gamma)$. $R_{\textcolor{xphys}{X}}(2\beta)$ corresponds to the $\textcolor{xphys}{X_2}$-label $\textcolor{xlog}{\yourfavouriteleftparanthesis12\yourfavouriterightparanthesis}$, therefore, to the logical rotation $\bar{R}_{\textcolor{xlog}{\bar{X}_1\bar{X}_2}}(2\beta)$. And $R_{\textcolor{xphys}{X}}(-2\alpha)$ corresponds to the $\textcolor{xphys}{X_2}$-label $\textcolor{xlog}{\yourfavouriteleftparanthesis12\yourfavouriterightparanthesis}\textcolor{zlog}{12}=-(i\textcolor{xlog}{\yourfavouriteleftparanthesis1\yourfavouriterightparanthesis}\textcolor{zlog}{1})(i\textcolor{xlog}{\yourfavouriteleftparanthesis2\yourfavouriterightparanthesis}\textcolor{zlog}{2})$, therefore, to the logical rotation $\bar{R}_{-\bar{Y}_1\bar{Y}_2}(-2\alpha)=\bar{R}_{\bar{Y}_1\bar{Y}_2}(2\alpha)$ as needed. In the reverse direction ``logical to physical'', the label corresponding to the pushforward of the logical $\textcolor{xlog}{\bar{X}_1}$ to the time step after the first $S$-gate can be read off by the $\textcolor{zlog}{i}$ at the third position of the Clifford labels and the qubits containing a $\textcolor{zlog}{1}$ in their labels (swapping $X$ and $Z$ for the Clifford meaning of the combined flow labels). The read off order is (Clifford phase, $\textcolor{xphys}{X}$ below wire, $\textcolor{zphys}{Z}$ above wire), i.e., $\textcolor{xlog}{\bar{X}_1}\mapsto i\textcolor{xphys}{X_1}\textcolor{zphys}{Z_1Z_2}$.} \label{fig:exp_heisenberg_combined_flow}
\end{figure*}
    
    In the beginning of the circuit we have the trivial flow tableau
    \begin{equation}
        F_0=F(\idmatrix)=\begin{pNiceArray}{c|cc|cc}[first-row,first-col,last-col]
        & \ep
        & \Block[color=xlog, draw=white, line-width=0pt, rounded-corners]{1-1}{{\bar{X}_{1}}} 
        & \Block[color=xlog, draw=white, line-width=0pt, rounded-corners]{1-1}{{\bar{X}_{2}}} 
        & \Block[color=zlog, draw=white, line-width=0pt, rounded-corners]{1-1}{{\bar{Z}_{1}}} 
        & \Block[color=zlog, draw=white, line-width=0pt, rounded-corners]{1-1}{{\bar{Z}_{2}}} & \\
        &1&0&0&0&0&\\\hline
        \Block[color=xphys, draw=white, line-width=0pt, rounded-corners]{1-1}{{X}_{1}}&0&\textcolor{xlog}{1}&\textcolor{xlog}{0}&\textcolor{zlog}{0}&\textcolor{zlog}{0}&\Block{2-1}{\begin{array}{c}
             \multirow{2}{10pt}{$\Big\rfloor \circledast$}\\
             \phantom{f}
        \end{array}}\\
        \Block[color=xphys, draw=white, line-width=0pt, rounded-corners]{1-1}{{X}_{2}}&0&\textcolor{xlog}{0}&\textcolor{xlog}{1}&\textcolor{zlog}{0}&\textcolor{zlog}{0}&\\\hline
        \Block[color=zphys, draw=white, line-width=0pt, rounded-corners]{1-1}{{Z}_{1}}&0&\textcolor{xlog}{0}&\textcolor{xlog}{0}&\textcolor{zlog}{1}&\textcolor{zlog}{0}&\Block{2-1}{\begin{array}{c}
             \multirow{2}{10pt}{$\Big\rceil \circledast$}\\
             \phantom{f}
        \end{array}}\\
        \Block[color=zphys, draw=white, line-width=0pt, rounded-corners]{1-1}{{Z}_{2}}&0&\textcolor{xlog}{0}&\textcolor{xlog}{0}&\textcolor{zlog}{0}&\textcolor{zlog}{1}&\\
    \end{pNiceArray}
    \end{equation}
    corresponding to the trivial labels $\flowlabel{}{1}{}$ and $\flowlabel{}{}{1}$ on the first qubit resp.\ $\flowlabel{}{2}{}$ and $\flowlabel{}{}{2}$ on the second qubit. The first elementary Clifford gate $\CNOT[2,1]$ acts on the lower half of the flow tableau by $\circledast$-adding the second row to the first row and on the middle of the flow tableau its action is reversed. As expected we do not get sign corrections from the $\circledast$-sums and the block structure is preserved:
    \begin{equation}
        F_1=F(\CNOT[2,1])=\begin{pNiceArray}{c|cc|cc}[first-row,first-col,last-col]
        & \ep
        & \Block[color=xlog, draw=white, line-width=0pt, rounded-corners]{1-1}{{\bar{X}_{1}}} 
        & \Block[color=xlog, draw=white, line-width=0pt, rounded-corners]{1-1}{{\bar{X}_{2}}} 
        & \Block[color=zlog, draw=white, line-width=0pt, rounded-corners]{1-1}{{\bar{Z}_{1}}} 
        & \Block[color=zlog, draw=white, line-width=0pt, rounded-corners]{1-1}{{\bar{Z}_{2}}} & \\
        &1&0&0&0&0&\\\hline
        \Block[color=xphys, draw=white, line-width=0pt, rounded-corners]{1-1}{{X}_{1}}&0&\textcolor{xlog}{1}&\textcolor{xlog}{0}&\textcolor{zlog}{0}&\textcolor{zlog}{0}&\Block{3-1}{\begin{array}{c}
             \multirow{3}{10pt}{$\Bigg]$}\\
             \phantom{f}\\
             \phantom{f}
        \end{array}}\\
        \Block[color=xphys, draw=white, line-width=0pt, rounded-corners]{1-1}{{X}_{2}}&0&\textcolor{xlog}{1}&\textcolor{xlog}{1}&\textcolor{zlog}{0}&\textcolor{zlog}{0}&\\\hline
        \Block[color=zphys, draw=white, line-width=0pt, rounded-corners]{1-1}{{Z}_{1}}&0&\textcolor{xlog}{0}&\textcolor{xlog}{0}&\textcolor{zlog}{1}&\textcolor{zlog}{1}&\\
        \Block[color=zphys, draw=white, line-width=0pt, rounded-corners]{1-1}{{Z}_{2}}&0&\textcolor{xlog}{0}&\textcolor{xlog}{0}&\textcolor{zlog}{0}&\textcolor{zlog}{1}&\\
    \end{pNiceArray}
    \end{equation}
    In the flow labels this corresponds to adding the $\textcolor{zphys}{Z}$-labels below the wires in the direction of the $\CNOT$, i.e., control onto target, $2\rightarrow 1$, and adding the $\textcolor{xphys}{X}$-labels above the wires in the reverse direction, i.e., target onto control. 
     \begin{equation}
     \includegraphics[height=15mm]{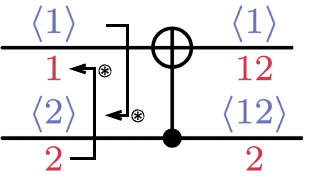}
     \end{equation}
     The following physical single-qubit rotations, $R_{\textcolor{zphys}{Z}}$ on the first qubit and $R_{\textcolor{xphys}{X}}$ on the second qubit, cannot and will not be tracked by the flow formalism as indicated by their gray coloring, but can now be easily translated into their logical meaning. Since the $\textcolor{zphys}{Z}$-label on the first qubit below the wire is $\flowlabel{}{}{12}$ the logical meaning $\CNOT[2,1]^\ast R_{\textcolor{zphys}{Z}}$ of the first rotation is $\bar{R}_{\textcolor{zlog}{\bar{Z}_1\bar{Z}_2}}(2\gamma)$. Analogously, the logical meaning $\CNOT[2,1]^\ast R_{\textcolor{xphys}{X}}$ of the second rotation is $\bar{R}_{\textcolor{xlog}{\bar{X}_1\bar{X}_2}}(2\beta)$.

     The next tracked gate $H_1$ corresponds to swapping the first rows of the middle and lower half of the flow tableau and swapping of the labels of the first qubit:
     \begin{equation}
        F_2=F(H_1\CNOT[2,1])=\begin{pNiceArray}{c|cc|cc}[first-row,first-col,last-col]
        & \ep
        & \Block[color=xlog, draw=white, line-width=0pt, rounded-corners]{1-1}{{\bar{X}_{1}}} 
        & \Block[color=xlog, draw=white, line-width=0pt, rounded-corners]{1-1}{{\bar{X}_{2}}} 
        & \Block[color=zlog, draw=white, line-width=0pt, rounded-corners]{1-1}{{\bar{Z}_{1}}} 
        & \Block[color=zlog, draw=white, line-width=0pt, rounded-corners]{1-1}{{\bar{Z}_{2}}} & \\
        &1&0&0&0&0&\\\hline
         \Block[color=xphys, draw=white, line-width=0pt, rounded-corners]{1-1}{{X}_{1}}&0&\textcolor{xlog}{0}&\textcolor{xlog}{0}&\textcolor{zlog}{1}&\textcolor{zlog}{1}&\Block{2-1}{\begin{array}{c}
              \multirow{2}{10pt}{$\Big\rfloor\circledast$}\\
              \phantom{f}
         \end{array}}\\
         \Block[color=xphys, draw=white, line-width=0pt, rounded-corners]{1-1}{{X}_{2}}&0&\textcolor{xlog}{1}&\textcolor{xlog}{1}&\textcolor{zlog}{0}&\textcolor{zlog}{0}&\\\hline
         \Block[color=zphys, draw=white, line-width=0pt, rounded-corners]{1-1}{{Z}_{1}}&0&\textcolor{xlog}{1}&\textcolor{xlog}{0}&\textcolor{zlog}{0}&\textcolor{zlog}{0}& \Block{2-1}{\begin{array}{c}
              \multirow{2}{10pt}{$\Big\rceil\circledast$}\\
              \phantom{f}
         \end{array}}\\
         \Block[color=zphys, draw=white, line-width=0pt, rounded-corners]{1-1}{{Z}_{2}}&0&\textcolor{xlog}{0}&\textcolor{xlog}{0}&\textcolor{zlog}{0}&\textcolor{zlog}{1}&\\
    \end{pNiceArray}
    \end{equation}
    resp.
     \begin{equation}
         \includegraphics[height=15mm]{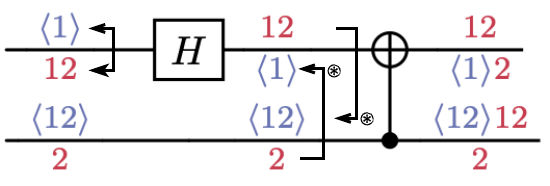}
     \end{equation}
    The action of the second $\CNOT$ is the same as that of the first one. Since we $\circledast$-add $\textcolor{zlog}{Z}$-type labels from the right onto $\textcolor{xlog}{X}$-type labels which is compatible with the standard form $\logicalflowp$, we do not get sign corrections. But the swapping of labels from the Hadamard gate gave us mixed labels, which we need in order to realize the logical $\bar{R}_{\bar{Y}}$-rotation as a physical $R_{\textcolor{xphys}{X}}$-rotation. The physical $R_{\textcolor{xphys}{X}}(-2\alpha)$ on the second qubit now corresponds to the $\textcolor{xphys}{X}$-label $\flowlabel{}{12}{12}=-(i\flowlabel{}{1}{1})(i\flowlabel{}{2}{2})$ and therefore, is a logical $\bar{R}_{-\bar{Y}_1\bar{Y}_2}(-2\alpha)=\bar{R}_{\bar{Y}_1\bar{Y}_2}(2\alpha)$ rotation, realizing the final needed rotation. The remaining Clifford gates are only needed in order to get back to the trivial encoding. We get $F_3=F(\CNOT[2,1]H_1\CNOT[2,1])$
    \begin{equation}
        F_3=\begin{pNiceArray}{c|cc|cc}[first-row,first-col,last-col]
        & \ep
        & \Block[color=xlog, draw=white, line-width=0pt, rounded-corners]{1-1}{{\bar{X}_{1}}} 
        & \Block[color=xlog, draw=white, line-width=0pt, rounded-corners]{1-1}{{\bar{X}_{2}}} 
        & \Block[color=zlog, draw=white, line-width=0pt, rounded-corners]{1-1}{{\bar{Z}_{1}}} 
        & \Block[color=zlog, draw=white, line-width=0pt, rounded-corners]{1-1}{{\bar{Z}_{2}}} & \\
        &1&0&0&0&0&\\\hline
         \Block[color=xphys, draw=white, line-width=0pt, rounded-corners]{1-1}{{X}_{1}}&0&\textcolor{xlog}{0}&\textcolor{xlog}{0}&\textcolor{zlog}{1}&\textcolor{zlog}{1}&\Block{3-1}{\begin{array}{c}
              \multirow{3}{10pt}{$\Bigg]$}\\
              \phantom{f}\\
              \phantom{f}
         \end{array}}\\
         \Block[color=xphys, draw=white, line-width=0pt, rounded-corners]{1-1}{{X}_{2}}&0&\textcolor{xlog}{1}&\textcolor{xlog}{1}&\textcolor{zlog}{1}&\textcolor{zlog}{1}&\\\hline
         \Block[color=zphys, draw=white, line-width=0pt, rounded-corners]{1-1}{{Z}_{1}}&0&\textcolor{xlog}{1}&\textcolor{xlog}{0}&\textcolor{zlog}{0}&\textcolor{zlog}{1}& \\
         \Block[color=zphys, draw=white, line-width=0pt, rounded-corners]{1-1}{{Z}_{2}}&0&\textcolor{xlog}{0}&\textcolor{xlog}{0}&\textcolor{zlog}{0}&\textcolor{zlog}{1}&\\
    \end{pNiceArray}
    \end{equation}
    $F_4=F(H_1\CNOT[2,1]H_1\CNOT[2,1])$
    \begin{equation}
        F_4=\begin{pNiceArray}{c|cc|cc}[first-row,first-col,last-col]
        & \ep
        & \Block[color=xlog, draw=white, line-width=0pt, rounded-corners]{1-1}{{\bar{X}_{1}}} 
        & \Block[color=xlog, draw=white, line-width=0pt, rounded-corners]{1-1}{{\bar{X}_{2}}} 
        & \Block[color=zlog, draw=white, line-width=0pt, rounded-corners]{1-1}{{\bar{Z}_{1}}} 
        & \Block[color=zlog, draw=white, line-width=0pt, rounded-corners]{1-1}{{\bar{Z}_{2}}} & \\
        &1&0&0&0&0&\\\hline
         \Block[color=xphys, draw=white, line-width=0pt, rounded-corners]{1-1}{{X}_{1}}&0&\textcolor{xlog}{1}&\textcolor{xlog}{0}&\textcolor{zlog}{0}&\textcolor{zlog}{1}&\Block{3-1}{\begin{array}{cc}
              \multirow{3}{10pt}{$\Bigg\rceil\circledast$}&\vert\circledast 3\\
              \phantom{f}&\\
              \phantom{f}&
         \end{array}}\\
         \Block[color=xphys, draw=white, line-width=0pt, rounded-corners]{1-1}{{X}_{2}}&0&\textcolor{xlog}{1}&\textcolor{xlog}{1}&\textcolor{zlog}{1}&\textcolor{zlog}{1}&\\\hline
         \Block[color=zphys, draw=white, line-width=0pt, rounded-corners]{1-1}{{Z}_{1}}&0&\textcolor{xlog}{0}&\textcolor{xlog}{0}&\textcolor{zlog}{1}&\textcolor{zlog}{1}& \\
         \Block[color=zphys, draw=white, line-width=0pt, rounded-corners]{1-1}{{Z}_{2}}&0&\textcolor{xlog}{0}&\textcolor{xlog}{0}&\textcolor{zlog}{0}&\textcolor{zlog}{1}&\\
    \end{pNiceArray}
    \end{equation}
$F_5=F(S_1H_1\CNOT[2,1]H_1\CNOT[2,1])$ gets a phase exponent of $3$ in the $\textcolor{xphys}{X_1}$-row coming from the action of the $S$-gate, but no sign correction from the $\circledast$-sum since the $\textcolor{zphys}{Z_1}$-row has $\bmlogicalxi[\hat{1}]=\textcolor{xlog}{(0,0)}$.
    \begin{equation}
        F_5=\begin{pNiceArray}{c|cc|cc}[first-row,first-col,last-col]
        & \ep
        & \Block[color=xlog, draw=white, line-width=0pt, rounded-corners]{1-1}{{\bar{X}_{1}}} 
        & \Block[color=xlog, draw=white, line-width=0pt, rounded-corners]{1-1}{{\bar{X}_{2}}} 
        & \Block[color=zlog, draw=white, line-width=0pt, rounded-corners]{1-1}{{\bar{Z}_{1}}} 
        & \Block[color=zlog, draw=white, line-width=0pt, rounded-corners]{1-1}{{\bar{Z}_{2}}} & \\
        &1&0&0&0&0&\\\hline
         \Block[color=xphys, draw=white, line-width=0pt, rounded-corners]{1-1}{{X}_{1}}&3&\textcolor{xlog}{1}&\textcolor{xlog}{0}&\textcolor{zlog}{1}&\textcolor{zlog}{0}&\Block{2-1}{\begin{array}{c}
              \multirow{2}{10pt}{$\Big\rfloor\circledast$}\\
              \phantom{f}
         \end{array}}\\
         \Block[color=xphys, draw=white, line-width=0pt, rounded-corners]{1-1}{{X}_{2}}&0&\textcolor{xlog}{1}&\textcolor{xlog}{1}&\textcolor{zlog}{1}&\textcolor{zlog}{1}&\\\hline
         \Block[color=zphys, draw=white, line-width=0pt, rounded-corners]{1-1}{{Z}_{1}}&0&\textcolor{xlog}{0}&\textcolor{xlog}{0}&\textcolor{zlog}{1}&\textcolor{zlog}{1}&\Block{2-1}{\begin{array}{c}
              \multirow{2}{10pt}{$\Big\rceil\circledast$}\\
              \phantom{f}
         \end{array}} \\
         \Block[color=zphys, draw=white, line-width=0pt, rounded-corners]{1-1}{{Z}_{2}}&0&\textcolor{xlog}{0}&\textcolor{xlog}{0}&\textcolor{zlog}{0}&\textcolor{zlog}{1}&\\
    \end{pNiceArray}
    \end{equation}
    In the flow labels the action of the $S$-gate looks like:
    \begin{equation}
        \includegraphics[height=8mm]{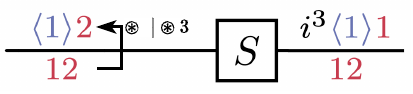}
    \end{equation}
    The third $\CNOT$ acts in the same way as the ones before, but now we do get a sign correction term, since we have to $\circledast$-add the flow labels $\flowlabel{}{12}{12}\circledast\flowlabel{i^3}{1}{1}$ where the $\flowlabel{}{}{1}$ from the first summand and the $\flowlabel{}{1}{}$ from the second summand do not commute, introducing an additional minus sign and therefore, resulting in $\flowlabel{1\cdot i^3\cdot(-1)}{12\triangle1}{12\triangle1}=\flowlabel{i}{2}{2}$ or equivalently in the flow tableau $F_6=F(\CNOT[2,1]S_1H_1\CNOT[2,1]H_1\CNOT[2,1])$:
    \begin{equation}
        F_6=\begin{pNiceArray}{c|cc|cc}[first-row,first-col,last-col]
        & \ep
        & \Block[color=xlog, draw=white, line-width=0pt, rounded-corners]{1-1}{{\bar{X}_{1}}} 
        & \Block[color=xlog, draw=white, line-width=0pt, rounded-corners]{1-1}{{\bar{X}_{2}}} 
        & \Block[color=zlog, draw=white, line-width=0pt, rounded-corners]{1-1}{{\bar{Z}_{1}}} 
        & \Block[color=zlog, draw=white, line-width=0pt, rounded-corners]{1-1}{{\bar{Z}_{2}}} & \\
        &1&0&0&0&0&\\\hline
         \Block[color=xphys, draw=white, line-width=0pt, rounded-corners]{1-1}{{X}_{1}}&3&\textcolor{xlog}{1}&\textcolor{xlog}{0}&\textcolor{zlog}{1}&\textcolor{zlog}{0}&\Block{4-1}{\begin{array}{cccc}
              \multirow{3}{10pt}{$\Bigg\rceil\circledast$}&\vert\circledast 3&&\\
              &\phantom{f}&\multirow{3}{10pt}{$\Bigg\rceil\circledast$}&\vert\circledast 3\\
              &\phantom{f}&&\\
              &\phantom{f}&&
         \end{array}}\\
         \Block[color=xphys, draw=white, line-width=0pt, rounded-corners]{1-1}{{X}_{2}}&3\oplus_42&\textcolor{xlog}{0}&\textcolor{xlog}{1}&\textcolor{zlog}{0}&\textcolor{zlog}{1}&\\\hline
         \Block[color=zphys, draw=white, line-width=0pt, rounded-corners]{1-1}{{Z}_{1}}&0&\textcolor{xlog}{0}&\textcolor{xlog}{0}&\textcolor{zlog}{1}&\textcolor{zlog}{0}& \\
         \Block[color=zphys, draw=white, line-width=0pt, rounded-corners]{1-1}{{Z}_{2}}&0&\textcolor{xlog}{0}&\textcolor{xlog}{0}&\textcolor{zlog}{0}&\textcolor{zlog}{1}&\\
    \end{pNiceArray}
    \end{equation}
    $F_8=F(S_2S_1\CNOT[2,1]S_1H_1\CNOT[2,1]H_1\CNOT[2,1])$:
    \begin{equation}
        F_8=\begin{pNiceArray}{c|cc|cc}[first-row,first-col,last-col]
        & \ep
        & \Block[color=xlog, draw=white, line-width=0pt, rounded-corners]{1-1}{{\bar{X}_{1}}} 
        & \Block[color=xlog, draw=white, line-width=0pt, rounded-corners]{1-1}{{\bar{X}_{2}}} 
        & \Block[color=zlog, draw=white, line-width=0pt, rounded-corners]{1-1}{{\bar{Z}_{1}}} 
        & \Block[color=zlog, draw=white, line-width=0pt, rounded-corners]{1-1}{{\bar{Z}_{2}}} & \\
        &1&0&0&0&0&\\\hline
         \Block[color=xphys, draw=white, line-width=0pt, rounded-corners]{1-1}{{X}_{1}}&2&\textcolor{xlog}{1}&\textcolor{xlog}{0}&\textcolor{zlog}{0}&\textcolor{zlog}{0}&\Block{4-1}{\begin{array}{cccc}
              \multirow{3}{10pt}{$\Bigg\rceil\circledast$}&\vert\circledast 3&\multirow{3}{10pt}{$\Bigg\rceil\circledast$}&\vert\circledast 3\\
              &\phantom{f}&&\\
              &\phantom{f}&&\\
              &\phantom{f}&&
         \end{array}}\\
         \Block[color=xphys, draw=white, line-width=0pt, rounded-corners]{1-1}{{X}_{2}}&0&\textcolor{xlog}{0}&\textcolor{xlog}{1}&\textcolor{zlog}{0}&\textcolor{zlog}{0}&\\\hline
         \Block[color=zphys, draw=white, line-width=0pt, rounded-corners]{1-1}{{Z}_{1}}&0&\textcolor{xlog}{0}&\textcolor{xlog}{0}&\textcolor{zlog}{1}&\textcolor{zlog}{0}& \\
         \Block[color=zphys, draw=white, line-width=0pt, rounded-corners]{1-1}{{Z}_{2}}&0&\textcolor{xlog}{0}&\textcolor{xlog}{0}&\textcolor{zlog}{0}&\textcolor{zlog}{1}&\\
    \end{pNiceArray}
    \end{equation}
    $F_{10}=F(S_1S_1S_2S_1\CNOT[2,1]S_1H_1\CNOT[2,1]H_1\CNOT[2,1])$
        \begin{equation}
        F_{10}=\begin{pNiceArray}{c|cc|cc}[first-row,first-col,last-col]
        & \ep
        & \Block[color=xlog, draw=white, line-width=0pt, rounded-corners]{1-1}{{\bar{X}_{1}}} 
        & \Block[color=xlog, draw=white, line-width=0pt, rounded-corners]{1-1}{{\bar{X}_{2}}} 
        & \Block[color=zlog, draw=white, line-width=0pt, rounded-corners]{1-1}{{\bar{Z}_{1}}} 
        & \Block[color=zlog, draw=white, line-width=0pt, rounded-corners]{1-1}{{\bar{Z}_{2}}} & \\
        &1&0&0&0&0&\\\hline
         \Block[color=xphys, draw=white, line-width=0pt, rounded-corners]{1-1}{{X}_{1}}&0&\textcolor{xlog}{1}&\textcolor{xlog}{0}&\textcolor{zlog}{0}&\textcolor{zlog}{0}&\\
         \Block[color=xphys, draw=white, line-width=0pt, rounded-corners]{1-1}{{X}_{2}}&0&\textcolor{xlog}{0}&\textcolor{xlog}{1}&\textcolor{zlog}{0}&\textcolor{zlog}{0}&\\\hline
         \Block[color=zphys, draw=white, line-width=0pt, rounded-corners]{1-1}{{Z}_{1}}&0&\textcolor{xlog}{0}&\textcolor{xlog}{0}&\textcolor{zlog}{1}&\textcolor{zlog}{0}& \\
         \Block[color=zphys, draw=white, line-width=0pt, rounded-corners]{1-1}{{Z}_{2}}&0&\textcolor{xlog}{0}&\textcolor{xlog}{0}&\textcolor{zlog}{0}&\textcolor{zlog}{1}&\\
    \end{pNiceArray}, 
    \end{equation}
    where the last two steps $S_1S_1=Z_1$ show that the action of the Pauli operator $\textcolor{zphys}{Z_1}$ is trivial on the phaseless part of the flow tableau, but introduces a minus sign to the phase of the $\textcolor{xphys}{X_1}$-row. Since we ended up in the trivial encoding, i.e., the identity matrix as flow tableau and the trivial labels as the flow labels, the whole circuit is equivalent to the logical meanings of its physical rotations, i.e., $R_{\bar{Y}_1\bar{Y_2}}(2\alpha)R_{\logicalpauliop{}{1,2}{}}(2\beta)R_{\logicalpauliop{}{}{1,2}}(2\gamma)$ as claimed.
\end{example}

\section{Flow tracking with stabilizers}\label{sec:stabilizers}
The flow labels developed so far are a good tool to determine the logical meaning of Pauli rotations in a circuit without auxiliary qubits, i.e., without stabilizers. They can for instance be used to analyze the QFT- and the QAOA-circuit on a linear nearest neighbour device (LNN) in the paper on the SWAPless implementation of quantum algorithms~\cite{klaver2024}. But if there are auxiliary qubits in the circuit, like in the QAOA-circuit on a ladder architecture in~\cite{klaver2024}, where one rail is initially composed of auxiliary qubits such that after encoding the stabilizers allow for parallel operations on both rails, we also need to track the stabilizers over time in order to be able to detect stabilizer violating rotations on the one hand and on the other hand to remove trivial auxiliary information from the logical meaning of a stabilizer respecting rotation. 

In this section we introduce the combined flow tableau (already hinted at in~\cite{Gidney2021stimfaststabilizer}), containing the flow tableau and the phases of the Clifford tableau, which enables us to answer both the logical to physical question, for example with respect to stabilizers, as well as the physical to logical question, e.g., which physical rotation has which logical meaning, without the need of the expensive calculation of the inverse of a tableau. We also consider auxiliary qubits initialized in some stabilizer state and how we can easily see from the flow labels, which rotations are stabilizer violating, which act trivially and therefore, answer the physical to logical question for rotations in the presence of stabilizers coming from auxiliary qubits and in the presence of general stabilizers.

\subsection{Combined flow tableau}
Since we know that the Clifford tableau of a circuit is the inverse of its flow tableau, we could track one and calculate the inverse tableau whenever we need it, but that would be costly in run time. Of course, we could track the information for the logical to physical and the physical to logical mapping by tracking the Clifford gates of the circuit by both the flow and the Clifford tableau leading to a doubled cost in memory. But there is a more efficient way to tracking by the \emph{combined flow tableau}, which essentially takes the flow tableau plus the phases of the Clifford tableau, also giving rise to its sibling the \emph{flow labels with Clifford phases}. Alternatively, one could consider the \emph{combined Clifford tableau}, taking the Clifford tableau plus the phases of the flow tableau, but we will stick to the combined flow tableau, the corresponding results on the combined Clifford tableau being analogous.

\begin{definition}\label{def:combined_tableau]}
    The \emph{Clifford phase vector} ($\text{cp}$) of a Clifford unitary $C\in\Cl$ is the block wise permutation $\bm{\eta}=(\bm{\eta_X}\vert\bm{\eta_Z})=({\bm{\ep'_Z}}^T\vert{\bm{\ep'_X}}^T)\in \mathbb{Z}_4^{2\npq}$ of the transposed $0^\text{th}$ column ${\bm{\ep}'}^T=({\bm{\ep'_X}}^T\vert{\bm{\ep'_Z}}^T)$ of the Clifford tableau $T(C)$.

    The \emph{combined flow tableau} $F_c(C)$ with respect to the presentation $\varrho$ of $C$ is defined as
    \begin{equation}\begin{split}
        F_c(C)&=\begin{pNiceArray}{c}[first-col]
    \Block{1-1}{\text{cp}}&(1\vert\bm{\eta})\\\hline
        \Block[color=xphys, draw=white, line-width=0pt, rounded-corners]{1-1}{\begin{array}{c}
              \vdots\\[0mm]{X}_{\ipq}\\[-2mm]\vdots
        \end{array}}&\varrho(C^\ast {\color{xphys}{X}_{\ipq}})\\\hline
        \Block[color=zphys, draw=white, line-width=0pt, rounded-corners]{1-1}{
        \begin{array}{c}
              \vdots\\[0mm]{Z}_{\ipq}\\[-2mm]\vdots
        \end{array}
        }& \varrho(C^\ast {\color{zphys}{Z}_{\ipq}})\\
    \end{pNiceArray}=
    \begin{pNiceArray}{c|c|c}[first-row,first-col]
        & \ep
        & \Block[color=xlog, draw=white, line-width=0pt, rounded-corners]{1-1}{\bm{\bar{X}}} 
        & \Block[color=zlog, draw=white, line-width=0pt, rounded-corners]{1-1}{\bm{\bar{Z}}} \\
        \Block{1-1}{\text{cp}}&1&\eta_{\ipq}&\eta_{\ipqhat}\\\hline
        \Block[color=xphys, draw=white, line-width=0pt, rounded-corners]{1-1}{\begin{array}{c}
              \vdots\\[0mm]{X}_{\ipq}\\[-2mm]\vdots
        \end{array}}&\ep_{\ipq}&\bmlogicalxi[\ipq]&\bmlogicalzeta[\ipq]\\\hline
        \Block[color=zphys, draw=white, line-width=0pt, rounded-corners]{1-1}{
        \begin{array}{c}
              \vdots\\[0mm]{Z}_{\ipq}\\[-2mm]\vdots
        \end{array}
        }& \ep_{\ipqhat}&\bmlogicalxi[\ipqhat]&\bmlogicalzeta[\ipqhat]\\
    \end{pNiceArray},
    \end{split}\end{equation}
    where $\ipqhat=\ipq+\npq$ and $\bm{\eta}=(\eta_{1}\dotsc\eta_{\npq}\vert\eta_{\hat{1}}\dotsc\eta_{\hat{\npq}})$ is the Clifford phase vector of $C$.
\end{definition}

\begin{remark}
    Note that the combined flow tableau $F_c(C)$ can be used for the calculation of the pullback of a physical Pauli operator through the Clifford circuit $C$, but the Clifford phases $\eta_{\ipq}$ have to be omitted in order to get the correct result $2s\circledast\varrho(P)\cdot F(C)$ of Prop.~\ref{prop:representation_pullback_flow}.
\end{remark}

\begin{corollary}\label{cor:combined_tableau_permutation}
    Given a Clifford unitary $C\in\Cl$ and its combined flow tableau $F_c(C)$ the rows $\varrho(C_\ast \bar{g}_{\ipq})=(\ep_j'\vert\bmphysicalxi_j\vert\bmphysicalzeta_j)$ of the full Clifford tableau can be read off from the combined flow tableau column-wise after swapping the columns $\logicalpauliop{}{\ipq}{}$ with $\logicalpauliop{}{}{\ipq}$ and the rows $\physicalpauliop{}{\ipq}{}$ with $\physicalpauliop{}{}{\ipq}$. 
    
    In detail, this swapping means looking up $\bar{g}_{\ipq}$ in the gray labels at the end of the columns and then reading off this column in the order given by the gray Clifford labels at the end of the rows, i.e., first the Clifford phase in the top row, then the gray labeled $\textcolor{clgrey}{X_{\ipq}}$ rows in the bottom block followed by the gray labeled $\textcolor{clgrey}{Z_{\ipq}}$ rows in the middle block:
    \begin{equation}
        F_c(C)=    \begin{pNiceArray}{c|c|c}[first-row,first-col,last-row,last-col]
        & \ep
        & \Block[color=xlog, draw=white, line-width=0pt, rounded-corners]{1-1}{\bm{\bar{X}}} 
        & \Block[color=zlog, draw=white, line-width=0pt, rounded-corners]{1-1}{\bm{\bar{Z}}} & \\
        \Block{1-1}{\text{cp}}&1&\eta_{\ipq}&\eta_{\ipqhat}&\Block[color=clgrey]{1-1}{\eta}\\\hline
        \Block[color=xphys, draw=white, line-width=0pt, rounded-corners]{1-1}{\begin{array}{c}
              \vdots\\[0mm]{X}_{\ipq}\\[-2mm]\vdots
    \end{array}}&\ep_{\ipq}&\bmlogicalxi[\ipq]&\bmlogicalzeta[\ipq]&\Block[color=clgrey]{1-1}{\begin{array}{c}
              \vdots\\[0mm]{Z}_{\ipq}\\[-2mm]\vdots
        \end{array}}\\\hline
        \Block[color=zphys, draw=white, line-width=0pt, rounded-corners]{1-1}{
        \begin{array}{c}
              \vdots\\[0mm]{Z}_{\ipq}\\[-2mm]\vdots
        \end{array}
        }& \ep_{\ipqhat}&\bmlogicalxi[\ipqhat]&\bmlogicalzeta[\ipqhat]&\Block[color=clgrey]{1-1}{\begin{array}{c}
              \vdots\\[0mm]{X}_{\ipq}\\[-2mm]\vdots
        \end{array}}\\
        &\Block[color=clgrey]{1-1}{\text{fp}}&\Block[color=clgrey]{1-1}{\bar{Z}_{\ipq}}&\Block[color=clgrey]{1-1}{\bar{X}_{\ipq}}&
    \end{pNiceArray}
    \end{equation}
\end{corollary}
\begin{proof}
    This is just a reformulation of the fact from Prop.~\ref{prop:inverse_flow_tableau} that for the block form 
    \[
    F(C)=\left(\begin{array}{c|c}
             1&\bm{0}\\\hline
             &\\[-11pt]
             \bm{\ep}&\tilde F
        \end{array}\right)=\left(\begin{array}{c|c|c}
             1&\multicolumn{2}{c}{\bm{0}}\\\hline
             \multirow{2}{7pt}{$\bm{\ep}$}&\textcolor{xlog}{A}&\textcolor{zlog}{B}  \\\cline{2-3}
             &\textcolor{xlog}{C}&\textcolor{zlog}{D}
        \end{array}\right)
        \]
    of the flow tableau we get the block form
    \[
        T(C)=F(C^\dagger)= \left(\begin{array}{c|c|c}
             1&\multicolumn{2}{c}{\bm{0}}\\\hline
             \multirow{2}{7pt}{$\bm{\ep}'$}&&\\[-11pt]
             &\textcolor{zlog}{D}^T&\textcolor{zlog}{B}^T  \\\cline{2-3}
             &&\\[-11pt]
             &\textcolor{xlog}{C}^T&\textcolor{xlog}{A}^T 
        \end{array}\right)
    \]
    of the corresponding Clifford tableau. The transposed Clifford tableau is then
    \begin{align*}
        T(C)^T&=\left(\begin{array}{c|c|c}
             1&\dots\ep'_{\ipq}\dots&\dots\ep'_{\ipqhat}\dots\\\hline
             \multirow{2}{7pt}{$\bm{0}$}&&\\[-11pt]
             &\dots\bmphysicalxi[\ipq]^T\dots&\dots\bmphysicalxi[\ipqhat]^T\dots \\\cline{2-3}
             &&\\[-11pt]
             &\dots\bmphysicalzeta[\ipq]^T\dots&\dots\bmphysicalzeta[\ipqhat]^T\dots
        \end{array}\right)\\
        &=\left(\begin{array}{c|c|c}
             1&{\bm{\ep'_X}}^T&{\bm{\ep'_Z}}^T\\\hline
             \multirow{2}{7pt}{$\bm{0}$}&&\\[-11pt]
             &\textcolor{zlog}{D}&\textcolor{xlog}{C}  \\\cline{2-3}
             &&\\[-11pt]
             &\textcolor{zlog}{B}&\textcolor{xlog}{A}
        \end{array}\right).
    \end{align*}
    Swapping the blocks of $T(C)^T$ and replacing $\bm{\ep'}^T$ by the Clifford phase vector $\bm{\eta}$ yields
    \begin{align*}
    F_c(C)&=\left(\begin{array}{c|c|c}
             1&\bm{\eta_X}&\bm{\eta_Z}\\\hline
             \multirow{2}{7pt}{$\bm{\ep}$}&\textcolor{xlog}{A}&\textcolor{zlog}{B}  \\\cline{2-3}
             &\textcolor{xlog}{C}&\textcolor{zlog}{D} 
        \end{array}\right)\\
        &=\left(\begin{array}{c|c|c}
             1&\dots\eta_{\ipq}\dots&\dots\eta_{\ipqhat}\dots\\\hline
             \multirow{2}{7pt}{$\bm{\ep}$}&&\\[-11pt]
             &\dots\bmphysicalzeta[\ipqhat]^T\dots&\dots\bmphysicalzeta[\ipq]^T\dots \\\cline{2-3}
             &&\\[-11pt]
             &\dots\bmphysicalxi[\ipqhat]^T\dots&\dots\bmphysicalxi[\ipq]^T\dots
        \end{array}\right)
    \end{align*}
    as claimed.
\end{proof}

\begin{remark}\label{rmk:pushforward_logicalp_via_combinedtableau}
    This means that we can calculate the pushforward of a logical $\logicalp$, i.e., the product $\varrho(\logicalp)\circledast T(C)$ from Prop.~\ref{prop:representation_pushforward_Clifford}, with the help of $F_c(C)$ as a $\circledast$-``linear combination'' of the columns of $F_c(C)$ with omitted $\bm{\ep}$. But because of the fact that the operation $\circledast$ is non-commutative we have to be very careful about the correct ordering of the terms. Intuitively, it is the product $F_c(C)\circledast\varrho(\logicalp)^T$, where the flow phases are omitted and the order of the $\circledast$-sum is (front, back, middle).
\end{remark}

We already know from Prop.~\ref{prop:Flow_iteratively} and \ref{prop:Clifford_iteratively} how to update the flow tableau and the Clifford tableau when an elementary Clifford gate gets appended. The combination of the operations needed for the flow tableau part of the combined flow tableau $F_c(C)$ with the operations needed to update the Clifford phases, adapted to the permutation of block matrices as in Cor.~\ref{cor:combined_tableau_permutation} gives the following

\begin{corollary}
    Let $C'\in\{\CNOT[\icq,\itq], H_{\itq}, S_{\itq}\}$ be an elementary Clifford gate and $C\in\Cl$ be an arbitrary Clifford unitary with combined flow tableau 
    \begin{equation}
        F_c(C)=\left(\begin{array}{c|c}
              1&\bm{\eta} \\\hline
              &\\[-10pt]
              \multirow{2}{7pt}{$\bm{\ep}$}&\tilde{\bm{f_{\ipq}}}\\\cline{2-2}
              &\\[-10pt]
              & \bm{\tilde{f_{\ipqhat}}}
        \end{array}\right)
    \end{equation}
    Then the combined flow tableau $F(C'C)$ after appending $C'$ to $C$ is 
    \begin{equation}
        F_c(C'C)=F(C'C)\oplus_4 \mathrm{cp}(C'C),
    \end{equation}
    where $F(C'C)$ is the flow tableau of $C'C$ and the correction term of the Clifford phase $\mathrm{cp}(C'C)$ is $0$ if $C'=\CNOT[\icq,\itq]$,
    \begin{equation}
        \mathrm{cp}(C'C)=
            \left(\begin{array}{c|c}
                 0&2\bm{\tilde{f_{\itq}}}\odot\bm{\tilde{f_{\itqhat}}}\\\hline&\\[-11pt]
                 \bm{0^T}& 0
            \end{array}\right),
    \end{equation}
    if $C'=H_{\itq}$, where $\odot$ is the element-wise product  and 
    \begin{equation}
        \mathrm{cp}(C'C)=
                    \left(\begin{array}{c|c}
                 0&\bm{\tilde{f_{\itqhat}}}\\\hline&\\[-11pt]
                 \bm{0^T}& 0
            \end{array}\right),
    \end{equation}
    if $C'=S_{\itq}$. 
\end{corollary}

\begin{remark}\label{rmk:combined_tableau_matrix_operations}
    These three elementary Clifford gate actions can be pictured on the combined flow tableau as the following row operations. 
    \begin{description}
        \item[$\CNOT$] In the $\color{xphys}{X}$ block $\circledast$-add the $\itq^{\text{th}}$ row to the $\icq^{\text{th}}$ row and in the $\color{zphys}{Z}$ block $\circledast$-add the $\icqhat^{\text{th}}$ row to the $\itqhat^{\text{th}}$ row. All other rows stay unchanged.
        \begin{equation}
        \begin{pNiceArray}{c|c|c}[first-row,first-col,last-row,last-col]
        & \ep
        & \Block[color=xlog, draw=white, line-width=0pt, rounded-corners]{1-1}{{\bar{X}_{\ipq}}} 
        & \Block[color=zlog, draw=white, line-width=0pt, rounded-corners]{1-1}{{\bar{Z}_{\ipq}}} & \\
        \Block{1-1}{\text{cp}}&1&\eta_{\ipq}&\eta_{\ipqhat}&\Block[color=clgrey]{1-1}{\begin{array}{cc}
             \eta & \phantom{\Big\rceil}
        \end{array}}\\\hline
        \Block[color=xphys, draw=white, line-width=0pt, rounded-corners]{1-1}{{X}_{\icq}}&\ep_{\icq}&\bmlogicalxi[\icq]&\bmlogicalzeta[\icq]&\Block{2-1}{\begin{array}{cc}
             \color{clgrey}{{Z}_{\icq}}&  \,\multirow{2}{5pt}{$\Big\rceil \circledast$}\\
             \color{clgrey}{{Z}_{\itq}}& 
        \end{array}}\\
        \Block[color=xphys, draw=white, line-width=0pt, rounded-corners]{1-1}{{X}_{\itq}}&\ep_{\itq}&\bmlogicalxi[\itq]&\bmlogicalzeta[\itq]&\\\hline
        \Block[color=zphys, draw=white, line-width=0pt, rounded-corners]{1-1}{{Z}_{\icq}}& \ep_{\icqhat}&\bmlogicalxi[\icqhat]&\bmlogicalzeta[\icqhat]&\Block{2-1}{\begin{array}{cc}
             \color{clgrey}{{X}_{\icq}}&\multirow{2}{5pt}{$\Big\rfloor \circledast$}\\
             \color{clgrey}{{X}_{\itq}}& 
        \end{array}}\\
        \Block[color=zphys, draw=white, line-width=0pt, rounded-corners]{1-1}{{Z}_{\itq}}& \ep_{\itqhat}&\bmlogicalxi[\itqhat]&\bmlogicalzeta[\itqhat]&\\
        &\Block[color=clgrey]{1-1}{\text{fp}}&\Block[color=clgrey]{1-1}{\bar{Z}_{\ipq}}&\Block[color=clgrey]{1-1}{\bar{X}_{\ipq}}&
    \end{pNiceArray}
    \end{equation}
    \item[$H_{\itq}$] Swap the $\itq^{\text{th}}$ row with the $\itqhat^{\text{th}}$ row. Whenever you swap two $1$'s in an $\color{xlog}{\bar{X}_\ipq}$ or a $\color{zlog}{\bar{Z}_\ipq}$ column, $\oplus_4$-add a $2$ to the corresponding Clifford phase entry $\eta_{\ipq}$ or $\eta_{\ipqhat}$. The $1$ in the upper left corner and all other rows stay unchanged.
            \begin{equation}
        \begin{pNiceArray}{c|c|c}[first-row,first-col,last-row,last-col]
        & \ep
        & \Block[color=xlog, draw=white, line-width=0pt, rounded-corners]{1-1}{{\bar{X}_{\ipq}}} 
        & \Block[color=zlog, draw=white, line-width=0pt, rounded-corners]{1-1}{{\bar{Z}_{\ipq}}} & \\
        \Block{1-1}{\text{cp}}&1&\eta_{\ipq}&\eta_{\ipqhat}&
        \Block{3-1}{
        \begin{array}{cl}
            \textcolor{clgrey}{\!\!\eta} & 
                 \scriptstyle \vert \oplus_4 2(0\vert \bmlogicalzeta[\itq]\odot\bmlogicalzeta[\itqhat]\vert\bmlogicalxi[\itq]\odot\bmlogicalxi[\itqhat])\\
             \textcolor{clgrey}{Z_{\itq}} &  \multirow{2}{10pt}{$\big]$}\\
             \textcolor{clgrey}{X_{\itq}}& 
        \end{array}
        }\\\hline
        \Block[color=xphys, draw=white, line-width=0pt, rounded-corners]{1-1}{{X}_{\itq}}&\ep_{\itq}&\bmlogicalxi[\itq]&\bmlogicalzeta[\itq]&\\\hline
        \Block[color=zphys, draw=white, line-width=0pt, rounded-corners]{1-1}{{Z}_{\itq}}& \ep_{\itqhat}&\bmlogicalxi[\itqhat]&\bmlogicalzeta[\itqhat]&\\
        &\Block[color=clgrey]{1-1}{\text{fp}}&\Block[color=clgrey]{1-1}{\bar{Z}_{\ipq}}&\Block[color=clgrey]{1-1}{\bar{X}_{\ipq}}&
    \end{pNiceArray}
    \end{equation}
    \item[$S_{\itq}$] $\circledast$-add the $\itqhat^{\text{th}}$ row onto the $\itq^{\text{th}}$ row, $\circledast$-multiply $3$ to the $\itq^{\text{th}}$ row and $\oplus_4$-add the $\itqhat^{\text{th}}$ row to the Clifford phase row leaving the $1$ in the upper left corner unchanged.
    \begin{equation}
        \begin{pNiceArray}{c|c|c}[first-row,first-col,last-row,last-col]
        & \ep
        & \Block[color=xlog, draw=white, line-width=0pt, rounded-corners]{1-1}{{\bar{X}_{\ipq}}} 
        & \Block[color=zlog, draw=white, line-width=0pt, rounded-corners]{1-1}{{\bar{Z}_{\ipq}}} & \\
         \Block{1-1}{\text{cp}}&1&\eta_{\ipq}&\eta_{\ipqhat}&
        \Block{3-1}{
        \begin{array}{cp{10mm}c}
             \textcolor{clgrey}{\eta}&\makebox{$\scriptstyle\vert\oplus_4(0\vert\bmlogicalzeta[\itqhat]\vert\bmlogicalxi[\itqhat])$}\\
             \textcolor{clgrey}{{Z}_{\itq}}& \multirow{2}{10pt}{$\big\rceil\circledast$}&\vert\circledast3\\
             \textcolor{clgrey}{{X}_{\itq}}&&
        \end{array}
        }\\\hline
        \Block[color=xphys, draw=white, line-width=0pt, rounded-corners]{1-1}{{X}_{\itq}}&\ep_{\itq}&\bmlogicalxi[\itq]&\bmlogicalzeta[\itq]&\\\hline
        \Block[color=zphys, draw=white, line-width=0pt, rounded-corners]{1-1}{{Z}_{\itq}}& \ep_{\itqhat}&\bmlogicalxi[\itqhat]&\bmlogicalzeta[\itqhat]&\\
        &\Block[color=clgrey]{1-1}{\text{fp}}&\Block[color=clgrey]{1-1}{\bar{Z}_{\ipq}}&\Block[color=clgrey]{1-1}{\bar{X}_{\ipq}}&
    \end{pNiceArray}
    \end{equation}
    \end{description}
    The $\circledast$-adding of one row onto another is adding from the right as indicated by the notation of the operation on the right of the ``matrix''. Remember that the $\circledast$-operation is just the standard sum of the rows (i.e., $\oplus_4$-sum in the flow phase column and $\oplus_2$-sum otherwise), where the flow phase gets added an additional $2$ iff the $\color{zlog}{\bar{Z}}$ entries of the first summand anticommute with the $\color{xlog}{\bar{X}}$ entries of the second summand, i.e., iff $\bmlogicalzeta[1]\bmlogicalxi[2]^T=1$.
\end{remark}

The Clifford phase row can easily be included in the flow label formalism. We just add a row above the circuit, in which the information of the Clifford phase row at this time step is stored (omitting the leading $1$). Since the flow phase is written as $i^\ep$, we also write the Clifford phases in this form $i^\eta$, i.e.:
\begin{equation}
    \flowlabel{}{i^{\eta_1}\dots i^{\eta_{\npq}}}{i^{\eta_{\hat{1}}}\dots i^{\eta_{\hat{\npq}}}}
\end{equation}
Note that we \emph{cannot} leave out a trivial phase $i^0$ to make the notation dense, because we would then lose the information to which column the Clifford phases afterwards correspond. We use $+$ and $-$ as shorthands for the even cases $i^0$ and $i^2$, but $i$ and $i^3$ for the odd cases, since writing $-i$ would make the notation ambiguous (see the example below).

\begin{example}\label{exp:heisenberg_combined}
    Coming back to our Exp.~\ref{exp:heisenberg} in Fig.~\ref{fig:exp_heisenberg_combined_flow}, the Clifford phases in the beginning of the circuit are trivial, i.e., $\flowlabel{}{++}{++}$. The $\CNOT$-gates do not change the Clifford phases. For the first Hadamard gate the $\color{xphys}{X}$ and the $\color{zphys}{Z}$ labels are swapped, but since there are no two identical labels swapped, the Clifford phase stays unchanged. On the contrary, the second Hadamard gate swaps two red $\color{zlog}{2}$'s, therefore the corresponding Clifford phase, the second red one without brackets, gets a sign flip. The operation of the $S$-gate on the Clifford phases is multiplication with $i$ for the occurring indices in the $\color{zphys}{Z}$ label, i.e., both red Clifford phases get multiplied with $i$: ${\color{zlog}{+}}\mapsto{\color{zlog}{i}}$ and ${\color{zlog}{-}}\mapsto{\color{zlog}{-i}}$. Since ${\color{zlog}{i-i}}$ would be ambiguous, we write ${\color{zlog}{ii^3}}$. The operation on the Clifford phases of the remaining $S$-gates is analogous.

    With the combined flow labels we can answer both types of questions -- physical to logical and logical to physical. As before the logical meaning of the physical $R_X(-2\alpha)$-gate can be read off from the flow label $\flowlabel{}{12}{12}$, i.e., the physical rotation acts as a logical $-\bar{Y}_1\bar{Y}_2$-rotation. On the other hand, if we would want to know the physical meaning of the logical Pauli $\logicalpauliop{}{1}{}$ say after the first $S$-gate, we just read off the corresponding labels taking into account that $X$ and $Z$ are swapped for the Clifford meaning, i.e., we need the Clifford phase $\textcolor{zlog}{i}$ and the $\textcolor{zlog}{1}$'s on the qubits. We first need the $\textcolor{xphys}{X}$-labels, which in the Clifford meaning are below the wires, giving $\physicalpauliop{i}{1}{}$, and afterwards the $\textcolor{zphys}{Z}$-labels above the wires, giving the answer $\physicalpauliop{i}{1}{1,2}$, i.e., the pushforward of the logical Pauli operator $\logicalpauliop{}{1}{}$ to time step $5$, immediately following the first $S$-gate, is the physical Pauli operator $\physicalpauliop{i}{1}{1,2}=Y_1\physicalpauliop{}{}{2}$.
\end{example}

\subsection{Circuits with auxiliary qubits}

Including the Clifford phases is especially valuable if there are auxiliary qubits initialized in stabilizer states in the circuit, since they give rise to stabilizers, e.g., $\textcolor{zlog}{\bar{Z}}$ in case of a $\ket{0}$-initialized auxiliary. As in the above example, we can answer both types of questions physical to logical and vice versa, the latter giving the evolution of the stabilizers. But this gives us also the possibility of checking whether a logical operation is stabilizer preserving. In other words, after pulling back a physical operator to the beginning of the circuit but just after the intialization of the auxiliary qubits, can we further pull back through the initialization to an operator without auxiliary information, as will be described in the next section.

Let us start with fixing some notation. For the \emph{auxiliary labels}, i.e., those corresponding to auxiliary qubits, we choose letters instead of numbers, in order to better differentiate between the two, i.e., $\alpha\in\{a, b, \dotsc\}$. We use a $\bullet$ above the label, if this label corresponds to the logical Pauli operator $\logicalp_{\alpha}$ stabilizing the initialized qubit, and a $\times$ above the label, if this label corresponds to the logical Pauli operator destabilizing the qubit. 

The stabilizer of the state $\ket{0}$ is $\textcolor{zlog}{\bar{Z}_a}$, i.e., $\textcolor{zlog}{\bar{Z}_a}$ acts as the identity $\idmatrix$ on this qubit $\textcolor{zlog}{\bar{Z}_a}\ket{0}=\idmatrix\ket{0}=\ket{0}$, and (one of) its destabilizer(s) is $\textcolor{xlog}{\bar{X}_a}$. Therefore, the label $\flowlabel{}{}{\trivial{a}}$ (``trivial $\textcolor{zlog}{Z}$-$a$'')  corresponds to the identity operator $\idmatrix$, while the label $\flowlabel{}{\violate{a}}{}$ (``violating $\textcolor{xlog}{X}$-$a$'') corresponds to the stabilizer violating action $\textcolor{xlog}{\bar{X}_a}$ and the Clifford phase is trivial $\flowlabel{}{+}{+}$. For $\ket{1}$ we use the fact that $\ket{1}=X\ket{0}=HSSH\ket{0}$ giving the labels shown in the end of this circuit: 

\begin{equation}
\includegraphics{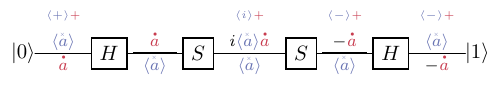}
\end{equation}

This is consistent, since for $\ket{1}$ the stabilizer is $-\textcolor{zlog}{\bar{Z}_a}$ giving the minus in the $Z$-flow phase and the minus in the $X$-Clifford phase, where we have to take into account the swapping of the meaning of $X$ and $Z$ when switching to the Clifford meaning.

The $\ket{\pm}$-cases are exactly the opposite, as shown on the left of the two circuits below, while for $\bar{Y}$-eigenstates $\ket{\pm i}$ we get the labels on the right hand sides:

\begin{equation}\label{eq:init_auxy_lables}
\includegraphics{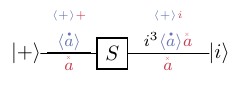}
\quad
\includegraphics{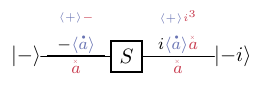}
\end{equation}
Note that the flow labels are less suited for the $Y$-basis because of the choice of presentation $\varrho$ favoring the $X$- and $Z$-basis. Nevertheless, it is of course still possible to see that a physical $Y$-gate acts trivially on the $\ket{i}$-state:

\begin{example}
    Consider the first circuit in Eq.~\eqref{eq:init_auxy_lables} and insert a physical Pauli rotation $R_P$ after the $S$-gate. Since the qubit at that time step is in the state $\ket{i}$, we know that $P=Y$ gives a stabilizer preserving rotation, while $P=X$ and $P=Z$ are stabilizer violating. In the flow labels this is seen as follows: In the cases of $R_X$ resp.\ $R_Z$ rotations we immediately see that their flow labels above resp.\ below the wire contain an $\flowlabel{}{}{\violate{a}}$, indicating that these are stabilizer violating operations. In the case of an $R_Y$ rotation, we know from the flow formalism that its logical meaning is given by the $\circledast$-sum corresponding to $\bar{Y}=i\textcolor{xlog}{\bar{X}}\textcolor{zlog}{\bar{Z}}$, i.e., $i\circledast\flowlabel{i^3}{\trivial{a}}{\violate{a}}\circledast\flowlabel{}{}{\violate{a}}=\flowlabel{}{\trivial{a}}{}$. Hence, in this case the stabilizer violating indices $\flowlabel{}{}{\violate{a}}$ cancel out, indicating that the rotation $R_Y$ is allowed, as expected. Since the logical meaning of the $R_Y$ rotation is given by $\flowlabel{}{\trivial{a}}{}$, which is just the stabilizer $\logicalp_a=\logicalpauliop{}{a}{}$ acting trivially on the $\ket{+}$ initialized qubit, we can just remove this auxiliary label $\flowlabel{}{\trivial{a}}{}$ to get the identity operator, as expected.
\end{example}

This can be generalized to a general Clifford circuit with auxiliaries
\begin{equation}\label{eq:general_Clifford_circuit_auxies}
\includegraphics{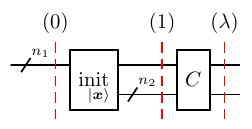}
\end{equation}
initialized in the stabilizer state $\ket{\bm{x}}$, with each entry $x_\alpha\in\{0, 1, \pm, \pm i\}$. Combining the above initialization of auxiliary labels and the workings of the combined flow labels gives the labels at each time step. What can we read off of the combined flow labels in the physical to logical direction and vice versa?

If we insert a physical single-qubit Pauli rotation $R_P$ at time step $(\itime)$ into the circuit of Eq.~\eqref{eq:general_Clifford_circuit_auxies} we can read off or calculate the \emph{flow $P$-label} by: 
\begin{description}
    \item[$P=\textcolor{xphys}{X}$] read off the $\textcolor{xphys}{X}$ label above the wire
    \item[$P=\textcolor{zphys}{Z}$] read off the $\textcolor{zphys}{Z}$ label below the wire
    \item[$P=Y$] calculate the $\circledast$-sum of the phase $i$, the $\textcolor{xphys}{X}$ label and the $\textcolor{zphys}{Z}$ label
    \item[general] write $P=\physicalpauliop{i^\eta}{ }{}^{\textcolor{xphys}{x}}\physicalpauliop{}{}{ }^{\textcolor{zphys}{z}}$ with $\textcolor{xphys}{x}, \textcolor{zphys}{z}\in\mathbb{F}_2$ and calculate the corresponding $\circledast$-sum of the phase $i^\eta$, if $\textcolor{xphys}{x}=1$ the $\textcolor{xphys}{X}$ label and if $\textcolor{zphys}{z}=1$ the $\textcolor{zphys}{Z}$ label (in this ordering)
\end{description}

The flow $P$-label $\logicalflowp$ contains auxiliary and non-auxiliary indices ($\ipq$ resp.\ $\alpha$) in general and gives the induced logical rotation $R_{\logicalp}$ with 
$\logicalp=i^\ep\textcolor{xlog}{\bm{\bar{X}}^{\bmlogicalxi}}\textcolor{zlog}{\bm{\bar{Z}}^{\bmlogicalzeta}}$ at time step (1). It is now easy to see, whether $\logicalp$ and hence the rotation $R_{\logicalp}$ pulls back through the initialization gate: For each auxiliary qubit $\alpha$, whose labels $\flowlabel{}{\alpha}{}$ or $\flowlabel{}{}{\alpha}$ occur in $\logicalflowp$, check whether the corresponding Pauli operator $\textcolor{xlog}{\bar{X}_\alpha^{\logicalxi[\alpha]}}\textcolor{zlog}{\bar{Z}_\alpha^{\logicalzeta[\alpha]}}$ commutes with the stabilizing Pauli operator $\logicalp_\alpha$. But this is exactly what is indicated by the $\violate{\alpha}$ indices, therefore, the physical rotation $R_P$ violates the stabilizer iff there are $\violate{\alpha}$ indices in its flow label $\logicalflowp$. If the flow label corresponding to $R_P$ does not contain violating indices $\violate{\alpha}$, then it can be pulled back through the initialization gate: Since all remaining auxiliary labels are trivial indices $\trivial{\alpha}$, they act as $\pm\idmatrix$ on their respective auxiliary qubit and can be removed from the label, while adjusting the phase exponent with $\oplus_42$ in the $-$ case.

We can summarize the above reasoning in the following

\begin{proposition}\label{prop:with_auxies_single_qubit}
    Let $\mathrm{init}_{\ket{\bm{x}}}$ be the gate corresponding to initializing $\npq_2$ auxiliary qubits in the stabilizer state $\ket{\bm{x}}\in\{0,1, +,-\}^{n_2}$, $C$ a Clifford circuit on $\npq=\npq_1+\npq_2$ physical qubits, $R_P$ a single-qubit Pauli rotation on the $\ipq^{\text{th}}$ qubit ($1\leq\ipq\leq\npq$) and $\logicalflowp$ its flow $P$-label.

    Then, the pullback of $R_P$ through the concatenated circuit $\bar{C}=C\cdot\mathrm{init}_{\ket{\bm{x}}}$ exists iff $\logicalxi[\alpha]=0$ for all auxiliaries with $x_{\alpha}\in\{0,1\}$ and $\logicalzeta[\alpha]=0$ for all auxiliaries with $x_{\alpha}=\pm$.

    If it exists, the pullback of $R_P$ through $\bar{C}$ is a rotation w.r.t.\ the Pauli operator corresponding to the flow label $\logicalpaulivec{\ep'}{\bmlogicalxi'}{\bmlogicalzeta'}$, where $\bmlogicalxi'$ and $\bmlogicalzeta'$ are the projections of $\bmlogicalxi$ and $\bmlogicalzeta$ to the first $\npq_{1}$ coordinates and 
    \begin{equation}\begin{split}
        \ep'=\ep&\oplus_4 2\cdot\big\vert\{\alpha\vert (\logicalzeta[\alpha]=1 \text{ and } x_{\alpha}=1)\}\big\vert\\
        &\oplus_4 2\cdot\big\vert\{\alpha\vert (\logicalxi[\alpha]=1 \text{ and } x_{\alpha}=-)\}\big\vert,
    \end{split}\end{equation}
    which is the sign correction for $\ket{1}$ and $\ket{-}$ initialized qubits.
\end{proposition}
\begin{remark}
    In case of a $\ket{\pm i}$ initialization, replace it by a $\ket{\pm}$ initialization and prepend a $S$-gate to the Clifford circuit. In the following we will skip the $\ket{\pm i}$ initialization.
\end{remark}

This generalizes to Pauli rotations $R_P$ on more than one qubit, since the above reasoning does only depend on the logical meaning of the Pauli operator $P$ at time step (1). Only the computation of this logical meaning as $\varrho(P)\circledast F(C)=2s\circledast\varrho(P)\cdot F(C)$ gets more involved.

\begin{proposition}\label{prop:multiqubitpauli_withauxies}
    In the situation of Prop.~\ref{prop:with_auxies_single_qubit} let $R_P$ be an arbitrary Pauli rotation and $\logicalflowp=2s\circledast\varrho(P)\cdot F(C)$ its flow label (see Prop.~\ref{prop:representation_pullback_flow} for the sign correction term $s$). Then the pullback of $R_P$ through $\bar{C}$ has the same properties as stated in Prop.~\ref{prop:with_auxies_single_qubit}.
\end{proposition}

\begin{remark}\label{rmk:combined_flow_tableau_auxies}
    The markers $\bullet$ and $\times$ for the trivial and the stabilizer violating indices can be transferred easily to the flow tableau:
    \begin{equation}\label{eq:combined_flow_tableau_auxies}
\begin{pNiceArray}{c|cIcc|cIcc}[first-row,first-col,last-row,last-col]
        & \phantom{\overset{\scalebox{0.5}{{$\bullet, \pm$}}}{{\bar{X}_{\ipq}}}}\ep
        & \Block[color=xlog, draw=white, line-width=0pt, rounded-corners]{1-1}{\overset{\scalebox{0.5}{\phantom{$\bullet, \pm$}}}{\bar{X}_{\ipq}}}
        & \Block[color=xlog, draw=white, line-width=0pt, rounded-corners]{1-1}{\violate{\bar{X}_{\alpha}}} 
        & \Block[color=xlog, draw=white, line-width=0pt, rounded-corners]{1-1}{\trivialpm{\bar{X}_{\alpha'}}} 
        & \Block[color=zlog, draw=white, line-width=0pt, rounded-corners]{1-1}{\overset{\scalebox{0.5}{\phantom{$\bullet, \pm$}}}{\bar{Z}_{\ipq}}}
        & \Block[color=zlog, draw=white, line-width=0pt, rounded-corners]{1-1}{\trivialpm{\bar{Z}_{\alpha}}}
        & \Block[color=zlog, draw=white, line-width=0pt, rounded-corners]{1-1}{\violate{\bar{Z}_{\alpha'}}}& \\
        \Block{1-1}{\text{cp}}&1&\eta_{\ipq}&\eta_{\alpha}&\eta_{\alpha'}&\eta_{\ipqhat}&\eta_{\hat{\alpha}}&\eta_{\hat{\alpha}'}&\Block[color=clgrey]{1-1}{\eta}\\\hline
        \Block[color=xphys, draw=white, line-width=0pt, rounded-corners]{1-1}{{X}_{i}}&\ep_{i}&\logicalxi[i,\ipq]&\logicalxi[i,\alpha]&\logicalxi[i,\alpha']
            &\logicalzeta[i,\ipqhat]&\logicalzeta[i,\hat\alpha]&\logicalzeta[i,\hat\alpha']&\Block{1-1}{\color{clgrey}{{Z}_{i}}}\\\hline
        \Block[color=zphys, draw=white, line-width=0pt, rounded-corners]{1-1}{{Z}_{i}}& \ep_{\hat\imath}
        &\logicalxi[\hat\imath,\ipq]&\logicalxi[\hat\imath,\alpha]&\logicalxi[\hat\imath,\alpha']&\logicalzeta[\hat\imath,\ipqhat]&\logicalzeta[\hat\imath,\hat{\alpha}]&\logicalzeta[\hat\imath,\hat{\alpha}']&\Block{1-1}{\color{clgrey}{{X}_{i}}}\\
        &\Block[color=clgrey]{1-1}{\phantom{\overset{\scalebox{0.5}{\phantom{$\bullet, \pm$}}}{\phantom{\bar{X}_{\ipq}}}}\text{fp}}&\Block[color=clgrey]{1-1}{\overset{\scalebox{0.5}{\phantom{$\bullet, \pm$}}}{\bar{Z}_{\ipq}}}&\Block[color=clgrey]{1-1}{\trivialpm{\bar{Z}_{\alpha}}}&\Block[color=clgrey]{1-1}{\violate{\bar{Z}_{\alpha'}}}&\Block[color=clgrey]{1-1}{\overset{\scalebox{0.5}{\phantom{$\bullet, \pm$}}}{\bar{X}_{\ipq}}}&\Block[color=clgrey]{1-1}{\violate{\bar{X}_{\alpha}}}&\Block[color=clgrey]{1-1}{\trivialpm{\bar{X}_{\alpha'}}}&
    \end{pNiceArray},
    \end{equation}
    where the columns with index $\ipq$ correspond to non-auxiliary qubits, the columns with index $\alpha$ to $\ket{0}$ or $\ket{1}$ initialized auxiliary qubits and those with index $\alpha'$ correspond to $\ket{\pm}$ initialized auxiliary qubits. The markers $\bullet$ get complemented by a $\pm$ to indicate, whether the action of this operator on the initialized auxiliary is $+\idmatrix$ or $-\idmatrix$.

    Checking whether a physical multi-qubit Pauli rotation $R_P$ is stabilizer violating then reduces to checking for each column with a $\times$, whether the sum of all entries without the phase is $0$. But caution is needed for the calculation of the pullback through the concatenation $\bar{C}$ of the initialization and $C$: If we set all entries of the columns marked with $\bullet$ to $0$, with the argument that these act trivially anyways, we would lose their information for the following calculations. We still need to compute the flow ``vector'' $\logicalflowp=2s\circledast\varrho(P)\cdot F(C)$, where disregarding the label $\trivial{\alpha}$ might give rise to a change of commutativity properties and therefore sign errors. We would also lose the ability to compensate the sign if qubits were initialized in a $-1$-eigenstate of a Pauli $X$ or $Z$ operator. Furthermore, the iterative updates of $F_c(C)$ would become erroneous, since we would ignore labels, which might give rise to sign errors in the phases.

    Of course, if the phases are not needed or definitely trivial, e.g., in case of a $\CNOT$-only circuit, this is not an issue and the $\bullet$ columns may be disregarded, while the $\times$ columns will be needed to check for violating operations. 
\end{remark}

\begin{remark}\label{rmk:run_time_allowed_logical_timestep}
    Of course, if we are only interested in checking, whether $R_P$ is an allowed operation, we do not need the phases of the tableau. We only need to $\oplus_2$-add all rows of the phaseless tableau corresponding to entries of $1$ in the physical Pauli $P$ giving a sum of $h$ summands of length $2\npq$, where $h=h(P)$ is the Hamming weight of $P$. In this sum we have to check, whether the entries corresponding to the $\violate{\alpha}$ labels are all zero, giving $\npq_2$ checks. Therefore, it is enough to calculate a $\oplus_2$-sum of $h$ summands of length $\npq_2$, giving a run time of $\mathcal{O}(h\npq_2)$. If only single-qubit rotations are used, for instance because of the hardware restrictions, $h$ is bounded and the run time reduces to $\mathcal{O}(\npq_2)$.

    On the other hand, if we are also interested in the logical meaning of the physical rotation including the phase, we have to take care of the sign corrections. This gives a run time of $\mathcal{O}(\max(h\npq, h^\omega))$ for the $\circledast$-sum of $h$ summands of length $2\npq+1$. The $\npq_2$ checks for stabilizer violating entries $\violate{\alpha}$, the $\npq_2$ entries $\trivial{\alpha}$, which have to be set to $0$ since they correspond to a trivial action, as well as the sign correction for auxiliaries initialized in $-1$-eigenstates of $\bar{X}$ or $\bar{Z}$ are only linear in $\npq_2$ leaving the overall run time at $\mathcal{O}(\max(h\npq, h^\omega))$, e.g., if $h=\mathcal{O}(\npq)$ we get $\mathcal{O}(\npq^\omega)$, while for bounded Hamming weights $h$ we get $\mathcal{O}(\npq)$. 
\end{remark}

Let us look at a small example with an auxiliary qubit.

\begin{example}\label{exp:heisenberg_with_auxy}
We can introduce an auxiliary qubit in Exp.~\ref{exp:heisenberg} and get the alternative implementation in Fig.~\ref{fig:exp_heisenberg_extended_flow_with_auxy} (see Fig.~\ref{fig:auxiliary_qubits}(b) in the main text). The goal is to find Clifford gates producing the desired labels $\flowlabel{}{12}{}$, $\flowlabel{-}{12}{12}$ and $\flowlabel{}{}{12}$ either as $\physicalpauliop{}{ }{}$, $\physicalpauliop{}{}{ }$ or $Y$ label at some qubit at some time step. Here this is even possible at the same time step, but in general the rotations do not have to be in parallel. The fact that in this example the physical and the logical rotations are with respect to the same type of Pauli operator, e.g., $R_{\bar Y_1\bar Y_2}$ is realized by $R_Y$, is merely a coincidence, in general, as seen in Fig.~\ref{fig:exp_heisenberg_combined_flow} this need not be the case. Here, introducing the auxiliary increases the 2-qubit gate depth of the circuit, but in more complex use-cases the combined flow formalism reveals usages of auxiliary qubits which lead to reduced circuit depth, as for example seen in the QAOA on a ladder \cite{klaver2024}.
\begin{figure*}
    \centering
\includegraphics{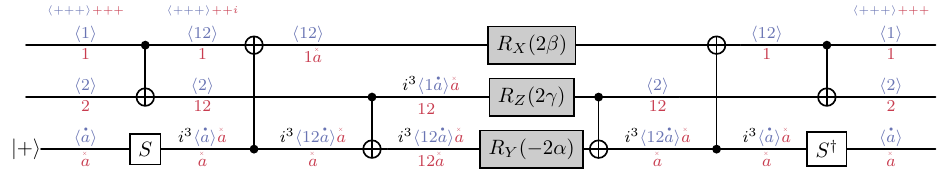}
   \caption{Adding an auxiliary qubit to the 1D-Heisenberg example: The physical rotation $R_X$ is allowed, since there is no $\violate{a}$ in its $X$ label. Its logical meaning is $R_{\bar{X}_1\bar{X}_2}(2\beta)$ since the $X$ label is $\textcolor{xlog}{\yourfavouriteleftparanthesis 12\yourfavouriterightparanthesis}$. Analogously, $R_Z$ is allowed and corresponds to a logical $R_{\bar{Z}_1\bar{Z}_2}(2\gamma)$ rotation. For the physical $R_Y$ rotation we have to calculate the $Y$ label as $i\circledast{i^3}\textcolor{xlog}{\yourfavouriteleftparanthesis12\trivial{a}\yourfavouriterightparanthesis}\textcolor{zlog}{\violate{a}}\circledast\textcolor{zlog}{12\violate{a}}=\textcolor{xlog}{\yourfavouriteleftparanthesis 12\trivial{a}\yourfavouriterightparanthesis}\textcolor{zlog}{12}\mapsto\textcolor{xlog}{\yourfavouriteleftparanthesis 12\yourfavouriterightparanthesis}\textcolor{zlog}{12}$. Since it does not contain an $\violate{a}$, the rotation is allowed and its logical meaning is $R_{-\bar{Y}_1\bar{Y}_2}(-2\alpha)=R_{\bar{Y}_1\bar{Y}_2}(2\alpha)$ as desired.
   } \label{fig:exp_heisenberg_extended_flow_with_auxy}
\end{figure*}
\end{example}

There is another way of checking, whether the physical Pauli rotation $R_P$ is stabilizer violating, which uses the logical to physical direction, i.e., the Clifford meaning of the combined flow labels. We can compute the pushforward of the stabilizers corresponding to the initialization state $\ket{\bm{x}}$ to the time step of $R_P$ and then check, whether $P$ commutes with the pushforwards of all the stabilizers. In order to get the Clifford meaning out of the combined flow labels, we have to swap $X$ and $Z$, i.e., use the gray Clifford meanings in Eq.~\eqref{eq:combined_flow_tableau_auxies}. Note that the markings as $\bullet$ and $\times$ do also swap their places. This is the case, since the property of being a (de-)stabilizer of the initial state is a property of the Pauli operator and therefore, stays with the operator, not with the column. The swapping of $X$ and $Z$ for the Clifford meaning is just a way to easily deduce the entries of the Clifford tableau from the flow tableau and does not change the action of the Pauli operator on its auxiliary qubit. 

As proved in the section on the combined flow tableau, we can read off the pushforward $C_\ast\logicalp_{\alpha}$ of a stabilizer $\logicalp_{\alpha}$ to some time step from the combined flow tableau $F_c(C)$ by finding the stabilizer in the gray Clifford meanings in the footer row, read off this column in the order $\textcolor{clgrey}{\eta}, \textcolor{clgrey}{\bm{X}}, \textcolor{clgrey}{\bm{Z}}$ (upper, lower, middle), as seen in the initialization of the Clifford phases in Eq.~\eqref{eq:init_auxy_lables}. In the flow labels that means finding the corresponding $\violate{\alpha}$ label ($\flowlabel{}{\violate{\alpha}}{}$ resp.\ $\flowlabel{}{}{\violate{\alpha}}$) in the time step, read off the corresponding Clifford phase, take the physical $\physicalpauliop{}{\ipq}{}$ operators for every qubit that has the $\violate{\alpha}$ label below the wire and then take the physical $\physicalpauliop{}{}{\ipq}$ operators for every qubit that has the $\violate{\alpha}$ label above the wire. The rotation $R_P$ is an allowed operation iff $P$ commutes with these physical meanings of $\logicalp_{\alpha}$ for every auxiliary qubit $\alpha$. This gives rise to exactly the same $\npq_2$ checks on $\oplus_2$-sums with $h=h(P)$ summands as explained in Rmk.~\ref{rmk:run_time_allowed_logical_timestep}, still giving a run time of $\mathcal{O}(h\npq_2)$. Since the auxiliary stabilizer $\logicalp_{\alpha}\in\{\pm\bar{X}_\alpha, \pm\bar{Z}_\alpha\}$ has Hamming weight $1$, its physical meaning can be simply read off of the appropriate column of the tableau (upper, lower, middle). 

This can be generalized to additional stabilizers, which were already present before the $\mathrm{init}$-gate by using Prop.~\ref{rmk:pushforward_logicalp_via_combinedtableau}. In the flow labels this means the following: 
\begin{itemize}
    \item Write the stabilizer as $\logicalp=i^\ep\textcolor{xlog}{\bm{\bar{X}}^{\bm{\xi}}}\textcolor{zlog}{\bm{\bar{Z}}^{\bm{\zeta}}}$.
    \item For every $\logicalpauliop{}{\ipq}{}$ appearing in $\logicalp$ find the corresponding \emph{flow} labels $\flowlabel{}{}{\ipq}$ at the time step of interest and read off its physical meaning as the product of the Clifford phase in column $\flowlabel{}{}{\ipq}$, an $\physicalpauliop{}{\ipq'}{}$ if there is a $\flowlabel{}{}{\ipq}$ below the ${j'}^{\text{th}}$ wire and a $\physicalpauliop{}{}{\ipq'}$ if there is a $\flowlabel{}{}{\ipq}$ above the wire.
    \item For every $\logicalpauliop{}{}{\ipq}$ appearing in $\logicalp$ find the corresponding \emph{flow} labels $\flowlabel{}{\ipq}{}$ and read off its physical meaning as above (Clifford phase, below wire, above wire).
    \item Calculate the product of the phase $i^\ep$ of $\logicalp$, the physical meanings of the $\logicalpauliop{}{\ipq}{}$ and those of the $\logicalpauliop{}{}{\ipq}$.
\end{itemize} 

Let us conclude this section by comparing the two different ways of checking whether an arbitrary stabilizer $\logicalp\in\bar\Pauli^{(1)}$ at time step $(1)$ commutes with an arbitrary multi-qubit physical Pauli operator $Q\in\Pauli^{(\itime)}$ at time step $(\itime)$ for the combined flow tableau $F_c(C)$:

Write the presentation $\varrho(Q)$ of the physical Pauli operator $Q$ as a column vector in front of the tableau, its $1$-entries indicating which rows of the tableau are of interest. Write the $\logicalpauliop{}{ }{}$-$\logicalpauliop{}{}{ }$-swapped presentation $\varrho(\logicalp)\cdot\Omega$ of the stabilizer $\logicalp$ as a row vector below the tableau, its $1$-entries indicating which columns of the tableau are of interest. We have to use the Clifford meaning of the tableau, hence the swapping of the middle and last part of $\varrho(\logicalp)$ with the help of $\Omega=\begin{psmallmatrix}1&0&0\\0&0&\idmatrix\\0&\idmatrix&0\end{psmallmatrix}$.
If we choose to use the time step $(1)$ for checking whether the physical $Q$ is compatible with the stabilizer $\logicalp$, we have to calculate the pullback $C^\ast Q$ as the $\circledast$-sum of the rows filtered out by $\varrho(Q)^T$, giving the row ``vector'' $(\ep'\vert\bmlogicalxi'\vert\bmlogicalzeta')$, whose symplectic product with $\varrho(\logicalp)$ can be easily calculated as the inner product with the already swapped row ``vector'' $\varrho(\logicalp)\cdot\Omega$. This corresponds to first calculating the $\oplus_2$-sums of all columns of interest filtered out by the swapped $\varrho(\logicalp)\Omega$, i.e., summing in the vertical direction, and afterwards calculating the $\oplus_2$-sum of the results in the horizontal direction. Alternatively, if we choose to use the time step $(\itime)$ for checking the commutativity, we have to calculate the pushforward $C_\ast\logicalp$ as the $\circledast$-sum of the columns filtered out by the swapped $\varrho(\logicalp)\Omega$, giving the swapped column ``vector'' $(\eta'\vert\bmphysicalzeta'\vert\bmphysicalxi')^T$, whose inner product with $\varrho(Q)$ gives the symplectic product of interest. This corresponds to first calculating the $\oplus_2$-sum of all rows of interest, i.e., summing in the horizontal direction, and afterwards calculating the $\oplus_2$-sum of the results in the vertical direction. Since these two different checks are just calculations in the finite field $\mathbb{F}_2$ the result is just the sum of all entries of $F_c(C)$ at the intersections of the interesting rows and columns and the physical Pauli $Q$ is compatible with the logical Pauli $\logicalp$ iff this sum is $0$. In the labels this corresponds to counting the corresponding labels of interest, i.e., we have to consider all qubit labels at time step $(\itime)$ that are present in the physical Pauli $Q$ and count the number of all interesting labels as indicated by the $X$/$Z$-swapped labels appearing in $\logicalp$. The physical $Q$ is compatible with the logical $\logicalp$ iff this number is even. 

If only the compatibility is to be checked, the two ways are equivalent. But each way does give additional information corresponding to the two questions posed at the beginning of these appendices, i.e., the logical to physical question and the physical to logical question. While in principle the Clifford phases, if needed, could always be computed from the flow phases via the inverse tableau, the additional memory requirements and the additional computations in each update step are very low compared to the computation of the inverse phases. Therefore, the combined flow labels and the combined flow tableau can answer questions in the physical to logical as well as in the logical to physical direction, without the need of calculating the expensive inverse of a tableau. 

\input{bibliography.bbl}

\end{document}

%% file: bibliography.bbl
%